\documentclass{article}

\usepackage{amsfonts}
\usepackage{amsmath}
\usepackage{graphicx}
\usepackage{epstopdf}
\usepackage{braket}
\usepackage{Preamble}

\usepackage{enumitem}
\usepackage{hyperref}

\usepackage{mathrsfs}

\usepackage[margin=1in]{geometry}
\usepackage[sorting=none]{biblatex}
\bibliography{biblio}

\newcommand{\End}{\mathrm{End}}

\renewcommand{\supp}{\mathrm{supp}\,}

\newcommand{\Span}{\mathrm{span}}

\newcommand{\Hom}{\mathrm{Hom}}

\newcommand{\tr}{\mathrm{tr}}

\newcommand{\diam}{\mathrm{diam}}

\newcommand{\vN}{\mathrm{vN}}
\newcommand{\id}{\mathrm{id}}
\newcommand{\hull}{\mathrm{hull}}

\newcommand{\sctr}{\mathrm{sctr}}

\newcommand{\nsub}[1]{\stackrel{#1}{\subset}}
\newcommand{\nin}[1]{\stackrel{#1}{\in}}

\newcommand{\ind}{\mathrm{ind}}

\newcommand{\bnu}{\boldsymbol \nu}

\newcommand{\boe}{\boldsymbol e}
\newcommand{\reg}{\mathrm{reg}}
\newcommand{\rhoreg}{\rho_{\reg}}
\newcommand{\fer}{\mathrm{fer}}
\newcommand{\Maj}{\mathrm{Maj}}

\newcommand{\ssubset}{\subset\joinrel\subset}

\newcommand{\loc}{\mathrm{loc}}

\newcommand{\gotimes}{\hat \otimes}
\newcommand{\biggotimes}{\widehat\bigotimes}

\renewcommand{\note}[1]{\null}
\renewcommand{\margin}[1]{\null}

\title{Classification of equivariant quasi-local automorphisms on quantum chains}
\author{Alex Bols \\
\normalsize
QMATH, Department of Mathematical Sciences, University of Copenhagen \\
\normalsize
Universitetsparken 5, 2100 Copenhagen, Denmark\\
\normalsize
e-mail : \texttt{alex-b@math.ku.dk}\\
}

\date{\today}

\begin{document}

\maketitle

\abstract{We classify automorphisms on quantum chains, allowing both spin and fermionic degrees of freedom, that are moreover equivariant with respect to a local symmetry action of a finite symmetry group $G$. The classification is up to equivalence through strongly continuous deformation and stacking with decoupled auxiliary automorphisms. We find that the equivalence classes are uniquely labeled by an index taking values in $\Q \cup \sqrt{2} \Q \times \Hom(G, \Z_2) \times H^2(G, U(1))$. We discuss te relation of this index to the index of one-dimensional symmetry protected topological phases on spin chains, which takes values in $H^2(G, U(1))$.}

\tableofcontents

\section{Introduction}

Dynamics in many-body quantum mechanics is usually local. The most common scenario is where the dynamics is generated by a suitably local Hamiltonian, then the time evolution satisfies Lieb-Robinson bounds which provide an approximate upper bound on the velocity at which signals can propagate through the system. There are also natural local dynamics that are not generated by a local Hamiltonian. Consider for example the edge of a chiral Floquet insulator that is many-body localized in the bulk \cite{FMPPV2016}. In this setup the stroboscopic evolution can be separated into a bulk and an edge component, and the edge component is a one-dimensional dynamics that moves excitations predominantly in one rather than the other direction along the edge, \ie it is chiral.

The chiral behaviour of the one-dimensional edge dynamics indicates that the two-dimensional system under consideration is topologically non-trivial in some sense. The chirality of the effective edge dynamics should therefore distinguish it from a `trivial' non-chiral dynamics such as a Hamiltonian evolution. The first important step in distinguishing chiral from non-chiral one-dimensional dynamics was taken by \cite{GNVW2012} where an index is associated to any strictly local automorphism of a one-dimensional spin chain. (See also \cite{CPSV2017} which achieves the same result using the language of matrix product operators). Such strictly local automorphisms are called quantum cellular automata (QCA) and the index of \cite{GNVW2012} is now called the GNVW index after the authors. A non-trivial value of the GNVW index indicates chiral behaviour of the QCA. Moreover, the GNVW index is stable under strongly continuous paths of QCAs and therefore distinguishes chiral from non-chiral in a topological sense. It is moreover shown in \cite{GNVW2012} that the GNVW index is a complete invariant in the sense that two QCAs are connected by a stongly continuous path if and only if they have the same GNVW index. The classification of \cite{GNVW2012} applies only to spin systems, the extension to quantum chains allowing both spin and fermionic degrees of freedom was obtained in \cite{FPPV2019}.

We see that the classification of local dynamics is related to the classification of Floquet topological phases. More generally, understanding the topological classes of local dynamics yields insight in the classification of topological phases of many-body systems \cite{Hastings2013}

As discussed above, generic dynamics considered in many-body physics are only approximately local and so the theory of \cite{GNVW2012, FPPV2019} does not strictly speaking apply. In the recent work \cite{RWW2020} all the results of \cite{GNVW2012} are appropriately generalized to apply to quasi-local automorphisms of spin chains. They show that the GNVW index extends to such automorphisms and is still complete in the sense that two quasi-local automorphisms can be connected by a time evolution generated by a local Hamiltonian if and only if they have the same index.

Often physical systems enjoy an invariance under a local action of some symmetry group $G$. Imposing such a symmetry can enrich the topologically distinct edge dynamics in the examples above and therefore enrich for example the possible topological phases of two-dimensional Floquet insulators \cite{ElseNayak2016}. The symmetry of the Floquet insulator is inherited by the stroboscopic edge dynamics, and the classification of symmetric local dynamics is finer than that of local dynamics without any imposed symmetry. We will call such symmetric dynamics `equivariant'.

In the case of equivariant QCAs on spin chains an index taking values in the second cohomology group $H^2(G, U(1))$ was identified in \cite{GSSC2020} (using the language of matrix-product unitaries). This index represents another obstruction to connect two equivariant QCAs through a path of equivariant QCAs in addition to the GNVW index. As shown in \cite{GSSC2020}, the cohomology index and GNVW index together form complete invariants of equivariant QCAs provided one is allowed to `stack' auxiliary one-dimensional systems in the construction of paths of QCAs connecting two equivariant QCAs with the same indices.

As with the GNVW index, it is a physically relevant problem to extend the classification of equivariant QCAs to equivariant quasi-local automorphisms of quantum chains allowing both spin and fermionic degrees of freedom. This is achieved in the present article.

In Section \ref{sec:setup and main results} we describe what is meant by a quantum chain with an on-site symmetry action and by even equivariant quasi-local automorphisms acting on such quantum chains. We define the notions of stable $G$-equivalence of such automorphisms. At the end of the section we state the main classification result, namely that the equivalence classes are uniquely labelled by an index taking values in $\Q \cup \sqrt{2} \Q \times \Hom(G, \Z_2) \times H^2(G, U(1))$. After the brief Section \ref{sec:operator algebraic preliminaries} dealing with operator algebraic preliminaries, we move on to Section \ref{sec:classification of QCAs} which is devoted to the classification of even equivariant QCAs. The final Section \ref{sec:index theory for quasi-local} is devoted to extending the classification result for even equivariant QCAs to general quasi-local automorphisms. This section is a mostly trivial extension of the analogous work in \cite{RWW2020} to the graded equivariant setting. Finally, the first appendix is concerned with graded von Neumann algebras. These are play an important role in the second appendix, which is devoted to proving a graded equivariant analogue of \cite{RWW2020}'s Theorem 2.6 on near inclusions of subalgebras. This is the key tool to extend the classification of even equivariant QCAs to quasi-local automorphisms.

\textbf{Acknowledgements}

The author was supported by VILLIUM FONDEN through the QMATH Centre of Excellence (grant no. 10059).

\section{Setup and main results} \label{sec:setup and main results}

\subsection{Quantum chains}

To each site $n \in \Z_n$ we associate a matrix algebra $\caA_n \simeq \caM_{d_n \times d_n}$ equipped with a grade involution $\theta_n$ making each $\caA_n$ into a superalgebra. The local dimensions $d_n$ are assumed to be uniformly bounded in $n$. For any finite $\Lambda \subset \Z$ we define $\caA_{\Lambda} = \biggotimes_{n \in \Lambda} \caA_n$ where $\gotimes$ is the graded tensor product. If we have an inclusion of finte subsets $\Lambda \subset \Lambda' \ssubset \Z$ we have a natural injection of $\caA_{\Lambda}$ into $\caA_{\Lambda'}$. This gives a directed system of superalgebras over the directed set of finite subsets of $\Z$. The direct limit is the algebra of local observables
\begin{equation}
	\caA_{\loc} := \varinjlim \caA_{\Lambda}.
\end{equation}
The quasi-local algebra is the norm-closure of the algebra of local observables
\begin{equation}
	\caA := \overline{\caA_{\loc}}^{\norm{\cdot}},
\end{equation}
which is a $C^*$-algebra. We call $\caA$ the quantum chain given by $( \{ \caA_n \}_{n \in \Z}, \{ \theta_n \}_{n \in \Z} )$ with local dimensions $\{  d_n \}_{n \in \Z}$.

For any infinite $\Gamma \subset \Z$ we let $\caA_{\Gamma}$ be the norm-closure of $\cup_{\Lambda \ssubset \Gamma} \caA_{\Lambda}$, in particular $\caA = \caA_{\Z}$. If $x \in \caA_{\Gamma}$ we say that $x$ is supported on $\Gamma$. The support of an operator $x \in \caA$ is the smallest set $\Gamma \subset \Z$ such that $x \in \caA_{\Gamma}$, it is denoted by $\supp(x)$.

We denote by $\theta$ the unique involutive automorphism of $\caA$ such that $\theta(x_n) = \theta_n(x_n)$. We call $\theta$ the grade involution of the quantum chain, it makes $\caA$ into a super $C^*$-algebra.

A *-automorphism $\al$ of $\caA$ is called even if $\al \circ \theta = \theta \circ \al$.

\subsection{Local symmetry}

We fix a finite group $G$ and equip any given quantum chain $\caA$ with even automorphisms $\{ \rho^g \}_{g \in G}$ such that $\rho^g(\caA_n) = \caA_n$ for all $n \in \Z$ and such that the restriction to each local algebra $\caA_n$ is given by $\rho^g|_{\caA_n} = \Ad_{u_n(g)}$ for some projective representations $g \mapsto u_n(g) \in U(\caA_n)$ of $G$. We call the map $g \mapsto \rho^g$ the symmetry of the quantum chain $\caA$, and denote it by $\rho$. If $\rho^g = \id$ for all $g$ we call $\rho$ \emph{trivial}. It will be assumed throughout that all quantum chains we consider come with a symmetry (possibly trivial) and we will not state this explicitly each time a quantum chain is introduced.

\subsection{Stacking, coarse graining and equivalence of quantum chains}

Given quantum chains $\caA^{(1)}$, $\caA^{(2)}$ equipped with automorphic $G$-actions $\rho_1$ and $\rho_2$ respectively, we can construct a new quantum chain $\caA$ by stacking. As an algebra, $\caA = \caA^{(1)} \gotimes \caA^{(2)}$, and the automorophic $G$-action on $\caA$ is given by $\rho = \rho_1 \gotimes \rho_2$. The on-site algebras of the quantum chain $\caA$ are $\caA_n = \caA_n^{(1)} \gotimes \caA_n^{(2)}$ for all $n \in \Z$ so $\rho|_{\caA_n} = \rho_1|_{\caA_n^{(1)}} \gotimes \rho_2|_{\caA_n^{(2)}}$.

From a given quantum chain $\caA$ with on-site algebras $\caA_n$ we can produce a new chain by coarse graining. A coarse graining is obtained by grouping sites together into new sites. For a partition $\{I_i\}_{i \in \Z}$ of $\Z$ into intervals $I_i$ of uniformly bounded length we get a caorse grained quantum chain $\widetilde \caA$ which as an algebra is identical to $\caA$, but has on-site algebras $\widetilde \caA_i = \gotimes_{n \in I_i} \caA_n$ for all $i \in \Z$.

We say two quantum chains  $\caA^{(1)}$, $\caA^{(2)}$ equipped with automorphic $G$-actions $\rho_1$ and $\rho_2$ respectively are equivalent if there are coarse grainings $\widetilde \caA^{(1)}$ and $\widetilde \caA^{(2)}$ such that for each $i \in \Z$ there is an isomorphism of graded algebras $\io_i : \widetilde \caA^{(1)}_{i} \rightarrow \caA^{(2)}_i$ with 
\begin{equation}
	\io_i \circ \rho_1|_{\widetilde \caA^{(1)}_i} = \rho_2|_{\widetilde \caA^{(2)}_i} \circ \io_i.
\end{equation}

In the language of Section \ref{sec:CS G-systems}, this says that after coarse graining, the on-site algebras of the chains are isomorphic as central simple $G$-systems.

\subsection{Quantum Cellular Automata and quasi-local automorphisms}

We will be interested in automorphisms that respect the symmetry $G$ and the grading:
\begin{definition}
	An automorphism $\al$ of $\caA$ is said to be even if $\al \circ \theta = \theta \circ \al$. It is said to be ($G$-)equivariant if $\rho^g \circ \al = \al \circ \rho^g$ for all $g \in G$.
\end{definition}

We now define some notions of locality for automorphisms.

For $X \subset \Gamma$ we denote by
\begin{equation}
	X^{(r)} := \{ n \in \Z \; : \; \dist(n, X) \leq r  \}
\end{equation}
the `$r$-fattening' of $X \subset \Z$.

\begin{definition} \label{def:QCA}
	A \emph{quantum cellular automaton} (QCA) of range $R$ on a quantum chain $\caA$ is an even automorphism $\al$ of $\caA$ such that for any operator $x$ supported in a subset $X \subset \Z$ we have that $\al(x)$ is supported in $X^{(R)}$.

	We say the QCA is \emph{equivariant} if $\al \circ \rho^{g} = \rho^{g} \circ \al$ for all $g \in G$.
\end{definition}

A QCA of range 1 is said to be a nearest-neighbour QCA. Any QCA can be brought in nearest-neighbour form by coarse graining the quantum chain, \ie grouping blocks of neighbouring sites into new coarse-grained sites.

A crisper way of saying that an automorphism is a QCA of range $r$ is that $\al \big( \caA_X \big) \subset \caA_{X^{(r)}}$ for all $X \subset \Z$. A convenient way of generalizing from the strictly local QCAs to a more robust notion `quasi-locality' is to allow the inclusion of $\al \big(  \caA_X \big)$ in $\caA_{X^{(r)}}$ to be only approximate, but to improve as $r$ gets bigger. To this end we use the notion of near inclusion:
\begin{definition} \label{def:near inclusion}
	Let $\caA$ and $\caB$ be two $C^*$-subalgebras of a $C^*$-algebra $\caC$. We say $x \in \caC$ is $\ep$-nearly in $\caA$, and write $x \nin{\ep} \caA$, if there is an element $a \in \caA$ such that $\norm{x - a} \leq \ep \norm{x}$.

	We say $\caB$ is $\ep$-nearly included in $\caA$, and write $\caB \nsub{\ep} \caA$ if $b \nin{\ep} \caA$ for all $b \in \caB$.
\end{definition}

\begin{definition}
	An automorphism $\al$ of $\caA$ is said to be quasi-local with $f$-tails if for each (possibly infinite) interval $I \subset \Z$ we have $\al \big( \caA_I \big) \nsub{f(r)} \caA_{I^{(r)}}$ for all $r = 1, 2, 3, \cdots$ and where $f : \N \rightarrow \R^+$ is a non-increasing function such that $\lim_{r \uparrow \infty} f(r) = 0$.
\end{definition}

Note that a QCA of range $R$ is a quasi-local automorphism with $f$-tails where $f(r) = 0$ for $r \geq R$.

\subsection{stable \texorpdfstring{$G$}{G}-equivalence}

We want to classify even equivariant QCAs and quasi-local automorphisms up to strongly continuous paths of even equivariant QCAs/quasi-local automorphisms, and allowing `stacking' with ancillary symmetric quantum chains.

We'll see below that, under certain conditions, we can connect two even equivariant quasi-local automorphisms with $f$-tails by a path of even equivariant quasi-local automorphisms with $\caO(f(Cr))$-tails for some universal constant $C$. It is therefore interesting to consider collections of quasi-local automorphisms that are closed under the stretching of $f(r)$-tails to $\caO(f(Cr))$-tails.

\begin{definition} \label{def:robust class}
	A condition $Q$ on tail bounds $f$ of quasi-local automorphisms is called a robust tail property if, whenever $f$ satisfies $Q$, then so does any $g(r) = \caO(f(Cr))$ with $C$ the universal constant introduced in Proposition \ref{prop:decoupling trivial quasi-local automorphisms}.

	We denote by $\scrA_Q$ the collection of all even equivariant quasi-local automorphisms with $f$-tails such that $f$ satisfies property $Q$. Such a collection $\scrA_Q$ is called a robust collection of even equivariant quasi-local automorphisms determined by $Q$.
\end{definition}

Some examples are in order.
\begin{examples} \label{example:robust collections}
	\begin{itemize}
		\item If $Q$ says nothing more than what is required by the definition of a quasi-local automorphism, \ie that $\lim_{\uparrow \infty} f(r) = 0$, then $Q$ is robust and $\scrA_Q =: \scrA$ is the collection of all even equivariant quasi-local automorphisms.

		\item If $Q$ says that $f(r)$ must be eventually equal to zero, then $Q$ is robust and $\scrA_Q = \scrA_{QCA}$ is precisely the collection of all even equivaraint QCAs.
		\item If $Q$ demands that $f(r) = \caO(r^{-a})$, \ie that there exists some $c$ such that for $r$ sufficiently large, $f(r) \leq c r^{-a}$, then $Q$ is robust and we call $\scrA_Q$ the collection of even equivariant quasi-local automorphisms with $\caO(r^{-a})$ tails.
		\item If $Q$ demands that $f(r) = \caO(r^{-\infty})$, \ie that for each $n \in \N$ there exists some $C_n < \infty$ tuch that $f(r) \leq C_n r^{-n}$ for $r$ large enough, then $Q$ is robust and we call $\scrA_Q$ the collection of even equivariant quasi-local automorphisms with superpolynomial tails.
		\item If $Q$ demands that $f(r) = \caO(e^{-\gamma r})$ for some $\gamma > 0$ then $Q$ is robust and we call $\scrA_Q$ the colleciton of even equivariant quasi-local automorphisms with exponential tails.
	\end{itemize}
\end{examples}

Let now $Q$ be a robust tail property. We define equivalence relations on the corresponding collection $scrA_Q$.

\begin{definition} \label{def:G-equivalence}
	Two even equivariant qusai-local automorphisms $\al, \beta \in \scrA_P$ on the same quantum chain $\caA$ are called $G$-equivalent in $\scrA_Q$ if there is a strongly continuous path of even equivariant quasi-local automorphisms in $\scrA_Q$ interpolating between them.
\end{definition}

This notion of equivalence does not allow us to compare automorphisms defined on different quantum chains. It turns out however that any two quantum chains can be made isomorphic to each other by `stacking' them with auxiliary quantum chains. We can then compare two given automorphisms on the original chains by first extending them with a `trivial' action on the auxiliary chains.

The most obvious choice for a trivial action is perhaps $\id$, but this still too restrictive to get a useful classification result. Indeed, consider
\begin{definition} 
	An even equivariant QCA $\al$ on a quantum chain $\caA$ is decoupled in of blocks of size $R$ if there is a partition of $\Z$ in intervals $\{I_i\}_{i \in \Z}$ such that $\al$ leaves each $\caA_{I_i}$ invariant, and $\max_{i} \abs{I_i} = R$. The restricitons $\al|_{\caA_{I_i}}$ are all even equivariant automorphisms.
\end{definition}
For any such decoupled QCA $\al$, we get infinitely many automorhisms $\al|_{\caA_{I_i}}$ on local algebras. In section \ref{sec:0-dimensional} a non-trivial index taking values in the first cohomology group of $\Z_2 \times G$ will be defined. This yields a `landscape' of essentially 0-dimensional obstructions to deforming the given decoupled QCA to the identity. Since we are interested in truly one-dimensional obstructions, we will take the point of view that decoupled QCAs are trivial, and define the following notion of stable $G$-equivalence:

\begin{definition} \label{def:stable G-equivalence}
	Two even equivariant quasi-local automorphisms $\al$, $\beta \in \scrA_Q$ on quantum chains $\caA$ and $\caA'$ respectively are called stably $G$-equivalent in $\scrA_Q$ if there exist quantum chains $\widetilde \caA$ and $\widetilde \caA'$ such that $\caA \otimes \widetilde \caA \simeq \caA' \otimes \widetilde \caA'$ as quantum chains and there are decoupled even equivariant automorphisms $\tilde \al$ and $\tilde \beta$ on $\widetilde \caA$ and $\widetilde \caA'$ respectively such that $\al \otimes \tilde \al$ and $\beta \otimes \tilde \beta$ are $G$-equivalent in $\scrA_{Q}$.
\end{definition}

\subsection{Classification result}

Our main result is that for any robust class of even equivariant quasi-local automorphisms, the $G$-stable equivalence classes are uniquely labelled by an index taking values in $\Q \cup \sqrt{2} \Q \times \Hom(G, \Z_2) \times H^2(G, U(1))$.

\begin{theorem} \label{thm:Main Theorem}
	There is an index
	\begin{equation}
		\ind : \scrA \rightarrow \Q \cup \sqrt{2} \Q \times \Hom(G, \Z_2) \times H^2(G, U(1)) 
	\end{equation}
	such that for each $(d, \zeta, \bnu) \in \Q \cup \sqrt{2} \Q \times \Hom(G, \Z_2) \times H^2(G, U(1))$ there is an even equivariant QCA $\al$ with $\ind(\al) = (d, \zeta, \bnu)$.

	If $\al$ and $\beta$ are two even equivariant quasi-local automorphisms belonging to the same robust collection $\scrA_Q$ then the following are equivalent:
	\begin{enumerate}[label=(\roman*)]
		\item $\ind(\al) = \ind(\beta)$.
		
		\item $\al$ and $\beta$ are stably $G$-equivalent in $\scrA_Q$.
	\end{enumerate}

	Moreover,
	\begin{equation}
		\ind(\al \gotimes \beta) = \ind(\al) \cdot \ind(\beta)
	\end{equation}
	and if $\al$ and $\beta$ are defined on the same quantum chain, then
	\begin{equation}
		\ind(\al \circ \beta) = \ind(\al) \cdot \ind(\beta)
	\end{equation}
	where the binary operation $\cdot$ is given by
	\begin{equation}
		\big( d_1, \zeta_1, \bnu_1 \big) \cdot \big( d_2, \zeta_2, \bnu_2 \big) = \big( d_1 d_2, \, \zeta_1 + \zeta_2, \,  \bnu_1 \cdot \bnu_2 \cdot \bnu(\zeta_1, \zeta_2) \big)
	\end{equation}
	for $(d_i, \zeta_i, \bnu_i) \in \Q \cup \sqrt{2} \Q \times \Hom(G, \Z_2) \times H^2(G, U(1))$ and $\bnu(\zeta_1, \zeta_2)$ is the element of $H^2(G, U(1))$ determined by the 2-cocycle
	\begin{equation}
		\nu(\zeta_1, \zeta_2) = (-1)^{\zeta_1(g) \zeta_2(h)}.
	\end{equation}
	This binary operation turns $\Q \cup \sqrt{2} \Q \times \Hom(G, \Z_2) \times H^2(G, U(1))$ into an abelian group. 
\end{theorem}

This theorem is proven at the end of Section \ref{sec:quasi-local completeness}.

\section{Operator algebraic preliminaries} \label{sec:operator algebraic preliminaries}

\subsection{Tracial state, GNS representation, and von Neumann algebras}

There is a state on the quantum chain $\caA$ called the \emph{tracial state}. It is given on local observables $a \in \caA_I$, $I \ssubset \Z$ by
\begin{equation}
	\tr (x) = \frac{1}{d_I} \Tr_I(x)
\end{equation}
where on the left hand side we consider $x$ to be an element of the $d_I \times d_I$ matrix algebra $\caA_I$ (where $d_I = \prod_{n \in I} d_n$) and $\Tr_I$ is the trace on that matrix algebra.

Let $(\caH, \pi, |\I\rangle )$ be the GNS triple associated to the tracial state, then the representation $\pi : \caA \rightarrow \caH$ is faithful. For any $I \subset \Z$ we denote by
\begin{equation}
	\caA_I^{\vN} := \big( \pi(\caA_I) \big)'' \subset \caB(\caH)
\end{equation}
the von Neumann algebra generated by $\pi(\caA_I)$. We denote also $\caA^{\vN} := \caA_{\Z}^{\vN}$. If $I$ is finite then $\caA_{I}^{\vN}$ is a finite full matrix algebra. If $I$ is infinite then $\caA_I^{\vN}$ is (isomorphic to) the hypefinite type $II_1$ factor.

The vectors in $\caH$ are equivalence classes of operators in $\caA$ under the equivalence relation $a \sim b$ if $\tr \big( (a-b)^*(a-b) \big) = 0$. But this is the case if and only if $a = b$, so as a vector space, $\caH$ is the same as $\caA$. We denote $| a \rangle \in \caH$ for each $a \in \caA$. The inner product on $\caH$ is given by $\langle a, b \rangle = \tr(a^* b)$.

\subsection{\texorpdfstring{$\Z_2$}{Z2}-grading and symmetry}

We note first that the grade automorphism $\theta$ and the local symmetry automorphisms $\rho^{g}$, $g \in G$, extend uniquely to automorphisms of the von Neumann Algebra $\caA^{\vN}$. Indeed, the GNS triple $(\caH, \pi, | \I \rangle)$ of $\caA$ w.r.t. to the tracial state $\tr$ is unique up to unitary equivalence. But since the tracial state is invariant under \emph{any} automorphism $\al$ of $\caA$, \ie $\tr \circ \al = \tr$, we find that $(\caH, \pi \circ \al, | \I \rangle)$ is also a GNS triple, so there exists a unitary $u_{\al} \in \caB(\caH)$ such that $\pi \circ \al = \Ad_{u_{\al}} \circ \pi$. Moreover, since the respresentation $\pi$ is continuous w.r.t. the weak operator topology, the von Neumann algebra $\caA^{\vN}$ is left invariant under $\Ad_{u_{\al}}$. Similarly, if $\al$ leaves $\caA_I$ invariant for some $I \subset \Z$, then $\Ad_{u_{\al}}$ leaves $\caA_I^{\vN}$ invariant. Note that the grade automorpism and the local symmetry automorphisms leave $\caA_I$ invariant for all $I \subset \Z$.

In the particular case of the grade automorphism $\theta$, we can take the unitary action as $\Theta \in \caB(\caH)$ with $\Theta | a \rangle = | \theta(a) \rangle$. We will denote both $\Ad_{\Theta}$ and $\Ad_{\Theta}|_{\caA_I^{\vN}}$, $I \subset \Z$ by $\theta$.

Similarly, for each $g \in G$ we can take the unitary $u_g \in \caB(\caH)$ implementing te local symmetry automorphism $\rho^{g}$ to act as $u_g | a \rangle = | \rho^g(a) \rangle$. We denote both $\Ad_{u_g}$ and $\Ad_{u_g}|_{\caA_I^{\vN}}$, $I \subset \Z$ also by $\rho^g$.

The following concept will be used extensively in the appendices:

\begin{definition}
	Let $g \in G$ or $g = \theta$. We say that a von Neumann algebra $\caB \subset \caH(\caH)$ is $g$-hyperfinite if it is the weak closure of an increasing family of $g$-invariant full matrix algebras all containing $\I$.

	If $\caB$ is $g$-hyperfinite for all $g$ in a group $G$, then we say that $\caB$ is $G$-hyperfinite
\end{definition}

Obviously $\caA_{I}^{\vN}$ is $\theta$, $G$-hyperfinite for any $I \subset \Z$.

\section{Classification of cellular automata on quantum chains} \label{sec:classification of QCAs}

We extend the results of \cite{FPPV2019} to include local symmetries, using techniques from \cite{RWW2020} adapted to the $\Z_2$-graded case.

\subsection{Overlap algebras}

Let $\al$ be a even QCA on a quantum chain $\caA$. Define subalgebras
\begin{align}
	\caB_n &:= \caA_{2n} \gotimes \caA_{2n+1} \\
	\caC_{n} &:= \caA_{2n-1} \gotimes \caA_{2n}
\end{align}
and let
\begin{align}
	\caL_n &:= \caC_n \cap \al(\caB_n) \label{eq:left overlap} \\
	\caR_n &:= \caC_{n+1} \cap \al(\caB_{n}), \label{eq:right overlap}
\end{align}
be the overlap algebras.

Following \cite{Vara2004} we define
\begin{definition}
	A finite-dimensional super *-algebra that is semisimple as an ungraded algebra and has trivial supercenter is called a central simple superalgebra.
\end{definition}
Theorem 6.2.5  of \cite{Vara2004} shows that each central simple superalgebra is isomorphic to $\caM^{p|q}$ or to $\caM^{p | q} \otimes K$ for some $p, q \in \N$, where $K := \C[\ep]$ for an odd self-adjoint element $\ep$ such that $\ep^2 = \I$.

We have
\begin{theorem}[\cite{FPPV2019}] \label{thm:overlap factorization}
	Suppose $\al$ is an even automrphism of $\caA$.

	If there is $n \in \Z$ such that
	\begin{align*}
		\al\big( \caB_n \big) &\subset \caC_n \gotimes \caC_{n+1}, \\
		\al \big( \caB_{n-1} \big) &\subset \caC_{n-1} \gotimes \caC_n, \\
		\al^{-1} \big( \caC_n \big) &\subset \caB_{n-1} \gotimes \caB_n 
	\end{align*}
	then
	\begin{equation}\label{eq:factorization of C} 
		\caC_n = \caR_{n-1} \gotimes \caL_n
	\end{equation}
	and $\caR_{n-1}$ and $\caL_n$ are central simple superalgebras.

	If there is $n \in \Z$ such that
	\begin{align*}
		\al^{-1}\big( \caC_n \big) &\subset \caB_{n-1} \gotimes \caB_{n}, \\
		\al^{-1}\big( \caC_{n+1} \big) &\subset \caB_n \gotimes \caB_{n+1}, \\
		\al \big( \caB_n \big) &\subset \caC_n \gotimes \caC_{n+1}
	\end{align*}
	then
	\begin{equation} \label{eq:factorization of B}
		\caB_n = \al^{-1}(\caL_n) \gotimes \al^{-1}(\caR_n)
	\end{equation}
	and $\caL_n$ and $\caR_n$ are central simple superalgebras.

	In particular, if $\al$ is an even QCA, then the above hold for all $n \in \Z$.
\end{theorem}

In order to prove the theorem we will make use of a family of conditional expectations on arbitary subalgebras $\caA_X$, $X \subset \Z$ of the quantum chain, introduced in \cite{Araki2004}. (see Definition \ref{def:conditional expectation} for a definition of conditional expectation.)

\begin{proposition}[Theorem 3.1. of \cite{Araki2004}] \label{prop:Araki's conditional expectations}
	Let $\tau$ be the tracial state on the quantum chain $\caA$. For each $X \subset \Z$ there is a conditional expectation $E_X : \caA \rightarrow \caA_X$ such that
	\begin{enumerate}[label=(\roman*)]
		\item \begin{equation} \tau(b_1 a b_2) = \tau( b_1 E_X(a) b_2 ) \end{equation}
			for all $b_1, b_2 \in \caA_X$ and all $a \in \caA$.
		\item The conditional expectations respect the grading, \ie
			\begin{equation}
				\theta \circ E_X = E_X \circ \theta
			\end{equation}
			for all $X \subset \Z$.
		\item For any $X, Y \subset \Z$,
			\begin{equation}
				E_X E_Y = E_Y E_X = E_{X \cap Y}.
			\end{equation}
	\end{enumerate}
\end{proposition}

Throught this paper we will denote $E_{B_n} := E_{\{2n, 2n+1\}}$, $E_{\caB_n \gotimes \caB_{n+1}} := E_{\{2n, 2n+1, 2n+2, 2n+3\}}$, $E_{\caC_n} = E_{\{2n-1, 2n\}}$, and $E_{\caC_n \gotimes \caC_{n+1}} = E_{\{2n-1, 2n, 2n+1, 2n+2\}}$ for all $n \in \Z$.

\begin{proofof}[Theorem \ref{thm:overlap factorization}]
	Following \cite{RWW2020}, to show that $\caC_n = \caR_{n-1} \gotimes \caL_n$ we first show that
	\begin{align}
		\caL_n &= E_{\caC_n} \big(\al(\caB_n) \big) \label{eq:cond expectation 1}\\
		\caR_n &= E_{\caC_{n+1}} \big(\al(\caB_n)\big) \label{eq:cond expectation 2} \\	
	\end{align}

	To show te first of these, note first that by definition, $\caL_n = \al(\caB_n) \cap \caC_n \subset \caC_n$, so $\caL_n = E_{\caC_n}(\caL_n) = E_{\caC_n}(\al(\caB_n)) \cap \caC_n = E_{\caC_n}(\al(\caB_n))$. We now show the reverse inclusion $E_{\caC_n}(\al(\caB_n)) \subset \caL_n$. Let $x \in \caB_n$ be a homogeneous element, then $E_{\caC_n}(\al(x))$ is also homogeneous and of the same degree as $x$. Moreover, since we assume $\al \big(\caB_n \big) \subset \caC_n \gotimes \caC_{n+1}$ we have $\al(x) \in \caC_n \gotimes \caC_{n+1}$ so
	\begin{equation}
		E_{\caC_{n-1} \gotimes \caC_{n}}(\al(x)) = E_{\caC_{n-1} \gotimes \caC_n} \big( E_{\caC_n \otimes \caC_{n+1}}(\al(x)) \big) = E_{\caC_n}(\al(x))
	\end{equation}
	where we used property (iii) of \ref{prop:Araki's conditional expectations}. Let now $x' \in \caB_{n-1}$ be homogeneous, then 
	\begin{align*}
		E_{\caC_{n}} \big(\al(x) \big) \al(x') &= E_{\caC_{n-1} \gotimes \caC_n} \big(\al(x) \al(x') \big) = (-1)^{\tau(x) \tau(x')} E_{\caC_{n-1} \gotimes \caC_n} \big(\al(x') \al(x) \big) \\
						       &= (-1)^{\tau(x) \tau(x')} \al(x') E_{\caC_n} \big(\al(x) \big).
	\end{align*}
	where we used the assumtion that $\al \big( \caB_{n-1} \big) \subset \caC_{n-1} \gotimes \caC_n$. It follows that $\al(\caB_{n-1})$ and $E_{\caC_n} \big(\al(\caB_n)\big)$ graded commute.

	From the assumption that $\al^{-1}(\caC_n) \subset \caB_{n-1} \gotimes \caB_n$ we get $\caC_n \subset \al(\caB_{n-1}) \gotimes \al(\caB_n)$, hence $E_{\caC_n} \big(\al(\caB_n) \big) \subset \caC_n \subset \al(\caB_{n-1}) \gotimes \al(\caB_{n})$. But $E_{\caC_n} \big(\al(\caB_n)\big)$ is also in the supercommutant of $\al(\caB_{n-1})$. It follows that $E_{\caC_n} \big(\al(\caB_n) \big) \subset \al(\caB_n)$ hence $E_{\caC_n} \big(\al(\caB_n) \big) \subset \al(\caB_n) \subset \al(\caB_n) \cap \caC_n = \caL_n$.

We conclude that $\caL_n = E_{\caC_n} \big(\al(\caB_n) \big)$. The other equality, Eq. \eqref{eq:cond expectation 2}, follows from a similar argument. \eqref{eq:cond expectation 3} and \eqref{eq:cond expectation 4} follow from a similar argument.

	We now show that $\caC_n \subset \caR_{n-1} \gotimes \caL_n$. If $c \in \caC_n$ then $\al^{-1}(c) \in \caB_{n-1} \gotimes \caB_n$ because $\al^{-1}$ is nearest-neighbour. As such, we have an expansion $\al^{-1}(c) = \sum_k b_k b'_k$ with $b_k \in \caB_{n-1}$ and $b'_k \in \caB_{n}$ for each $k$. We then find
	\begin{equation}
		c = E_{\caC_n}(c) = \sum_k E_{\caC_n} \big( \al(b_k) \big) E_{\caC_n} \big( \al(b'_k) \big) \in \caR_{n-1} \gotimes \caL_n
	\end{equation}
	where we used \eqref{eq:cond expectation 1} and \eqref{eq:cond expectation 2}. Since we obviously have the reverse inclusion, we conclude that $\caC_n = \caR_{n-1} \gotimes \caL_n$ as required.

	It follows that the overlap algebras $\caR_{n-1}$ and $\caL_n$ are subalgebras of the finite matrix algebra $\caC_n$. They are therefore semisimple. They are moreover invariant under the grade automorphism $\theta$, giving them the structure of superalgebras. Suppose any one of them had non-trivial supercenter, then also $\caC_n$ has non-trivial supercenter, and since $\caA$ is the direct limit of the graded tensor product of all the $\caC_m$'s, it would follow that the quantum chain $\caA$ itself has non-trivial supercenter. This would contradict the simplicity of $\caA$ (\cite{BratteliRobinsonVol1}, Corollary 2.6.19.) so the overlap algebras $\caR_{n-1}$ and $\caL_n$ are semisimple and have trivial supercenter, \ie they are central simple superalgebras.

	Finally, $\caB_n = \al^{-1}(\caL_n) \gotimes \al^{-1}(\caR_n)$ follows from a similar argument using
	\begin{align}
		\al^{-1}(\caL_n) &= E_{\caB_n} \big( \al^{-1}(\caC_n) \big) \label{eq:cond expectation 3} \\
		\al^{-1}(\caR_n) &= E_{\caB_n} \big( \al^{-1}(\caC_{n+1}) \big), \label{eq:cond expectation 4}
	\end{align}
	which are shown in the same way as Eqs. \eqref{eq:cond expectation 1} and \eqref{eq:cond expectation 2}, using the assumptions $\al^{-1}\big( \caC_n \big) \subset \caB_{n-1} \gotimes \caB_{n}$, $\al^{-1}\big( \caC_{n+1} \big) \subset \caB_n \gotimes \caB_{n+1}$ and $\al \big( \caB_n \big) \subset \caC_n \gotimes \caC_{n+1}$.

	It follows in the same way as above that $\caL_n$ and $\caR_n$ are central simple superalgebras.
\end{proofof}

\subsection{Central Simple \texorpdfstring{$G$}{G}-systems} \label{sec:CS G-systems}

In this section we develop an index theory for central simple $G$-systems. There is some similarity with \cite{BO2021} and we use some ideas from that work.

Let $\caB$ be a central simple superalgebra with grading automorphism $\theta$ and let $\rho$ be an even automorphic $G$-action on $\caB$. We call the pair $(\caB, \rho)$ a \emph{central simple $G$-system}. To any such central simple $G$-system we associate an index
\begin{align}
	\ind(\caB, \rho) &= (\ind_1(\caB), \ind_{2}(\caB, \rho), \ind_3(\caB, \rho)) \\
			 &\in \caS \times \Hom(G, \Z_2) \times H^2(G, U(1))
\end{align}
where
\begin{equation}
	\caS = \N \cup \sqrt{2} \N.  
\end{equation}
(Here $\N$ does not contain zero.)

\subsubsection{Definition of an index}

The first component of $\ind(\caB, \rho)$ is simply the square root of the dimension (as a complex vector space) of $\caB$,
\begin{equation}
	\ind_1(\caB) := \sqrt{ \dim(\caB) }.
\end{equation}
We saw above that for a central simple $\caB$, either $\caB \simeq \caM^{p|q}$ or $\caB \simeq \caM^n \otimes K$, so the dimension is either the square of a whole number, or twice the square of a whole number (since $\dim(K) = 2$), so $\ind_1 \in \caS$. We call $\caB$ rational of $\ind_1(\caB) \in \N$ and we call $\caB$ radical if $\ind_1(\caB) \in \sqrt{2} \N$

The definition of the other components of the index depends on whether $\caB$ is rational or radical.

\textbf{Case I : $\caB \simeq \caM^{p|q}$ is rational.} Then, since $\caM^{p|q}$ seen as an ungraded algebra, is simply a full matrix algebra, we have unitaries $\Theta$ and $V(g)$ for all $g \in G$ such that $\theta = \Ad_{\Theta}$ and $\rho^g = \Ad_{V(g)}$. These unitaries are uniquely determined up to phase. Moreover, since $\theta^2 = \id$ we have $\Theta^2 = \mu \I$ for some phase $\mu$, and by choosing a different phase for $\Theta$ we can without loss of generality assume $\Theta^2 = \I$. The unitaries $V(g)$ form a projective representation of $G$. \ie there are phases $\nu(g, h)$ such that
\begin{equation}
	V(g) V(h) = \nu(g, h) V(gh)
\end{equation}
for all $g, h \in G$. The function $\nu : G \times G \rightarrow U(1)$ is a 2-cocycle and picks out an equivalence class
\begin{equation}
	\ind_3(\caB, \rho) := [\nu] \in H^2(G, U(1))
\end{equation}
in the second cohomology group of $G$. This equivalence class is independent of the choice of phases of the $V(g)$, so $\ind_3$ is well defined.

The index $\ind_2(\caB, \rho)$ characterizes the interplay of the group action and the $\Z_2$-grading. Since $\rho^g \circ \theta = \theta \circ \rho^g$, we have $\rho^g(\Theta) = V(g) \Theta V(g)^* = \xi(g) \Theta$ for some phases $\xi(g) \in U(1)$. Squaring both sides we get $\I = \xi(g)^2 \I$, hence $\xi(g) \in \{ \pm 1 \}$ and we can write $\rho^g(\Theta) = (-1)^{\zeta(g)} \Theta$ for $\zeta(g) \in \{0, 1\}$. Moreover, regarding $\zeta(g)$ as elements of the additive group $\Z_2$, for any $g, h \in G$ we find
\begin{equation}
	(-1)^{\zeta(g) + \zeta(h)} \Theta = \rho^g(\rho^h(\Theta)) = \rho^{gh}(\Theta) = (-1)^{\zeta(gh)} (\Theta)
\end{equation}
so $\zeta(g) + \zeta(h) = \zeta(g h)$ and, taking $g = h = e$, we find $\zeta(e) = 0$. We see that $\zeta$ is a homomorphism from $G$ to $\Z_2$. Thus we put
\begin{equation}
	\ind_2(\caB, \rho) := \zeta \in \Hom(G, \Z_2).
\end{equation}

\textbf{Case II : $\caB \simeq \caM^n \otimes K$ is radical.} Since $\rho^g \circ \theta = \theta \circ \rho^g$, the automorphisms $\rho^g$ leave the even subalgebra of $\caB$ invariant. The even subalgebra consists of elements $a \otimes \I$ and is a full matrix algebra isomorphic to $\caM^n$. It follows that there are unitaries $V^{(0)}(g)$ forming a projective representation of $G$ such that $\rho^g(a \otimes \I) = \Ad_{V^{(0)}(g)}(a) \otimes \I$ for all $a \in \caM^n$. The index $\ind_3(\caB, \rho)$ is defined to be the second cohomology class associated to this projective representation.

To define $\ind_2(\caB, \rho)$, note that $\pm \I \otimes \ep$ are the only odd self-adjoint unitaries in the center of $\caB$. Since for each $g \in G$, the element $\rho^g(\ep)$ is an odd self-adjoint unitary in the center of $\caB$, we must have $\rho^g(\ep) = (-1)^{\zeta(g)} \ep$ for some $\zeta(g) \in \{0, 1\}$. Regarding the $\zeta(g)$ as elements of the additive group $\Z_2$, we have moreover
\begin{equation}
	(-1)^{\zeta(g) + \zeta(h)} \ep = \rho^g(\rho^h(\ep)) = \rho^{gh}(\ep) = (-1)^{\zeta(gh)} \ep
\end{equation}
so $\zeta(g) + \zeta(h) = \zeta(g h)$ and, taking $g = h = e$, we find $\zeta(e) = 0$. We see that $\zeta$ is a homomorphism from $G$ to $\Z_2$. Thus we put
\begin{equation}
	\ind_2(\caB, \rho) := \zeta \in \Hom(G, \Z_2).
\end{equation}

This concludes the definition of the index in all cases.

\subsubsection{Representation of the graded tensor product}

Let $\caB_1 \subset \caM^{m_1}$ and $\caB_2 \subset \caM^{m_2}$ be superalgebras with grading operators $\Theta_1 \in \caM^{m_1}$ and $\Theta_2 \in \caM^{m_2}$ respectively. Note that the grading operators are not necessarily in the algebras $\caB_1, \caB_2$ themselves. The graded tensor product $\caB_1 \gotimes \caB_2$ has a faithful representation $\pi : \caB_1 \gotimes \caB_2 \rightarrow \caM^{m_1} \otimes \caM^{m_2}$ and is defined for homogeneous $b_1 \in \caB_1$ and $b_2 \in \caB_2$ by
\begin{equation} \label{eq:definition of pi}
	\pi \big( b_1 \gotimes b_2 \big) = b_1 \Theta_1^{\tau(b_2)} \otimes b_2
\end{equation}
where $\tau(x) \in \{0, 1\}$ is such that $\theta(x) = (-1)^{\tau(x)} x$, \ie $\tau$ is the \emph{parity} of $x$ (see Appendix \ref{app:super vN} for notation and basic properties of super algebras.)

The grading automorphism $\theta$ of $\caB_1 \gotimes \caB_2$ is given by
\begin{equation}
	\pi \circ \theta = \Ad_{\Theta_1 \otimes \Theta_2} \circ \pi.
\end{equation}

We have

\begin{lemma} \label{lem:action of graded tensor product}
	Let $\al_1$ and $\al_2$ be even automorphisms acting on $\caB_1$ and $\caB_2$. Suppose there are unitaries $U_i \in \caM^{m_i}$ such that
	\begin{equation}
		\al_i(b) = \Ad_{U_i}(b), \quad \text{for all} \,\,\, b \in \caB_i
	\end{equation}
	and $\theta_i(U_i) = (-1)^{\xi_i} U_i$ for some $\xi_i \in \{0, 1\}$, then
	\begin{equation}
		\pi \circ (\al_1 \gotimes \al_2) = \Ad_{(U_1 \otimes U_2 \Theta_2^{\xi_1})} \circ \pi = \pi \circ \Ad_{(U_1 \Theta_1^{\xi_2} \gotimes U_2 \Theta_2^{\xi_1})}.
	\end{equation}
\end{lemma}

\begin{proof}
	For homogeneous elements $b_i \in \caB_i$, $i = 1, 2$,
	\begin{align*}
		\pi \big( (\al_1 \gotimes \al_2)(b_1 \gotimes b_2) \big) &= \al_1(b_1) \Theta_1^{\xi_2} \otimes \al_2(b_2) \\
									 &= \Ad_{U_1}(b_1) \Theta_1^{\xi_2} \otimes \Ad_{U_2}(b_2) \\
									 &= (-1)^{\xi_1 \xi_2} \Ad_{U_1} \big( b_1 \Theta_1^{\xi_2} \big) \otimes \Ad_{U_2}(b_2) \\
									 &= \Ad_{(U_1 \otimes U_2 \Theta_2^{\xi_1})} \big( b_1 \Theta_1^{\xi_2} \otimes b_2 \big) \\
									 &= \Ad_{(U_1 \otimes U_2 \Theta_2^{\xi_1})} \big( \pi(b_1 \gotimes b_2) \big),
	\end{align*}
	which proves the claim.
\end{proof}

\subsubsection{stacking}

\begin{proposition} \label{prop:stacking for absolute ind}
	Let $\caB_1$ and $\caB_2$ be central simple superalgebras equipped with even automorphic $G$-actions $\rho_1$ and $\rho_2$ respectively. Then $\caB_1 \gotimes \caB_2$ is a central simple superalgebra equipped with even automorphic $G$-action $\rho_1 \gotimes \rho_2$. If $\ind(\caB_i, \rho_i) = (d_i, \zeta_i, \bnu_i)$ for $i = 1, 2$, then
	\begin{equation}
		\ind(\caB_1 \gotimes \caB_2, \rho_1 \gotimes \rho_2) = \big( d_1 d_2, \, \zeta_1 + \zeta_2, \,  \bnu_1 \cdot \bnu_2 \cdot \bnu(\zeta_1, \zeta_2) \big)
	\end{equation}
	where $\bnu(\zeta_1, \zeta_2)$ is the element of $H^2(G, U(1))$ determined by the 2-cocycle
	\begin{equation}
		\nu(\zeta_1, \zeta_2) = (-1)^{\zeta_1(g) \zeta_2(h)}.
	\end{equation}
\end{proposition}

\begin{proof}
	\textbf{Case I : $\caB_1 \simeq \caM^{p_1 | q_1}$ and $\caB_2 \simeq \caM^{p_2 | q_2}$ are both rational.} We have $\caM^{p_1 | q_1} \gotimes \caM^{p_2 | q_2} \simeq \caM^{p_1 p_2 + q_1 q_2 | p_1 q_2 + q_1 p_2}$ so we get again a rational central simple superalgebra. The dimension of the graded tensor product is the product of the dimensions of the individual algebras, thus verifying the claim about $\ind_1$.
	
	Let $V_1, V_2$ be the projective representations of $G$ such that $\rho_1^g = \Ad_{V_1(g)}$ and $\rho_2^g = \Ad_{V_2(g)}$, and let $\nu_1$ and $\nu_2$ be the 2-cocylcles associated to these projective representations. The group $G$ acts on $\caB = \caB_1 \gotimes \caB_2$ by $\rho^g = \rho_1^g \gotimes \rho_2^g$. In the representation $\pi : \caB_1 \gotimes \caB_2 \rightarrow \caM^{p_1 + q_1} \otimes \caM^{p_2 + q_2}$ defined in Eq. \eqref{eq:definition of pi}, we have by Lemma \ref{lem:action of graded tensor product} that $\rho^g = \Ad_{V(g)}$ with $V(g) := V_1(g) \Theta_1^{\zeta_2(g)} \gotimes V_2(g) \Theta_2^{\zeta_1(g)}$ forming a projective representation on $G$,
	\begin{equation}
		\pi \big( V(g) V(h) \big) = V_1(g) V_1(h) \otimes V_2(g) \Theta_2^{\zeta_1(g)} V_2(h) \Theta_2^{\zeta_1(h)} = \nu_1(g, h) \nu_2(g, h) \nu(\zeta_1, \zeta_2)(g, h) \pi(V(gh)).
	\end{equation}
	Since $\pi$ is a faithful representation, we have
	\begin{equation}
		\ind_3(\caB_1 \gotimes \caB_2, \rho_1 \gotimes \rho_2) = [ \nu_1 \cdot \nu_2 \cdot \nu(\zeta_1, \zeta_2) ].
	\end{equation}

	Finally, remebering that in the representation $\pi$ the grading operator $\Theta$ of $\caB_1 \gotimes \caB_2$ is given by $\Theta_1 \otimes \Theta_2$, we find
	\begin{equation}
		\pi \big( \rho^g(\Theta) \big) = \Ad_{V_1(g) \otimes V_2(g) \Theta_2^{\zeta_1(g)}} \big( \Theta_1 \otimes \Theta_2 \big) = (-1)^{\zeta_1(g) + \zeta_2(g)} \pi \big( \Theta \big)
	\end{equation}
	so
	\begin{equation}
		\ind_2(\caB_1 \gotimes \caB_2, \rho_1 \gotimes \rho_2) = \zeta_1 + \zeta_2,
	\end{equation}
	as required.

	\textbf{Case II : $\caB_1 \simeq \caM^{p | q}$ and $\caB_2 \simeq \caM^n \otimes K$.} Denote by $\ep$ the odd self-adjoint unitary generator of $K$. We represent $\caB_2$ as a subalgebra of $\caM^n \otimes \caM^2$ through the embedding $\iota : \caM^n \otimes K \rightarrow \caM^n \otimes \caM^2$ such that $\iota(a \otimes \I) = a \otimes \I$ and $\iota(\I \otimes \ep) = \I \otimes \sigma_X$. This determines the injective *-homomorphism $\iota$ uniquely. We will henceforth identify $\caB_2$ with its image $\iota(\caB_2) \subset \caM^n \otimes \caM^2$. The grading automorphism $\theta_2$ of $\caB_2$ is given by conjugation with the grading operator $\Theta_2 = \I \otimes \sigma_Z$. Note that $\Theta_2$ is not an element of $\caB_2$.

	The graded tensor product $\caB = \caB_1 \gotimes \caB_2$ is isomorphic to the image of the representation
	\begin{equation}
		\pi : \caB_1 \gotimes \caB_2 \rightarrow \caM^{p + q} \otimes ( \caM^n \otimes \caM_2 )
	\end{equation}
	defined in Eq. \eqref{eq:definition of pi} where the grading operator in the image is $\Theta_1 \otimes (\I \otimes \sigma_Z)$.

	We show that $\pi(\caB_1 \gotimes \caB_2)$ (and hence $\caB_1 \gotimes \caB_2$ itself) is isomorphic as a superalgebra to $\caM^{n(p+q)} \otimes K$. We have $\pi(a \gotimes (b \otimes \ep^{\kappa})) = a \Theta_1^{\kappa} \otimes (b \otimes \sigma_X^{\kappa})$, $\kappa = 0, 1$ so the image of $\pi$ is spanned by elements of the form $a \otimes b \otimes \sigma_X^{\kappa}$, $\kappa = 0, 1$, whose parity is $\tau(a) + \kappa$. Denote by $p_1^{\pm}$ the porojections on the $\pm 1$ eigenspaces of the grading operator $\Theta_1$, so $\Theta_1 = p_1^+ - p_1^-$. We define a *-isomorphism $f : \caM^{p+q} \otimes (\caM^n \otimes \caM^2) \rightarrow \caM^{n(p+q)} \otimes \caM^2$ by conjugation with the unitary
	\begin{equation}
		p_{1}^+ \otimes ( \I \otimes  \I) \,\, + \,\, p_{1}^- \otimes ( \I \otimes \sigma_X).
	\end{equation}
	This *-automorphism satisfies
	\begin{align}
		f(a \otimes (b \otimes \sigma_X^{\tau(a)}) ) &= (a \otimes b) \otimes \I \nonumber \\
		f(\I \otimes (\I \otimes \sigma_X) ) &= \I \otimes \sigma_X \label{eq:properties of f case II} \\
		f(\Theta_1 \otimes (\I \otimes \sigma_Z)) &= \I \otimes \sigma_Z. \nonumber
	\end{align}
	The image of the even subalgebra $(f \circ \pi)(\caB^{(0)})$ is spanned by elements $a \otimes \I$ for arbitrary $a \in \caM^{n(p+q)}$, and the image of the odd elements $(f \circ \pi)(\caB^{(1)})$ is spanned by elements $a \otimes \sigma_X$ for arbitrary $a \in \caM^{n(p+q)}$. Hence, $(f \circ \pi)|_{\caB_1 \gotimes \caB_2}$ is a graded *-isomorphism from $\caB_1 \gotimes \caB_2$ to $\caM^{n(p + q)} \otimes K$, where $K$ is now identified as the subalgebra of $\caM^2$ generated by $\sigma_X$.

	The dimension of $\caB_1 \gotimes \caB_2 \simeq \caM^{n(p+1)} \otimes K$ is the product of the dimensions of $\caB_1 \simeq \caM^{p|q}$ and $\caB_2 \simeq \caM^n \otimes K$, thus veryfying the claim about $\ind_1$.

	Still identifying $\caB_2$ with its image $\iota(\caB_2) \subset \caM^n \otimes \caM^2$ we can extend the *-automorphisms $\rho_2^g$ to *-automorphisms $\tilde \rho_2^g$ acting on the whole of $\caM^n \otimes \caM_2$ by setting $\tilde \rho_2^g(\I \otimes \sigma_Z) = (-1)^{\zeta_2(g)} \I \otimes \sigma_Z$. The grading automorphism $\theta_2$ also extends to a $\tilde \theta_2$, simply letting it act by conjugation with the grading operator $\I \otimes \sigma_Z$. One easily verifies that the $\tilde \rho_2^g$ thus extended still form an even automorphic $G$ representation w.r.t. the grading defined by $\tilde\theta_2$ on the whole of $\caM^n \otimes \caM^2$. Since $\caM^n \otimes \caM^2$ is a full matrix algebra there are unitaries $V_2(g)$ such that $\tilde \rho_2^g = \Ad_{V_2(g)}$, and these unitaries form a projective representation of $G$.

	Recall that the $\rho_2^g$, and therefore the $\tilde \rho_2^g$, leave the even subalgebra $\caB_2^{(0)} = \caM^n \otimes \I$ invariant, so that there are unitaries $V_2^{(0)}(g) \in \caM^n$ such that
	$\tilde \rho_2^{g}(a \otimes \I) = \rho_2^g(a \otimes \I) = \Ad_{V_2^{(0)}(g)}(a) \otimes \I$ for all $a \in \caM^n$. The unitaries $V_2(g)^* ( V_2^{(0)}(g) \otimes \sigma_Y^{\zeta_2(g)} )$ commute with the even subalgebra $\caM^n \otimes \I$ and with the elements $\I \otimes \sigma_X$ and $\I \otimes \sigma_Z$ (recall that, by definition, $\rho_2^g(\I \otimes \sigma_X) = (-1)^{\zeta_2(g)} (\I \otimes \sigma_X)$). These elements together generate the whole of $\caM^n \otimes \caM^n$. It follows that, after possibly redefining the phases of the $V_2^{(0)}(g)$, we have
	\begin{equation}
		V_2(g) = V_2^{(0)}(g) \otimes \sigma_Y^{\zeta_2(g)}.
	\end{equation}
	The second cohomology index of $(\caB_2, \rho_2)$ is by definition $[\nu_2]$ where $\nu_2$ is the 2-cocycle associated to the projective representation $V_2^{(0)}$.

	By Lemma \ref{lem:action of graded tensor product}, the automorphisms $\tilde \rho^g = \rho_1^g \gotimes \tilde \rho_2^g$ act on $\caM^{p+q} \otimes (\caM^n \otimes \caM^2) \supset \pi \big( \caB_1 \gotimes \caB_2 \big)$ through conjugation by
	\begin{align*}
		V_1(g) \otimes V_2(g) (\I \otimes \sigma_Z)^{\zeta_1(g)} &= V_1(g) \otimes (V_2^{(0)}(g) \otimes \sigma_Z^{\zeta_1(g)} \sigma_Y^{\zeta_2(g)}) \\
									 &= \iu^{\zeta_2(g)} V_1(g) \otimes (V_2^{(0)}(g) \otimes \sigma_Z^{\zeta_1(g)} (\sigma_X \sigma_Z)^{\zeta_2(g)}).
	\end{align*}
	It follows that the even automorphic $G$-action on $(f \circ \pi)(\caB_1 \gotimes \caB_2) \subset \caM^{n(p+q)} \otimes \caM^2$ acts by conjugation with unitaries
	\begin{equation}
		f \big( V_1(g) \otimes (V_2^{(0)}(g) \otimes \sigma_Z^{\zeta_1(g)} (\sigma_X \sigma_Z)^{\zeta_2(g)}) \big) = V_1(g) \Theta_1^{\zeta_1(g) + \zeta_2(g)} \otimes V_2^{(0)}(g) \otimes (- \iu \sigma_Y)^{\zeta_1(g) + \zeta_2(g)},
	\end{equation}
	which follows after a short computation using the properties of $f$ listed in Eq. \eqref{eq:properties of f case II}. We see that the even subalgebra $\caM^{n(p+q)} \otimes \I$ is invariant as required, and it is acted upon through conjugation by unitaries $V^{(0)}(g) = V_1(g) \Theta_1^{\zeta_1(g) + \zeta_2(g)} \otimes V_2^{(0)}(g)$. These unitaries form a projective representation
	\begin{equation}
		V^{(0)}(g) V^{(0)}(h) = \nu_1(g, h) \nu_2(g, h) \nu(\zeta_1, \zeta_2)(g, h) V^{(0)}(gh)
	\end{equation}
	so
	\begin{equation}
		\ind_3(\caB_1 \gotimes \caB_2, \rho_1 \gotimes \rho_2) = [ \nu_1 \cdot \nu_2 \cdot \nu(\zeta_1, \zeta_2) ].
	\end{equation}
	Finally, the unique (up to sign) odd self-adjoint unitary in the center of $(f \circ \pi)(\caB_1 \gotimes \caB_2)$ is $\I \otimes \I \otimes \sigma_X$ and
	\begin{equation}
		\Ad_{V_1(g) \Theta_1^{\zeta_1(g) + \zeta_2(g)} \otimes V_2^{(0)}(g) \otimes (- \iu \sigma_Y)^{\zeta_1(g) + \zeta_2(g)}}(\I \otimes \I \otimes \sigma_X) = (-1)^{\zeta_1(g) + \zeta_2(g)} (\I \otimes \I \otimes \sigma_X),
	\end{equation}
	so that
	\begin{equation}
		\ind_2(\caB_1 \gotimes \caB_2, \rho_1 \gotimes \rho_2) = \zeta_1 + \zeta_2,
	\end{equation}
	as required.

	\textbf{Case III : $\caB_1 \simeq \caM^m \otimes K$ and $\caB_2 \simeq \caM^n \otimes K'$ are both radical.} As in case II we can represent $\caB_1 = \caM^m \otimes K$ and $ \caB_2 = \caM^n \otimes K'$ as subalgebras of the full matrix algebras $\caM^m \otimes \caM^2$ and $\caM^n \otimes \caM^2$ where the odd elements $\I \otimes \ep$, $\I \otimes \ep'$ are identified with $\I \otimes \sigma_X$. The grading operators are then $\I \otimes \sigma_Z$. The graded tensor product $\caB = \caB_1 \gotimes \caB_2$ is isomorphic to the image of the representation
	\begin{equation}
		\pi : \caB_1 \gotimes \caB_2 \rightarrow \left( \caM^m \otimes \caM^2 \right) \otimes \left( \caM^n \otimes \caM^2 \right).
	\end{equation}
	The grading operator in the image is $( \I \otimes \sigma_Z) \otimes (\I \otimes \sigma_Z)$.

	We show that $\pi(\caB_1 \gotimes \caB_2)$ (and hence $\caB_1 \gotimes \caB_2$ itself) is isomorphic as a superalgebra to $\caM^{mn | mn}$. Define a *-isomorphism $f : \left( \caM^m \otimes \caM^2 \right) \otimes \left( \caM^n \otimes \caM^2 \right) \rightarrow \left( \caM^m \otimes \caM^2 \right) \otimes \left( \caM^n \otimes \caM^2 \right)$ by conjugation with
	\begin{equation}
		(\I \otimes p_+) \otimes (\I \otimes \I) + (\I \otimes p_-) \otimes (\I \otimes \sigma_Y)
	\end{equation}
	where $p_{\pm}$ is the projector on the $\pm 1$ eigenspace of $\sigma_Z$.

	Then
	\begin{align}
		f\big( (a \otimes \I) \otimes ( b \otimes \I) \big) &= (a \otimes \I) \otimes (b \otimes \I) \nonumber \\
		f\big( (a \otimes \sigma_X) \otimes ( b \otimes \I) \big) &= (a \otimes \sigma_X) \otimes (b \otimes \sigma_Y) \nonumber \\
		f\big( (a \otimes \I) \otimes ( b \otimes \sigma_X) \big) &= (a \otimes \sigma_Z) \otimes (b \otimes \sigma_X) \nonumber \\
		f\big( (a \otimes \sigma_X) \otimes ( b \otimes \sigma_X) \big) &= - (a \otimes \sigma_Y) \otimes (b \otimes \sigma_Z) \label{eq:properties of f case III} \\
		f\big( (a \otimes \sigma_Z) \otimes ( b \otimes \I) \big) &= (a \otimes \sigma_Z) \otimes (b \otimes \I) \nonumber \\
		f\big( (a \otimes \I) \otimes ( b \otimes \sigma_Z) \big) &= (a \otimes \sigma_Z) \otimes (b \otimes \sigma_Z) \nonumber 
	\end{align}
	and the grading operator $(\I \otimes \sigma_Z) \otimes (\I \otimes \sigma_Z)$ is mapped to
	\begin{equation*}
		f\big( (\I \otimes \sigma_Z) \otimes ( \I \otimes \sigma_Z) \big) = (\I \otimes \I) \otimes (\I \otimes \sigma_Z).
	\end{equation*}

	Next, we consider the *-homomorphism $g : \left( \caM^m \otimes \caM^2 \right) \otimes \left( \caM^n \otimes \caM^2 \right) \rightarrow \caM^{mn} \otimes \caM^2$ which forgets the second factor, i.e. $g$ is defined by
	\begin{equation}
		g(a \otimes b \otimes c \otimes d) = (a \otimes c) \otimes d.
	\end{equation}

	Elements of the form $(a \otimes \I) \otimes ( b \otimes \I)$ and $(a \otimes \sigma_X) \otimes (b \otimes \sigma_X)$ form a basis for the even subalgebra of $\pi(\caB_1 \gotimes \caB_2)$. They are mapped by the *-homomorphism $g \circ f$ to 
	\begin{align*}
		(g \circ f) \big( (a \otimes \I) \otimes (b \otimes \I) \big) &= (a \otimes b) \otimes \I, \\
		(g \circ f) \big( (a \otimes \sigma_X) \otimes (b \otimes \sigma_X) \big) &= (a \otimes b) \otimes (-\sigma_Z).
	\end{align*}

	The odd subspace of $\pi(\caB_1 \gotimes \caB_2)$ is spanned by elements of the form $(a \otimes \sigma_X) \otimes ( b \otimes \I)$ and $(a \otimes \I) \otimes (b \otimes \sigma_X)$. They are mapped by the *-homormorphism $g \circ f$ to
	\begin{align*}
		(g \circ f) \big( (a \otimes \sigma_X) \otimes (b \otimes \I) \big) &= (a \otimes b) \otimes \sigma_Y, \\
		(g \circ f) \big( (a \otimes \sigma_X) \otimes (b \otimes \sigma_X) \big) &= (a \otimes b) \otimes (\sigma_X).
	\end{align*}
	Finally, the grading operator is mapped to
	\begin{equation*}
		(g \circ f) \big( (\I \otimes \sigma_Z) \otimes (\I \otimes \sigma_Z) \big) = \I \otimes \sigma_Z.
	\end{equation*}

	We see then that the image of $\caB_1 \gotimes \caB_2$ under $g \circ f \circ \pi$ is the graded algebra $\caM^{mn | mn}$. In fact, $g \circ f \circ \pi$ provides an isomorphism between $\caB_1 \gotimes \caB_2$ and $\caM^{mn | mn}$ because the dimensions of these algebras are equal.

	We now verify the composition law for $\ind$. The fact that $\ind_1(\caB_1 \gotimes \caB_2) = \ind_1(\caB_1) \ind_1(\caB_2)$ is just multiplicativity of dimension under the graded tensor product.

	The even automorphic $G$-actions $\rho^g_1$, $\rho^g_2$ on $\caB_1$, $\caB_2$ can be extended to even automorphic $G$-actions $\tilde \rho^g_1$, $\tilde \rho^g_2$ acting on the whole of $\caM^m \otimes \caM^2 \supset \pi(\caB_1)$ and $\caM^n \otimes \caM^2 \supset \pi(\caB_2)$ respectively, by setting $\tilde \rho^g_i(\I \otimes \sigma_Z) = (-1)^{\zeta_i(g)} (\I \otimes \sigma_Z)$ for $i = 1, 2$. The automorphisms $\tilde \rho^g_i$ then act on full matrix algebras, so there are unitaries $V_1(g) \in \caM^m \otimes \caM^2$ and $V_2(g) \in \caM^n \otimes \caM^2$ such that $\tilde \rho^g_i = \Ad_{V_i(g)}$ for $i = 1, 2$.
	
	The automorphisms $\rho^g_i$, and therefore also their extensions $\tilde \rho^g_i$, leave the even subalgebras $\caM^m \otimes \I$, $\caM^n \otimes \I$ invariant, so there are unitaries $V_1^{(0)}(g) \in \caM^m$ and $V_2^{(0)}(g) \in \caM^n$ such that $\tilde \rho_i^g(a \otimes \I) = \Ad_{V_i^{(0)}}(a) \otimes \I$. The unitaries $V_i(g)^* (V_i^{(0)}(g) \otimes \sigma_Y^{\zeta_i})$ commute with the even subalgebras, $\caM^m \otimes \I$ for $i = 1$ and $\caM^n \otimes \I$ for $i = 2$, and with the elements $\I \otimes \sigma_X$ and $\I \otimes \sigma_Z$ (recall that, by definition, $\rho_i^g(\I \otimes \sigma_X) = (-1)^{\zeta_i(g)} (\I \otimes \sigma_X)$). These together generate the whole of $\caM^m \otimes \caM^2$, $\caM^n \otimes \caM^2$ so, after possibily redefining the phases of the $V_i^{(0)}(g)$, we have
	\begin{equation}
		V_i(g) = V_i^{(0)}(g) \otimes \sigma_Y^{\zeta_i(g)}
	\end{equation}
	for $i = 1, 2$.

	By Lemma \ref{lem:action of graded tensor product}, the automorphisms $\tilde \rho^g = \tilde \rho^g_1 \gotimes \tilde \rho^g_2$ act on $(\caM^m \otimes \caM^2) \otimes (\caM^n \otimes \caM^2) \supset \pi(\caB_1 \gotimes \caB_2)$ through conjugation by
	\begin{align*}
		V_1(g) \otimes V_2(g)(\I \otimes \sigma_Z)^{\zeta_1(g)} &= (V_1^{(0)}(g) \otimes \sigma_Y^{\zeta_1(g)}) \otimes (V_2^{(0)}(g) \otimes \sigma_Y^{\zeta_2(g)} \sigma_Z^{\zeta_1(g)}) \\
									&= (\iu)^{\zeta_1(g) + \zeta_2(g)}  (V_1^{(0)}(g) \otimes (\sigma_X \sigma_Z)^{\zeta_1(g)}) \otimes (V_2^{(0)}(g) \otimes \sigma_X^{\zeta_2(g)} \sigma_Z^{\zeta_1(g) + \zeta_2(g)})
	\end{align*}
	and so $f \circ \tilde \rho^g \circ f^{-1}$ acts through conjugation with
	\begin{equation}
		V_1^{(0)}(g) \otimes \sigma_X^{\zeta_1(g)} \otimes V_2^{(0)}(g) \otimes ( \sigma_X^{\zeta_1(g)} \sigma_Y^{\zeta_2(g)} )
	\end{equation}
	and $g \circ f \circ \tilde \rho^g \circ f^{-1} \circ g^{-1}$ acts on $\caM^{mn | mn}$ through conjugation by
	\begin{equation}
		V(g) = \big( V_1^{(0)}(g) \otimes V_2^{(0)}(g) \big) \otimes (\sigma_X^{\zeta_1(g)} \sigma_Y^{\zeta_2(g)}).
	\end{equation}
	These unitaries form a projective representation whose associated second cohomology class is by definition $\ind_3(\caB_1 \gotimes \caB_2, \rho^g_1 \gotimes \rho^g_2)$. We have
	\begin{align*}
		V(g) V(h) &= \big( V_1^{(0)}(g) V_1^{(0)}(h) \otimes V_2^{(0)}(g) V_2^{(0)}(h) \big) \otimes (\sigma_X^{\zeta_1(g)} \sigma_Y^{\zeta_2(g)} \sigma_X^{\zeta_1(h)} \sigma_Y^{\zeta_2(h)}) \\
			  &= \nu_1(g, h) \nu_2(g, h)  \big( V_1^{(0)}(gh) \otimes V_2^{(0)}(gh) \big) \otimes ( (-1)^{\zeta_1(g) \zeta_2(h)} \sigma_X^{\zeta_1(g) + \zeta_1(h)} \sigma_Y^{\zeta_2(g) + \zeta_2(h)}) \\
			  &= \nu_1(g, h) \nu_2(g, h) \nu(\zeta_1, \zeta_2)(g, h) V(gh),
	\end{align*}
	so
	\begin{equation}
		\ind_3(\caB_1 \gotimes \caB_2, \rho_1 \gotimes \rho_2) = [\nu_1 \cdot \nu_2 \cdot \nu(\zeta_1, \zeta_2)],
	\end{equation}
	as required.

	Finally, the action of $\Ad_{V(g)}$ on the grading operator $\I \otimes \sigma_Z$ is
	\begin{equation}
		\Ad_{V(g)}(\I \otimes \sigma_Z) = (-1)^{\zeta_1(g) + \zeta_2(g)} (\I \otimes \sigma_Z)
	\end{equation}
	so
	\begin{equation}
		\ind_2(\caB_1 \gotimes \caB_2, \rho_1 \gotimes \rho_2) = \zeta_1 + \zeta_2.
	\end{equation}
\end{proof}

\subsubsection{Isomorphism of central simple \texorpdfstring{$G$}{G}-systems}

\begin{definition} \label{def:isomorphism of central simple G-systems}
	Two central simple $G$-systems $(\caB_1, \rho_1)$ and $(\caB_2, \rho_2)$ are isomorphic if there is a super *-algebra isomorphism $\Phi : \caB_1 \rightarrow \caB_2$ such that $\Phi \circ \rho_1 = \rho_2 \circ \Phi$.
\end{definition}

\begin{lemma} \label{lem:isomorphic invariance of ind}
	If $(\caB_1, \rho_1)$ and $(\caB_2, \rho_2)$ are isomorphic central simple G-systems then
	\begin{equation}
		\ind(\caB_1, \rho_1) = \ind(\caB_2, \rho_2).
	\end{equation}
\end{lemma}

\begin{proof}
	Let $\ind(\caB_i, \rho_i) = (d_i, \zeta_i, \bnu_i)$ for $i = 1, 2$. The equality $d_1 = d_2$ is simply the fact that isomorpic algebras have the same dimension. In particular, either $\caB_1$ and $\caB_2$ are both rational, or they are both radical.

	Let $\Phi : \caB_1 \rightarrow \caB_2$ implement the isomorphism. To show that $\zeta_1 = \zeta_2$ and $\bnu_1 = \bnu_2$ we work case by case:

	\textbf{Case I : $\caB_1$ and $\caB_2$ are rational.} Then there are projective representations $g \mapsto V_i(g) \in \caB_i$ such that $\rho_i^g = \Ad_{V_i(g)}$ and $\bnu_i = [V_i]$. Since $\Phi \circ \rho_1 = \rho_2 \circ \Phi$ we have that $\rho_2 = \Ad_{\Phi(V_1(g))}$. It follows that $V_2(g) = \mu(g) \Phi(V_1(g))$ for some phases $\mu(g) \in U(1)$ so if
	\begin{equation}
		V_1(g) V_1(h) = \nu_1(g, h) V_1(gh)
	\end{equation}
	then
	\begin{equation}
		V_2(g) V_2(h) = \frac{\mu(g) \mu(h)}{\mu(gh)} \nu_1(g, h) V_2(gh)
	\end{equation}
	so the projective representations $V_1$ and $V_2$ have equivalent 2-cocycles. Thus
	\begin{equation}
		\bnu_1 = [V_1] = [V_2] = \bnu_2. 
	\end{equation}

	There are grading operators $\Theta_i \in \caB_i$ with $\Theta_i^2 = \I$. Then $\rho_i^g(\Theta_i) = (-1)_{\zeta_i(g)} \Theta_i$. Using the fact that $\Phi$ respects the grading we can take $\Theta_2 = \Phi(\Theta_1)$. Then using that $\Phi$ intertwines $\rho_1$ and $\rho_2$ we find
	\begin{equation}
		(-1)^{\zeta_2(g)} \Theta_2 = \rho_2^g(\Theta_2) = \rho_2^g(\Phi(\Theta_1)) = \Phi( \rho_1^g(\Theta_1) ) = (-1)^{\zeta_1(g)} \Theta_2
	\end{equation}
	hence $\zeta_1 = \zeta_2$ as required.

	\textbf{Case II : $\caB_1$ and $\caB_2$ are radical.} Then the even subalgebras $\caB^+_i$ are full matrix algebras, invariant under the action of $\rho_i$. As such, there are projective representations $g \mapsto V^{(0)}_i(g) \in \caB^+_i$ such that $\rho_i^g|_{\caB^+_i} = \Ad_{V^{(0)}_i(g)}$ and $\bnu_i = [V^{(0)}_i]$. Since $\Phi \circ \rho_1 = \rho_2 \circ \Phi$ we have that $\rho_2|_{\caB^+_2} = \Ad_{\Phi(V^{(0)}_1(g))}$. It follows that $V^{(0)}_2(g) = \mu(g) \Phi(V^{(0)}_1(g))$ for some phases $\mu(g) \in U(1)$ so if
	\begin{equation}
		V^{(0)}_1(g) V^{(0)}_1(h) = \nu_1(g, h) V^{(0)}_1(gh)
	\end{equation}
	then
	\begin{equation}
		V^{(0)}_2(g) V^{(0)}_2(h) = \frac{\mu(g) \mu(h)}{\mu(gh)} \nu_1(g, h) V^{(0)}_2(gh)
	\end{equation}
	so the projective representations $V^{(0)}_1$ and $V^{(0)}_2$ have equivalent 2-cocycles. Thus
	\begin{equation}
		\bnu_1 = [V^{(0)}_1] = [V^{(0)}_2] = \bnu_2. 
	\end{equation}

	If $\ep_1$ is an odd self-adjoint unitary of $\caB_1$ then $\ep_2 = \Phi(\ep_1)$ is an odd self-adjoint unitary of $\caB_2$ and
	\begin{equation}
		(-1)^{\zeta_2(g)} \ep_2 = \rho_2^g(\ep_2) = \Phi( \rho_2^g(\ep_1)) = (-1)^{\zeta_1(g)} \ep_2
	\end{equation}
	so $\zeta_1 = \zeta_2$.
\end{proof}

\subsubsection{Relative index and group structure}

The index of central simple $G$-systems takes values in $(\N \cup \sqrt{2} \N) \times \Hom(G, \Z_2) \times H^2(G, U(1)$. This set is equipped with a binary operation given by the stacking rule of Proposition \ref{prop:stacking for absolute ind} as
\begin{equation} \label{eq:binary operation}
	(d_1, \zeta_1, \bnu_1) \cdot (d_2, \zeta_2, \bnu_2) = \big( d_1 d_2, \zeta_1 + \zeta_2, \bnu_1 \cdot \bnu_2 \cdot \bnu(\zeta_1, \zeta_2) \big)
\end{equation}
where $\bnu(\zeta_1, \zeta_2) \in H^2(G, U(1))$ is the second cohomology class determined by the 2-cocycle
\begin{equation}
	\nu(\zeta_1, \zeta_2)(g, h) = (-1)^{\zeta_1(g) \zeta_2(g)}.
\end{equation}
This binary operation does not have inverses because of the first component. If we extend the set to $(\Q \cup \sqrt{2} \Q) \times \Hom(G, \Z_2), H^2(G, U(1))$ with the same binary operation, then we do have inverses and we get a group.

\begin{lemma} \label{lem:abelian group}
	The set $(\Q \cup \sqrt{2} \Q) \times \Hom(G, \Z_2, H^2(G, U(1))$ equipped with the binary operation \eqref{eq:binary operation} is an abelian group. The inverse of $(d, \zeta, \bnu)$ is $(d^{-1}, \zeta, \bnu^{-1})$ and the neutral element is $(1, 0, \boe)$, where $\boe$ is the neutral element of $H^2(G, U(1))$.
\end{lemma}

\begin{proof}
	The fact that the binary operation is abelian is obvious, multiplication is abelian, and $\Hom(G, \Z_2)$ and $H^2(G, U(1))$ are both abelian groups. To see that $(1, 0, \boe)$ is the neutral element, compute
	\begin{equation}
		(1, 0, \boe) \cdot (d, \zeta, \bnu) = \big(d, \zeta, \boe \cdot \bnu \cdot \bnu(0, \zeta) \big) = (d, \zeta, \bnu)
	\end{equation}
	where we used the fact that
	\begin{equation}
		\nu(0, \zeta)(g, h) = (-1)^{0 \times \zeta(h)} = 1
	\end{equation}
	is a trivial cocycle, so $\bnu(0, \zeta) = \boe$.

	To see that the inverse of $(d, \zeta, \bnu)$ is $(d^{-1}, \zeta, \bnu^{-1})$, compute
	\begin{equation}
		(d, \zeta, \bnu) \cdot (d^{-1}, \zeta, \bnu^{-1}) = \big( 1, \zeta + \zeta, \bnu \cdot \bnu^{-1} \cdot \bnu(\zeta, \zeta) \big) = (1, 0, \boe)
	\end{equation}
	where we used $\zeta + \zeta = 0$, because addition is modulo 2, and the fact that $\bnu(\zeta, \zeta) = \boe$ for any $\zeta \in \Hom(G, \Z_2)$, see Lemma \ref{lem:triviality of nu zeta zeta}.
\end{proof}

\begin{lemma} \label{lem:triviality of nu zeta zeta}
	For any $\zeta \in \Hom(G, U(1))$, the 2-cocycle $\nu(\zeta, \zeta)$ is cohomologous to the trivial cocycle and so $\bnu(\zeta, \zeta) = \boe$.
\end{lemma}

\begin{proof}
	Let $\mu(g) = \iu^{\zeta(g)}$. Then one easily checks that
	\begin{equation}
		\nu(\zeta, \zeta)(g, h) = \frac{\mu(g) \mu(h)}{\mu(gh)}.
	\end{equation}
\end{proof}

\begin{definition} \label{def:relative ind}
	Let $(\caB_1, \rho_1)$ and $(\caB_2, \rho_2)$ be  a pair of central simple $G$-systems. The relative index of $(\caB_1, \rho_1)$ and $(\caB_2, \rho_2)$ is defined to be
	\begin{equation}
		\ind(\caB_1, \rho_1, \caB_2, \rho_2) := \ind(\caB_1, \rho_1) \cdot \ind(\caB_2, \rho_2)^{-1}.
	\end{equation}
\end{definition}

Like the index, the relative index also behaves nicely under graded tensor products,
\begin{proposition} \label{prop:stacking for relative ind}
	Let $\caA_1, \caA_2, \caB_1, \caB_2$ be central simple superalgebras equipped with even automorphic $G$-actions $\rho_{\caA_1}, \rho_{\caA_2}, \rho_{\caB_1}$ and $\rho_{\caB_2}$ respectively. We have
	\begin{equation}
		\ind(\caA_1 \gotimes \caB_1, \rho_{\caA_1} \gotimes \rho_{\caB_1}, \caA_2 \gotimes \caB_2, \rho_{\caA_2} \gotimes \rho_{\caB_2}) = \ind(\caA_1, \rho|_{\caA_1}, \caA_2, \rho|_{\caA_2}) \cdot \ind(\caB_1, \rho|_{\caB_1}, \caB_2, \rho|_{\caB_2})
	\end{equation}
\end{proposition}

\begin{proof}
	This follows immediately from proposition \ref{prop:stacking for absolute ind}.
\end{proof}

\subsection{Index for QCAs}

Fix a quantum chain $\caA$ with symmetry $\rho$. For any central simple super subalgebras $\caB_1, \caB_2 \subset \caA$ that are invariant under the symmetry, the restriction $\rho|_{\caB_i}$, $i = 1, 2$ is an even automorphic $G$-action on $\caB_i$ so the pair $(\caB_i, \rho|_{\caB_i})$ is a central simple $G$-system. Since the $G$-action is always the symmetry actio of the quantum chain, we will usually drop it from the notation, so that $\caB_i$ refers to the central simple $G$-system $(\caB_i, \rho|_{\caB_i})$. Consequently, we also write
\begin{equation}
\ind(\caB_1) := \ind(\caB_1, \rho_{\caB_1}), \quad \text{and} \quad \ind(\caB_1, \rho|_{\caB_1}, \caB_2, \rho|_{\caB_2})
\end{equation}
for the index and relative index.

\begin{proposition} \label{prop:index for QCA}
	Let $\caA$ be a quantum chain with symmetry $\rho$, and $\al$ an equivariant QCA on $\caA$. Consider any coarse-graining of the quantum chain putting $\al$ in nearest neighbour form, and denote the corresponding on-site algebras by  $\caA_n$ and the support algebras by $\caL_n$, $\caR_n$. All these algebras are central simple superalgebras that are invariant under the symmetry $\rho$. The relative indices $\ind(\caR_n, \caA_{2n+1})$ and $\ind(\caA_{2n}, \caL_n)$ are equal to each other, independent of $n$, and independent of the coarse graining.
\end{proposition}

This allows us to define
\begin{definition} \label{def:index for QCA}
	Let $\caA$ be a quantum chain with symmetry $\rho$, and $\al$ an equivariant QCA on $\caA$. Denote by $\caL_n$, $\caR_n$ the support algebras defined w.r.t. any coarse graining of the quantum chain that brings $\al$ in nearest neighbour form and put
	\begin{equation}
		\ind(\al) := \ind(\caR_n, \caA_{2n+1}).
	\end{equation}
\end{definition}

\begin{proofof}[Proposition \ref{prop:index for QCA}]
	Fix a coarse graining such that $\al$ is nearest neighbour, leading to on-site algebras $\caA_n$ and support algebras $\caL_n, \caR_n$.

	From Theorem \ref{thm:overlap factorization} we have
	\begin{align}
		\caC_n &= \caA_{2n-1} \gotimes \caA_{2n} = \caR_{n-1} \gotimes \caL_n \\
		\caB_n &= \caA_{2n} \gotimes \caA_{2n+1} = \al^{-1}(\caL_n) \gotimes \al^{-1}(\caR_n),
	\end{align}
	so it follows from the stacking law, Proposition \ref{prop:stacking for absolute ind}, and Lemma \ref{lem:isomorphic invariance of ind} that
	\begin{equation}
		\ind(\caA_{2n-1}) \cdot \ind(\caA_{2n}) = \ind(\caR_{n-1}) \cdot \ind(\caL_n)
	\end{equation}
	and
	\begin{equation}
		\ind(\caA_{2n}) \cdot \ind(\caA_{2n + 1}) = \ind(\caL_n) \cdot \ind(\caR_n)
	\end{equation}
	for all $n \in \Z$.

	Rearranging these equations yields
	\begin{equation}
		\ind(\caR_{n-1}, \caA_{2n-1}) = \ind(\caR_{n-1}) \cdot \ind(\caA_{2n - 1})^{-1} = \ind(\caA_{2n}) \cdot \ind(\caL_n) = \ind(\caA_{2n}, \caL_n)
	\end{equation}
	and
	\begin{equation}
		\ind(\caA_{2n}, \caL_n) = \ind(\caA_{2n}) \cdot \ind(\caL_n)^{-1} = \ind(\caR_n) \cdot \ind(\caA_{2n+1}) = \ind(\caR_n, \caA_{2n+1})
	\end{equation}
	from which the required equality of relative indices, and their independence of $n$, follows.

	To see independence of coarse graining, consider coarse grained sites $\widetilde \caA_{n} = \widehat \bigotimes_{m \in I_n} \caA_{m}$ where the $I_n = \{ a_n +1, \cdots, a_{n + 1} \}$ are non-empty intervals that partition $\Z$. Clearly the support algebra $\widetilde \caR_n$ depends only on the coarse grained algebras $\widetilde \caA_{2n}$, $\widetilde \caA_{2n+1}$ and $\caA_{2n+2}$ so $\ind(\widetilde \caR_n, \widetilde \caA_{2+1})$ can be constructed by coarse graining only these three blocks and leaving the rest of the lattice as is was. We can then compute $\ind(\caR_m, \caA_{2m+1})$ anywhere else and since we've just shown that this quantity doesn't depend on where along the chain it is computed we have $\ind(\caR_m, \caA_{2m+1}) = \ind(\widetilde \caR_n, \widetilde \caA_{2n+1})$ as required.
\end{proofof}

\subsection{Basic properties of the index for QCAs}

\subsubsection{Stability}

The following lemma generalizes Proposition 4.10. of \cite{RWW2020}.
\begin{lemma} \label{lem:QCA stability}
	Let $\al_1$ and $\al_2$ be two even equivariant automorphisms on the same quantum chain $\caA$ such that
	\begin{align*}
		\al_i\big( \caB_n \big) &\subset \caC_n \gotimes \caC_{n+1}, \\
		\al_i \big( \caB_{n-1} \big) &\subset \caC_{n-1} \gotimes \caC_n, \\
		\al_i^{-1} \big( \caC_n \big) &\subset \caB_{n-1} \gotimes \caB_n 
	\end{align*}
	for $i = 1, 2$ and some $n \in \Z$. This holds in particular if $\al_1$ and $\al_2$ are even equivariant nearest-neighbour QCAs. 

	If $\norm{(\al_1 - \al_2)|_{\caB_n}} \leq \ep < 1/24$, then the central simple $G$-systems $(\caL_n^{(1)}, \rho|_{\caL_n^{(1)}})$ and $(\caL_n^{(2)}, \rho|_{\caL_n^{(2)}})$ are isomorphic. The isomorphism is implemented by conjugation with an even $G$-invariant unitary $u \in \caC_n$ with $\norm{u - \I} \leq 36 \ep$.
	It follows that $\ind(\al_1) = \ind(\al_2)$.
\end{lemma}

\begin{proof}
	Let $\caL_n^{(i)}, \caR_{n-1}^{(i)}$, $i = 1, 2$ be the overlap algebras. They are central simple superalgebras by Theorem \ref{thm:overlap factorization} and they are invariant under the symmetry $\rho$ of the quantum chain so they define central simple $G$-systems.

	For $x \in \caL^{(1)}_n$, we construct an element $z \in \caL_n^{(2)}$ that is close to $x$. Define $y = (\al_2 \circ \al_1^{-1})(x) \in \al_2(\caB_n)$. Then
	\begin{equation}
		\norm{x - y} = \norm{(\al_1 - \al_2) \big( \al_1^{-1}(x) \big)} \leq \norm{(\al_1 - \al_2)|_{\caB_n}} \leq \ep.
	\end{equation}

	Now put $z = E_{\caC_n}(y) \in \caL_n^{(2)}$ (recall the conditional expectations from Proposition \ref{prop:Araki's conditional expectations}). We have
	\begin{equation}
		\norm{z - y} = \norm{ E_{\caC_n}(y) - y} = \norm{ E_{\caC_n}(y - x) + x - y} \leq 2 \norm{x - y} \leq 2 \ep
	\end{equation}
	where we used $x \in \caL_n^{(1)} \subset \caC_n$. Therefore $\norm{x - z} \leq 3\ep$ and $\caL_n^{(1)} \nsub{3 \ep} \caL_n^{(2)}$. In the same way we find $\caL_n^{(2)} \nsub{3 \ep} \caL_n^{(1)}$.

	Since $3\ep < 1/8$ we can apply Theorem \ref{thm:near inclusions} to conclude that there is an even $G$-invariant unitary $u \in (\caL_n^{(1)} \cup \caL_n^{(2)})'' \subset \caC_n$ such that $\Ad_u : \caL_n^{(1)} \rightarrow \caL_n^{(2)}$ is an isomorphism of *-algebras. Since $u$ is even, $\Ad_u$ is in fact an isomorphism of super *-algebras, and since $u$ is $G$-invariant, $\Ad_u$ is an isomorphism of central simple $G$-systems, as required.

	By Lemma \ref{lem:isomorphic invariance of ind}, we have $\ind(\caL_n^{(1)}) = \ind(\caL_n^{(2)})$ so
	\begin{equation}
		\ind(\al_1) = \ind(\caA_{2n}) \cdot \ind(\caL_n^{(1)})^{-1} = \ind(\caA_{2n}) \cdot \ind(\caL_n^{(2)})^{-1} = \ind(\al_2).
	\end{equation}
	This concludes the proof.
\end{proof}

\begin{proposition} \label{prop:QCA stability}
	Let $I$ be a closed interval and $I \ni t \mapsto \al_t$ a strongly continuous family of even equivariant QCAs of uniformly bounded range. Then $t \mapsto \ind(\al_t)$ is constant along the path.
\end{proposition}

\begin{proof}
	Since the range of the QCAs is uniformly bounded, we can coarse grain the lattice such that $\al_t$ is in nearest-neighbour form for all $t \in I$. By the strong continuity and the fact that $I$ is closed we can find $\ep > 0$ such that $\norm{\al_s(x) - \al_t(x)} < \norm{x} / 24$ for all $x \in \caA$ whenever $\abs{s - t} < \ep$. It then follows from Lemma \ref{lem:QCA stability} that $\ind(\al_s) = \ind(\al_t)$ if $\abs{s - t} < \ep$. It follows that $\ind(\al_t)$ takes the same value for all $t \in I$.
\end{proof}

\subsubsection{composition and stacking}

\begin{proposition} \label{prop:QCA composition and stacking}
	Suppose $\al$ and $\beta$ are even equivariant QCAs defined on quantum chains $\caA$ and $\caA'$ respectively. Then $\al \gotimes \beta$ acting on the stacked quantum chain $\caA \gotimes \caA'$ is again an even equivariant QCA and
	\begin{equation}
		\ind(\al \gotimes \beta) = \ind(\al) \cdot \ind(\beta).
	\end{equation}
	If $\al$ and $\beta$ are defined on the same quantum chain $\caA$ then $\al \circ \beta$ is also an even equivariant QCA and
	\begin{equation}
		\ind(\al \circ \beta) = \ind(\al) \cdot \ind(\beta).
	\end{equation}

	Moreover, $\ind(\id) = (1, 0, \boe)$ is always trivial and $\ind(\al^{-1}) = \ind(\al)^{-1}$.
\end{proposition}

\begin{proof}

	Let us first show that $\ind(\id)$ is trivial. $\id$ is always in earest-neighbour for, its left overlap algebras are $\caL_n = \id(\caB_n) \cap \caC_n = \caA_{2n}$ so
	\begin{equation}
		\ind(\id) = \ind(\caA_{2n}, \caA_{2n}) = \ind(\caA_{2n}) \cdot \ind(\caA_{2n})^{-1} = (1, 0, \boe).
	\end{equation}

	Let us now show the claim about $\al \gotimes \beta$. Coarse grain both quantum chains such that $\al$ and $\beta$ are in nearest-neighbour form and denote by $\caL_n$ and $\caL'_n$ the overlap algebras corresponding to $\al$ and $\beta$ respectively. The stacking $\al \gotimes \beta$ is also in nearest-neighbour form and the overlap algeras of $\al \gotimes \beta$ are $\caL_n \gotimes \caL'_n$. Then, using Proposition \ref{prop:stacking for relative ind} we find
	\begin{equation}
		\ind(\al \gotimes \beta) = \ind(\caA_{2n} \gotimes \caA'_{2n}, \caL_n \gotimes \caL'_n) = \ind(\caA_{2n}, \caL_n) \cdot \ind(\caA'_{2n}, \caL'_n) = \ind(\al) \cdot \ind(\beta)
	\end{equation}
	as required.

	Now suppose $\al$ and $\beta$ are defined on the same quantum chain $\caA$. By the stacking rule we have $\ind(\al) \cdot \ind(\beta) = \ind(\al \gotimes \beta)$. We will construct a strongly continuous path of even equivariant QCAs of uniformly bounded range interpolating between $\al \gotimes \beta$ and $(\al \circ \beta) \gotimes \id$. With such a path in hand it follows from Proposition \ref{prop:QCA stability}, the stacking rule, and triviality of $\ind(\id)$ that $\ind(\al \gotimes \beta) = \ind \big( (\al \circ \beta) \gotimes \id \big) = \ind( \al \circ \beta)$. Putting these equalities together yields the required $\ind(\al) \cdot \ind(\beta) = \ind(\al \circ \beta)$. To construct the required interpolating path, consider on each site $n$ the swap operation $\Phi_{sw, n}$ acting on the local algebra $\caA_n \gotimes \caA_n$ of the stacked chain as $\Phi_{sw, n} (a \gotimes b) = b \gotimes a$. The swap operations are even equivariant *-automorphisms acting on rational central simple $G$-systems $\caA_n \gotimes \caA_n$. One easily checks that the $\Phi_{sw, n}$ have trivial first cohomology index (Definition \ref{def:first cohomology index}) and so, by Lemma \ref{lem:0-dimensional deform to identity} they can be continuously deformed to the identity along a path $[0, 1] \ni t \mapsto \Phi_{sw, n}(t)$ of even equivariant automorphisms of $\caA_n \gotimes \caA_n$. The simultaneous action of all the $\Phi_{sw, n}(t)$ defines an even equivariant QCA of range 0 which we denote by $\Phi_{sw}(t)$. The family $[0, 1] \ni t \mapsto \Phi_{sw}(t)$ is strongly continuous. Now consider the strongly continuous path of even equivariant nearest neighbour QCAs
	\begin{equation}
		\Psi_t = (\al \gotimes \id) \circ \Phi_{sw}(t) \circ (\id \gotimes \beta) \circ \Phi_{sw}(t)^{-1}.
	\end{equation}
	It interpolates between $\Psi_0 = \al \gotimes \beta$ and $\Psi_1 = (\al \circ \beta) \gotimes \id$, as required. We conclude that $\ind(\al \circ \beta) = \ind(\al) \cdot \id(\beta)$.

	Finally, so see that $\ind(\al^{-1}) = \ind(\al)^{-1}$ we simply note that
	\begin{equation}
		\ind (\id) = \ind(\al \circ \al^{-1}) = \ind(\al) \cdot \ind(\al^{-1})
	\end{equation}
	is trivial, so $\ind(\al^{-1}) = \ind(\al)^{-1}$ follows.
\end{proof}

\subsection{Completeness}

We will prove that the index is a complete invariant up to stable G-equivalence, \ie that two even equivariant QCAs are stably $G$-equivalent if and only if they have the same index. We will require the notion of \emph{stable} equivalence to prove this result, in our construction of interpolating paths between QCAs with the same index we will need to stack several kinds of auxiliary quantum chains. We begin by defining these auxiliary quantum chains and proving some of their properties.

\subsubsection{Auxiliary chains}

\begin{definition}
	Given a quantum chain $\caA$ with symmetry $\rho$ that acts as $\rho|_{\caA_n} = \Ad_{u_n(g)}$ for some projective representations $u_n$, we define its conjugate $\bar \caA$ to be the quantum chain whose local algebras are the same (as superalgebras) as those of $\caA$, and whose symmetry $\bar \rho$ acts as $\bar \rho|_{\bar \caA_n} = \Ad_{\bar u_n(g)}$ where $\bar u_n$ is the representation conjugate to $u_n$.
\end{definition}

\begin{definition} \label{def:trivial chain}
	A quantum chain $\caA$ is called trivial if $\ind(\caA_n) = (d_n, 0, \boe)$ has trivial second and third components for all $n \in \Z$.
\end{definition}

\begin{lemma} \label{lem:conjugate stacking}
	If $\caA$ is a quantum chain and $\bar \caA$ is its conjugate, then $\caA \gotimes \bar \caA$ is a trivial quantum chain.
\end{lemma}

\begin{proof}
	Let $\ind(\caA_n) = (d, \zeta, \bnu = [u_n])$, then $\ind(\bar \caA_n) = (d, \zeta, \bnu^{-1} = [\bar u_n])$. Indeed, the first and second component are obvious, and $\zeta(g)$ is simply the parity of $u_n(g)$, which is the same as the parity of $\bar u_n(g)$. Using Proposition \ref{prop:stacking for absolute ind} we then get
	\begin{equation}
		\ind(\caA_n \gotimes \bar \caA_n) = \big(d^2, \zeta + \zeta, \bnu \cdot \bnu^{-1} \cdot \bnu(\zeta, \zeta) \big) = (d^2, 0, \boe)
	\end{equation}
	where we used that $\bnu(\zeta, \zeta) = \boe$, \cf Lemma \ref{lem:triviality of nu zeta zeta}.
\end{proof}

The \emph{regular representation} of $G$ is a $\abs{G}$-dimensional unitary representation that acts on the finite dimensional Hilbert space $\caH^{\reg}$ spanned by orthonormal vectors $\{ | g \rangle \}_{g \in G}$. The representation is given by $g \mapsto u^{\reg}(g) \in \caU(\caH^{\reg})$ which act as $u^{\reg}(g) | h \rangle = | gh \rangle$ for all $g, h \in G$.

\begin{definition}
	The regular chain is the quantum chain $\caA^{\reg}$ with on-site algebras $\caA^{reg}_n \simeq \End(\caH^{\reg})$ for all $n \in \Z$, with trivial grading, and symmetry $\rho^{\reg}$ acting on the on-site algebras by conjugation with the regular representation, $(\rho^{\reg})^g|_{\caA_n^{\reg}} = \Ad_{u^{\reg}(g)}$.
\end{definition}

Stacking with the regular chain will be used to turn two unitary $G$-representations of the same dimension into isomorphic representations:
\begin{lemma} \label{lem:make reps isomorphic}
	Let $v_1 : G \rightarrow \caB_1$ and $v_2 : G \rightarrow \caB_2$ be unitary representations of $G$ of the same dimension $d$, where $\caB_i$ is isomorphic to the full matrix algebra $\caM^d$ for $i = 1, 2$. Then the representations $v_1 \otimes u^{\reg} : G \rightarrow \caB_1 \otimes \End(\caH^{\reg})$ and $v_2 \otimes u^{reg} : G \rightarrow \caB_2 \otimes \End(\caH^{\reg})$ are isomorphic, \ie there is a *-isomorphism $\Phi : \caB_1 \gotimes \End(\caH^{\reg}) \rightarrow \caB_2 \otimes \End(\caH^{\reg})$ such that $\Phi(v_1) = v_2$.
\end{lemma}

\begin{proof}
	Two unitary representations of finite groups are isomorpic if and only if their characters are the same, cf. Corollary 2, chapter 2 of \cite{Serre1977}. We have
	\begin{equation}
		\chi_i(g) = \Tr \lbrace v_i(g) \otimes u^{\reg}(g) \rbrace = \Tr \lbrace v_i(g) \rbrace \Tr \lbrace u^{\reg}(g) \rbrace = \delta_{g, e} d \abs{G}
	\end{equation}
	where we used
	\begin{equation}
		\Tr \lbrace u^{\reg}(g) \rbrace = \sum_{h \in G} \langle h, u^{\reg}(g) \, h \rangle = \sum_{h \in G} \langle h, gh \rangle = \delta_{g, e} \abs{G}.
	\end{equation}
	We see that the characters $\chi_1$ and $\chi_2$ are the same, so the representations $v_1 \otimes u^{\reg}$ and $v_2 \otimes u^{\reg}$ are isomorphic.
\end{proof}

\begin{definition} \label{def:fermion chain}
	The fermion chain $\caA^{\fer}$ is the quantum chain with on-site algebras $\caA_n \simeq \caM^{1 | 1}$ and trivial symmetry.
\end{definition}
Stacking with the fermion chain will be used to turn two rational central simple superalgebras of the same dimension into isomorphic rational central simple superalgebras.

Using auxiliary degrees of freedom from the conjugate, regular and fermion chains, we will produce pairs of isomorphic rational central simple $G$-systems the we will call \emph{balanced}.

\begin{definition} \label{def:balanced central simple G-system}
	A central simple $G$-system $(\caB, \rho)$ is called balanced if $\caB \simeq \caM^{d | d}$ and $\rho^g = \Ad_{V(g)}$ where $V(g) = E_+ V(g) E_+ + E_- V(g) E_-$ is a unitary representation of $G$ that is diagonal with respect to the grading operator $\Theta = E_+ - E_-$ and is such that $g \mapsto E_{\pm} V(g) E_{\pm}$ are isomorpic representations. We then denote $V(g) = v(g) \oplus v(g)$ where the direct sum decomposition is the one induced by the grading $\Theta = \I \oplus -\I$, and the isiomorphism of the representations allows us to choose a basis in which the diagonal blocks are identical.
\end{definition}

\begin{lemma} \label{lem:fermion stacking I}
	Let $(\caB \simeq \caM^{p | q}, \rho)$ be a central simple $G$-system such that $\rho^g = \Ad_{V(g)}$ with $V(g)$ even for all $g \in G$. Then $\big( \caB \gotimes \caM^{1 | 1} \simeq \caM^{p+q | p+q}, \rho \gotimes \id \big)$ is a balanced central simple $G$-system whose $G$-action is given (with respect to a suitable basis) by conjugation with $\big( V(g) \oplus V(g) \big)$.
\end{lemma}

\begin{proof}
	Write $V(g) = v_p(g) \oplus v_q(g)$ where the direct sum decomposition is induced by the grading operator $\Theta_{\caB} = \I_p - \I_q$. The $G$-action $\rho^g \gotimes \id$ is given by conjugation with $\big( v_p(g) \oplus v_q(g) \big) \gotimes \I_2$. In the representation of the graded tensor product $\pi (a \gotimes b) = a \Theta_{\caB}^{\tau(b)} \otimes b$ we get
	\begin{equation}
		\pi \big( \big( v_p(g) \oplus v_q(g) \big) \gotimes \I_2 \big) = \begin{bmatrix} v_p &  &  &  \\   &  v_q &  &  \\ &  &  v_p  &  \\ &  &  &  v_q \end{bmatrix}
	\end{equation}
	while the grading operator becomes
	\begin{equation}
		\pi(\Theta_{\caB} \gotimes \sigma_Z) = \begin{bmatrix}  \I_p &  &  &  \\ &  -\I_q &  &  \\ &  & -\I_p &  \\ &  &  &  \I_q   \end{bmatrix}
	\end{equation}
	Writing this grading operator as a difference of its spectral projections, $\pi(\Theta_{\caB} \gotimes \sigma_Z) = E_+ - E_1$ we see that
	\begin{equation}
		E_+ \pi \big( \big( v_p(g) \oplus v_q(g) \big) \gotimes \I_2 \big) E_+ = \begin{bmatrix} v_p &  &  &  \\   &  0 &  &  \\ &  &  0  &  \\ &  &  &  v_q \end{bmatrix}
	\end{equation}
	and
	\begin{equation}
		E_- \pi \big( \big( v_p(g) \oplus v_q(g) \big) \gotimes \I_2 \big) E_- = \begin{bmatrix} 0 &  &  &  \\   &  v_q &  &  \\ &  &  v_p  &  \\ &  &  &  0 \end{bmatrix}
	\end{equation}
	which, after an appropriate change of basis, can be put in the desired form.

\end{proof}

\begin{lemma} \label{lem:fermion stacking II}
	Let $(\caB_1, \rho_1)$ and $(\caB_2, \rho_2)$ be balanced central simple $G$-systems with $\caB_1 \simeq \caB_2 \simeq \caM^{d | d}$ and $\rho_{i}^g = \Ad_{v_i(g) \oplus v_i(g)}$ with $v_1, v_2$ isomorphic unitary representations. Then $(\caB_1, \rho_1)$ and $(\caB_2, \rho_2)$ are isomorphic central simple $G$-systems (definition \ref{def:isomorphism of central simple G-systems}). Moreover, the isomorphism $\Phi : \caB_1 \rightarrow \caB_2$ satisfies $\Phi(v_1(g) \oplus v_1(g)) = v_2(g) \oplus v_2(g)$ for all $g \in G$.
\end{lemma}

\begin{proof}
	Choose grading operators $\Theta_1$ and $\Theta_2$ of $\caB_1$ and $\caB_2$ respectively (they are unique up to a sign) such that $\Tr \Theta_i = 0$. Each $\Theta_i$ generates a unitary representation of $\Z_2 = \{e, \theta\}$ (because $\Theta_i^2 = \I$). By definition of a balanced central symple $G$-system, we have $[v_i(g)\oplus v_i(g), \Theta_i] = 0$ for $i = 1, 2$, $g \in G$. Extend the representations $v_i \oplus v_i$ of $G$ to representations $V_i$ of $\Z_2 \times G$ by setting $V_i(\theta) = \Theta_i$ and $V_i(g) = v_i(g) \oplus v_i(g)$ for all $g \in G$. It follows that $V_i(\theta g) = V_i(g \theta) = \Theta_i \big( v_i(g) \oplus v_i(g) \big)$. Since all elements of $\Z_2 \times G$ are either in $G$ or can be written as $\theta g$ for some $g \in G$, this determines all values of the representations $V_i$. We now show that the representations $V_1$ and $V_2$ are isomorphic. Since they have the same dimension, it is sufficient to show that they have the same characters $X_1 = X_2$. The representations $v_1$ and $v_2$ are isomorphic by assumption and thus have the same character, which we denote by $\chi$. We have
	\begin{align*}
		X_i(g) &= \Tr V_i(g) = \Tr v_i(g) \oplus v_i(g) = 2 \chi(g) \\
		X_i(\theta g) &= \Tr \lbrace \Theta_i \big(v_i(g) \oplus v_i(g) \big) \rbrace = 0
	\end{align*}
	which does not depend on $i = 1, 2$, \ie the characters $X_1$ and $X_2$ are equal. This implies that there is a *-isomorphism $\Phi : \caB_1 \rightarrow \caB_2$ such that $\Phi(V_1) = V_2$. In particular $\Phi \big( v_1(g) \oplus v_1(g) \big) = v_2(g) \oplus v_2(g)$, so $\Phi \circ \rho_1 = \rho_2 \circ \Phi$. Moreover, $\Phi(\Theta_1) = \Theta_2$ so $\Phi$ is in fact a super *-isomorphism. This shows that the central simple $G$-systems $(\caB_1, \rho_1)$ and $(\caB_2, \rho_2)$ are isomorphic in the sense of Definition \ref{def:isomorphism of central simple G-systems}.
\end{proof}

\subsubsection{0-dimensional systems} \label{sec:0-dimensional}

In our analysis of even equivariant QCA we will want to deform even equivariant automorphisms of local matrix algebras to the identity. This can not always be done in an even, equivariant way. The obstruction to doing so is an index taking values in $H^1(\Z^2 \times G, U(1))$ where the $\Z_2$ comes from the grading, which for the problem treated here is most conveniently considered as a symmetry.

Let $(\caA, \rho)$ be a rational central simple $G$-system and $\al$ an even equivariant automorphism of $\caA$. Since $\caA$ considered as an ungraded algebra is isomorphic to a full matrix algebra, the symmetry acts as $\rho^g = \Ad_{v(g)}$ for a projective representation $v$ of $G$. Denote further by $\Theta$ the grading operator. Let $G' = \Z_2 \times G$ and denote by $\theta \in G'$ the generator of the $\Z_2$ factor, corresponding to the grading. All elements of $G'$ are of the form $g$ or $\theta g$ for $g \in G$. Define $V : G' \rightarrow \caA$ by $V(g) = v(g)$ and $V(\theta g) = \Theta v(g)$ for all $g \in G$. Then $V$ is a projective representation of $G'$. Denote its 2-cocycle by $\nu$. Since $\al$ is even and equivariant there are phases $\mu(g')$ such that $\al \big(V(g') \big) = \mu(g') V(g')$ for all $g' \in G'$. Moreover,
\begin{align*}
	\mu(g') \mu(h') V(g'h') &= \nu(g', h')^{-1} \mu(g') \mu(h') V(g') V(h') \\ 
				&= \nu(g', h')^{-1} \al \big( V(g') \big) \al \big( V(h') \big) = \al \big( V(g' h') \big) = \mu(g'h') V(g' h')
\end{align*}
so $\mu \in H^1(G', U(1))$.

\begin{definition} \label{def:first cohomology index}
	The first cohomology index of $\al$ is $\mu$.
\end{definition}

\begin{lemma} \label{lem:multiplicativity of the first cohomology index under composition}
	The first cohomology index is multiplicative under composition.
\end{lemma}

\begin{proof}
	Let $\al$ have first chomology index $\mu$ and let $\beta$ have first cohomology index $\nu$. Then $(\al \circ \beta)(V(g)) = \mu(g) \nu(g) V(g)$, so the first cohomology index of $\al \circ \beta$ is $\mu \cdot \nu$, the product of the first cohomology indices of $\al$ and $\beta$.
\end{proof}

\begin{lemma} \label{lem:first cohomology index under stacking}
	If $(\caA_1, \rho_1)$ and $(\caA_2, \rho_2)$ are rational central simple $G$-systems acted upon by even equivariant automorphisms $\al_1$ and $\al_2$ respectively, and $\mu_1, \mu_2$ are their first cohomology indices, then the first cohomology index $\mu$ of $\al_1 \gotimes \al_2$ acting on the rational central simple $G$-system $(\caA_1 \gotimes \caA_2, \rho_1 \gotimes \rho_2)$ is given by
	\begin{align*}
		\mu(g) = &\mu_1(g) \mu_2(g) \mu_1(\theta)^{\zeta_2(g)} \mu_2(\theta)^{\zeta_1(g)} \quad \text{for all} \,\,\, g \in G \\
		\mu(\theta) &= \mu_1(\theta) \mu_2(\theta),
	\end{align*}
	where $\zeta_i \in \Hom(G, \Z_2)$ is determined by $\rho_i^g(\Theta_i) = (-1)^{\zeta_i(g)} \Theta_i$ with $\Theta_i$ the grading operator of $\caA_i$.
\end{lemma}

\begin{proof}
	Suppose $\al_i$ is given by cojugation with $U_i \in \caU(\caA_i)$, and denote by $V_i$ the projective representations of $G'$ associated to $(\caA_i, \rho_i)$ for $i = 1, 2$. By Lemma \ref{lem:action of graded tensor product}, $\al_1 \gotimes \al_2$ is given by conjugation with $U = U_1 \Theta_1^{\mu_2(\theta)} \gotimes U_2 \Theta_2^{\mu_1(\theta)}$. Similarly, the projective representation $V$ of $G'$ associated to $(\caA_1 \gotimes \caA_2, \rho_1 \gotimes \rho_2)$ is given by
	\begin{align*}
		V(g) &= V_1(g) \Theta_1^{\zeta_2(g)} \gotimes V_2(g) \Theta_2^{\zeta_1(g)},  \\
		V(\theta g) &= \Theta_1 V_1(g) \Theta_1^{\zeta_2(g)} \gotimes \Theta_2 V_2(g) \Theta_2^{\zeta_1(g)}
	\end{align*}
	for all $g \in G$. The value of $\mu(g)$ then follows from
	\begin{equation}
		\Ad_U \big( V(g) \big) = \mu_1(g) \mu_2(g) \mu_1(\theta)^{\zeta_2(g)} \mu_2(\theta)^{\zeta_1(g)} \, V(g)
	\end{equation}
	and the value of $\mu(\theta)$ follows from
	\begin{equation}
		\Ad_U \big( V(\theta) \big) = \mu_1(\theta) \mu_2(\theta) \, V(\theta).
	\end{equation}
\end{proof}

\begin{lemma} \label{lem:stacking with inverse yields trivial first cohomology}
	Let $\al$ be an even equivariant automorphism of a rational central simple $G$-system $(\caA, \rho)$. Then the even equivariant automorphism $\al \gotimes \al^{-1}$ of the rational central simple $G$-systems $(\caA \gotimes \caA, \rho \gotimes \rho)$ has trivial first cohomology index.	
\end{lemma}

\begin{proof}
	If $\al$ as first chomology index $\mu$ and $V$ is the projective representation of $G' = \Z_2 \times G$ associated to $(\caA, \rho)$ then $\al^{-1} \big( V(g') \big) = \mu(g')^{-1} V(g')$ for all $g' \in G'$. We see that the first cohomology index of $\al^{-1}$ is $g' \mapsto \mu(g')^{-1}$, the inverse of $\mu$.

	Using Lemma \ref{lem:first cohomology index under stacking} we then find that the first cohomology index of $\al \gotimes \al^{-1}$ is given by
	\begin{align*}
		g &\mapsto \mu(g) \mu(g)^{-1} \mu(\theta)^{\zeta(g)} \mu(\theta)^{-\zeta(g)} = 1, \quad \text{for all} \,\,\, g \in G, \\
		\theta &\mapsto \mu(\theta)\mu(\theta)^{-1} = 1
	\end{align*}
	where $\rho^g(\Theta) = (-1)^{\zeta(g)} \Theta$ with $\Theta$ the grading operator of $\caA$. We find that the first cohomology index of $\al \gotimes \al^{-1}$ is trivial as required.
\end{proof}

\begin{lemma} \label{lem:0-dimensional deform to identity}
	If $\al$ has trivial first cohomology index, then $\al$ can be continuously deformed to the identity along a path of even equivariant automorphisms.
\end{lemma}

\begin{proof}
	Let $\al = \Ad_U$. The triviality of the first cohomology index means that $U = v(g)^* U v(g)$ for all $g \in G'$, \ie $U$ is $G$-invariant and even. It follows that $U = \ed^{\iu H}$ for a $G$-invariant even self-adjoint $H$. The continuous path $t \mapsto U_t = \ed^{\iu t H}$ of $G$-invariant even unitaries interpolates between $\I$ and $U$, so the path $t \mapsto \al_t = \Ad_{U_t}$ continuously interpolates between $\id$ and $\al$ through equivariant even automorphisms.
\end{proof}

\begin{proposition} \label{prop:classification of 0-dimensional systems}
	Two even equivariant automorphisms $\al$ and $\beta$ on the same rational central simple $G$-system $(\caA, \rho)$ can be continuously deformed into each other along a path of even equivariant automorphisms if and only if they have the same first cohomology index.
\end{proposition}

\begin{proof}
	By Lemma \ref{lem:multiplicativity of the first cohomology index under composition}, $\al \circ \beta^{-1}$ has trivial first cohomology index. By Lemma \ref{lem:0-dimensional deform to identity} there is a continuous path $[0, 1] \ni t \mapsto \gamma_t$ of even equivariant automorphisms such that $\gamma_0 = \id$ and $\gamma_1 = \al \circ \beta^{-1}$. Then the path $t \mapsto \gamma_t \circ \beta$ is a continuous path of even equivariant automorphisms interpolating from $\beta$ to $\al$, as required.

	If $\al$ and $\beta$ have different first cohomology index, then they cannot be deformed into each other. Indeed, suppose $[0, 1] \ni t \mapsto \gamma_t$ interpolates continuously between $\al$ and $\beta$ and let $\gamma_t(V(g')) = \mu_t(g') V(g')$ with $V$ the projective representation of $G'$ associated to $(\caA, \rho)$. Since $t \mapsto \gamma_t$ is continuous, the $t \mapsto \mu_t(g')$ are continuous for all $g' \in G'$. But $H^1(G', U(1))$ is a discrete set. It follows that the $\mu_t(g')$ are all constant along the path, so that $\al$ and $\beta$ have the same first cohomology index.
\end{proof}

The following lemma provides automorphisms with arbitrary first cohomology indices:

\begin{lemma} \label{lem:regular rep yields all first cohomology elements}
	For each $\mu \in H^1(G, U(1))$ there is a unitary $u \in \caU(\caH^{\reg})$ such that $\Ad_u \big( u^{\reg}(g) \big) = \mu(g) u^{\reg}(g)$ for all $g \in G$.
\end{lemma}

\begin{proof}
	Simply take $u | g \rangle = \mu(g) | g \rangle$ for all $g \in G$, then
	\begin{equation}
		\Ad_{u} \big( r^{\reg}(g) \big) | h \rangle = u u^{\reg} u^* | h \rangle = \overline{\mu(g)} u u^{\reg}(g) | h \rangle = \overline{\mu(g)} u | gh \rangle = \overline{\mu(h)} \mu(gh) | gh \rangle = \mu(g) u^{\reg}(g) | h \rangle
	\end{equation}
	for all $g, h \in G$, \ie $\Ad_{u} \big( u^{\reg}(g) \big) = \mu(g) u^{\reg}(g)$ as required.
\end{proof}

The following Lemma ensures the existence of product automorphisms with arbitrary associated first cohomology index in various cases.

\begin{lemma} \label{lem:existence of product automorphisms with given first cohomology I}
	Let $(\caA_1, \rho_1)$ and $(\caA_2, \rho_2)$ be balanced central simple $G$-systems (Definition \ref{def:balanced central simple G-system}) such that $\rho_i = \Ad_{v_i \oplus v_i}$ with the $v_i$ direct sums of copies of the regular representation. Then for any $\mu \in H^1(G' = \Z^2 \times G, U(1))$ there exists an even equivariant automorphism on $(\caA_1 \gotimes \caA_2, \rho_1 \gotimes \rho_2)$ of the form $\al_1 \gotimes \al_2$ that has $\mu$ as its first cohomology index.
\end{lemma}

\begin{proof}
	Denote by $\Theta_1$, $\Theta_2$ the grading operators of $\caA_1$ and $\caA_2$. The graded tensor product $\caA_1 \gotimes \caA_2$ is represented on $\caA_2 \gotimes \caA_2$ by $\pi(a_1 \gotimes a_2) = a_1 \Theta_1^{\tau(a_2)} \otimes a_2$ and has in this representation the grading operator $\Theta = \Theta_1 \otimes \Theta_2$.

	Let $\al_i$, $i = 1, 2$ be even equivariant automorphisms on $\caA_i$ given by $\Ad_{U_i}$. Let $\mu_i \in H^1(G' = \Z_2 \times G, U(1))$ be the corresponding first cohomology classes. Denote by $\theta$ the generating element of the $\Z_2$ factor.  By Lemma \ref{lem:first cohomology index under stacking} and the fact that the $v_i \oplus v_i$ are all even, so $\zeta_1 = \zeta_2 \equiv 0$, we get that the first cohomology class associated to $\al_1 \gotimes \al_2$ is simply the product of the first cohomology classes $\mu_1$, $\mu_2$ associated to $\al_1$ and $\al_2$. But we can choose $\al_1$ and $\al_2$ to get any $\mu_1$ and $\mu_2$ we want. Indeed the representations $V_i$ defined by $V_i(g) = v_i(g) \oplus v_i(g)$ for $g \in G$ and $V_i(\theta) = \Theta_i$ are direct sums of the regular representation of $G' = \Z_2 \times G$. To see this it is sufficient to show that their characters are integer multiples of $2 \abs{G} \delta_{e}$, the character of the regular representation of $G'$. This is the case:
	\begin{align}
		\Tr V_i(g) = \Tr v_i(g) \oplus v_i(g) = 2 n_i \abs{G} \delta_e(g) \\
		\Tr V_i(\theta g) = \Tr \Theta_i \big( v_i(g) \oplus v_i(g) \big) = 0
	\end{align}
	for all $g \in G$, where $v_i$ is assumed to be isomorphic to a direct sum of $n_i$ copies of the regular representation. From Lemma \ref{lem:regular rep yields all first cohomology elements} we then get a unitaries $U_i$ such that $\Ad_{U_i} (V_i(g)) = \mu_i(g) V_i(g)$ for all $g \in G'$ and any $\mu_i$ we want. The automorphisms $\al_i = \Ad_{U_i}$ are even and equivariant, and have associated first cohomology indices $\mu_i$. Now choose $\mu_1 = \mu$ and $\mu_2 \equiv 1$ to complete the proof.
\end{proof}

\begin{lemma} \label{lem:existence of product automorphisms with given first cohomology II}
	Let $(\caA_1, \id)$ and $(\caA_2, \id)$ be balanced central simple $G$-systems (Definition \ref{def:balanced central simple G-system}) carrying the trivial representations of $G$. Then for any $\mu \in H^1(G' = \Z^2 \times G, U(1))$ arising as the first cohomology index of an even automorphism $\al$ of $(\caA_1 \gotimes \caA_2, \id)$ we have $\mu(g) = 1$ for all $g \in G$. Moreover, there exists an even equivariant automorphism on $(\caA_1 \gotimes \caA_2, \rho_1 \gotimes \rho_2)$ of the form $\al_1 \gotimes \al_2$ that has $\mu$ as its first cohomology index.
\end{lemma}

\begin{proof}
	The $G$-action of $(\caA_1 \gotimes \caA_2, \id)$ is trivial, hence given by conjugation with the identity for all $G \in G$. Thus $\mu(g) \I = \al(\I) \I^{-1} = \I$ \ie $\mu(g) = 1$. It follows that $\mu$ is completely determined from its value $\mu(\theta) \in \{-1, +1\}$. Let $\al_1$ and $\al_2$ be even automorphisms on $\caA_1$ and $\caA_2$ respectively and denote by $\mu_1$ and $\mu_2$ their first cohomology indices. As before, $\mu_{i}(g) = 1$ for all $g \in G$ and they are completely determined by $\mu_i(\theta) \in \{-1, 1\}$. We can choose the $\al_i$ to get any $\mu_i(\theta) \in \{-1, 1\}$ we want. Indeed, to get $\mu_i(\theta) = 1$ just let $\al_i = \id$. to get $\mu_i(\theta) = -1$ let $\al_i$ be conjugation with the odd unitary
	\begin{equation}
		\begin{bmatrix} & \I \\ \I &  \end{bmatrix}
	\end{equation}
	in a basis where the grading operator of $\caA_i$ is
	\begin{equation}
		\Theta_i = \begin{bmatrix} \I &  \\  &  -\I \end{bmatrix}.
	\end{equation}
	The cohomology index associated to $\al_1 \gotimes \al_2$ is $\mu_1 \cdot \mu_2$, so we can choose $\al_1$ and $\al_2$ to get the required index $\mu$.

\end{proof}

\begin{lemma} \label{lem:existence of product automorphisms with given first cohomology III}
	Let $(\caA_1, \rho_1)$ and $(\caA_2, \rho_2)$ be purely even central simple $G$-systems such that $\rho_i = \Ad_{v_i}$ with the $v_i$ direct sums of copies of the regular representation. Then for any $\mu \in H^1(G, U(1))$ there exists an even equivariant automorphism on $(\caA_1 \gotimes \caA_2, \rho_1 \gotimes \rho_2)$ of the form $\al_1 \gotimes \al_2$ that has $\mu$ as its first cohomology index.
\end{lemma}

\begin{proof}
	By Lemma \ref{lem:regular rep yields all first cohomology elements} we can choose $\al_1$ and $\al_2$ with any associated first cohomology indices $\mu_1, \mu_2 \in H^1(G, U(1))$. (Regarding $\mu_1, \mu_2$ as elements of $H^1(\Z_2 \times G, U(1))$ we have $\mu_1(\theta) = \mu_2(\theta) = 1$ automatically due to the trivial grading, and $\mu_1, \mu_2$ are completely determined by specifying them as elements of $H^1(G, U(1))$). As in the proof of Lemma \ref{lem:existence of product automorphisms with given first cohomology I} one can show that the first cohomology index associated to $\al_1 \gotimes \al_2$ is $\mu_1 \cdot \mu_2$. since we can choose $\al_1$ and $\al_2$ to obtain any $\mu_1$ and any $\mu_2$, we can obtain any $\mu \in H^1(G, U(1))$.
\end{proof}

\begin{lemma} \label{lem:existence of product automorphisms with given first cohomology IV}
Let $(\caA_1, \id)$ and $(\caA_2, \id)$ be purely even central simple $G$-systems carrying trivial representations of $G$. Then the first cohomology index $\mu$ associated to any automorphism of $(\caA_1, \id)$, $(\caA_2, \id)$ or $(\caA_1 \gotimes \caA_2, \rho_1 \gotimes \rho_2)$ is trivial, \ie $\mu \equiv 1$.
\end{lemma}

\begin{proof}
	Let $\al$ be such an automorphism, then	$\mu(g) = \al(\I) \I^{-1} = 1$ for any $g \in G$ because the $G$-reps are trivial, and $\mu(\theta) = \al(\I) \I^{-1} = 1$ because the grading is trivial.
\end{proof}

\subsubsection{Circuits}

After stacking a QCA with trivial index with appropriate auxiliary chains we will be able to put it in circuit form:

\begin{definition} \label{def:decoupled QCA}
	An even equivariant QCA $\al$ on a quantum chain $\caA$ is decoupled in of blocks of size $R$ if there is a partition of $\Z$ in intervals $\{I_i\}_{i \in \Z}$ such that $\al$ leaves each $\caA_{I_i}$ invariant, and $\max_{i} \abs{I_i} = R$. The restricitons $\al|_{\caA_{I_i}}$ are all even equivariant automorphisms.
\end{definition}

\begin{definition} \label{def:circuit}
	An even equivariant automorphism $\al$ on a quantum chain $\caA$ is a circuit of depth $D$ with blocks of size $R$ if it is a composition of $D$ decoupled automorphisms with blocks of size $R$.
\end{definition}

We first show that circuits have trivial index.

\begin{lemma} \label{lem:triviality of circuits}
	If $\al$ is an even equivariant circuit, then $\ind(\al) = (1, 0, \boe)$ is trivial.
\end{lemma}

\begin{proof}
	An even equivariant circuit is a composition of decoupled even equivariant automorphisms, which are in particular all QCAs. By Proposition \ref{prop:QCA composition and stacking} it is sufficient to show that the index of any decoupled even equivariant automorphism is trivial. So suppose $\al$ is decoupled according to a partition $\{I_i\}_{i \in \Z}$ of $\Z$ into finite intervals $I_i$. We know from Proposition \ref{prop:index for QCA} that the index is independent of coarse graining and of where along the chain the index is computed. We caorse grain by blocking the intervals $I_i$ into single sites $\widetilde \caA_i = \caA_{I_i}$, so $\al$ leaves each site invariant w.r.t. to this caorse graining. For any $n \in \Z$, the right overlap algebra is then given by 
	\begin{equation}
		\widetilde \caR_n = \al \big( \widetilde \caA_{2n} \gotimes \widetilde \caA_{2n+1} \big) \cap \big( \widetilde \caA_{2n+1} \gotimes \widetilde \caA_{2n+2} \big) = \widetilde \caA_{2n+1}
	\end{equation}
	and the index of $\al$ is (Definition \ref{def:index for QCA})
	\begin{equation}
		\ind(\al) = \ind \big( \widetilde \caR_n, \widetilde \caA_{2n+1} \big) = \ind \big( \widetilde \caA_{2n+1}, \widetilde \caA_{2n+1} \big) = (1, 0, \boe)
	\end{equation}
	where we used the definition \ref{def:relative ind} of the relative index of central simple $G$-systems.
\end{proof}

\begin{proposition} \label{prop:decoupling circuits}
	Let $\caA$ be a quantum chain with $\caA_n$ either
	\begin{enumerate}[label=(\roman*)]
		\item balanced central simple $G$-systems (Definition \ref{def:balanced central simple G-system}) such that $\rho|_{\caA_n} = \Ad_{v_n \oplus v_n}$ with the $v_n$ direct sums of copies of the regular representation.
		\item balanced central simple $G$-systems all carrying the trivial $G$-representation.
		\item purely even central simple $G$-systems such that $\rho|_{\caA_n} = \Ad_{v_n}$ with the $v_n$ direct sums of copies of the regular representation.
		\item purely even central simple $G$-systems carrying trivial representations of $G$.
	\end{enumerate}
	Then any even equivariant circuit on $\caA$ of depth $D$ with blocks of size $R$ is $G$-equivalent to a decoupled automorphism with blocks of size 1 through a path of even equivariant circuits of depth $D$ with blocks of size $R$.
\end{proposition}

\begin{proof}
	It is sufficient to show that a decoupled automorphism with blocks of size $R$ is $G$-equivalent to a decoupled automorphism with blocks of size 1 through a path of decoupled automorphisms with blocks of size $R$. Let $\al$ be such a trivial product automorphism leaving blocks $\caA_{I_i}$ invariant and denote by $\al_i = \al|_{\caA_{I_i}}$. It is sufficient to show that each $\al_i$ can be deformed to a product. The fact that this can be done follows from a repeated application of Proposition \ref{prop:classification of 0-dimensional systems} and Lemma \ref{lem:existence of product automorphisms with given first cohomology I} in case (i), Lemma \ref{lem:existence of product automorphisms with given first cohomology II} in case (ii), Lemma \ref{lem:existence of product automorphisms with given first cohomology III} in case (iii) and Lemma \ref{lem:existence of product automorphisms with given first cohomology IV} in case (iv).
\end{proof}

\subsubsection{Proof of completeness}

We first show that every even equivariant QCAs with trivial index is stably $G$-equivalent to a decoupled QCA.
\begin{proposition} \label{prop:decoupling trivial QCA}
	Let $\al$ be an even equivariant QCA acting on a quantum chain $\caA$ with trivial index $\ind(\al) = (1, 0, \boe)$. Then $\al \gotimes \id \gotimes \id \gotimes \id$ acting on $\widetilde \caA = \caA \gotimes \bar \caA \gotimes \caA^{\reg} \gotimes \caA^{\fer}$ can be written as a depth 2 circuit consisting of even equivariant blocks supported on intervals of length at most twice the range of $\al$. In particular, $\al \gotimes \id \gotimes \id \gotimes \id$ is $G$-equivalent to an even equivariant QCA that is decoupled in blocks of size at most the range of $\al$ through a strongly continuous path of even equivariant QCAs with ranges uniformly bounded by twice the range of $\al$.

	Moreover, if the quantum chain $\caA$ is trivial in the sense of Definition \ref{def:trivial chain}, then the stacking with the conjugate chain $\bar \caA$ can be omitted. If $\caA$ has trivial symmetry then the stacking with $\caA^{\reg}$ can be ommited and $\widetilde \caA$ also has trivial symmetry. If $\caA$ has trivial grading then the stacking with $\caA^{\fer}$ can be omitted and $\widetilde \caA$ also has trivial grading.
\end{proposition}

\begin{proof}
	Coarse grain the chain so that $\al$ is nearest-neighbour and denote by $\caL_n$ and $\caR_n$ the overlap algebras. Since $\ind(\al)$ is trivial we have $\dim \caR_{n-1} = \dim \caA_{2n-1}$ and $\dim \caL_n = \dim \caA_{2n}$. By assumption, $\caA_{2n-1}$ and $\caA_{2n-1}$ are full matrix algebras, so all of the central simple superalgebras $\caL_n, \caR_{n-1}, \caA_{2n}, \caA_{2n-1}$ are rational. Again by triviality of $\ind(\al)$, the projective representations $v_{\frI}$, $\frI = \caL_n, \caA_{2n}, \caR_{n-1}, \caA_{2n-1}$ induced by the symmetry $\rho$ on the $G$-invariant full matrix algebras $\caL_n, \caA_{2n}, \caR_{n-1}, \caA_{2n-1}$ have associated cohomology classes satisfying
	\begin{equation}
		[v_{\caL_n}] = [v_{\caA_{2n}}], \quad [v_{\caR_{n-1}}] = [v_{\caA_{2n-1}}].
	\end{equation}
	Denoting by $\Theta_{\frI}$, $\frI = \caL_n, \caR_{n-1}, \caA_{2n}, \caA_{2n-1}$ the grading operator for $\frI$, the second component of $\ind(\frI)$ is $\zeta_{\frI} : G \rightarrow \Z_2$ appearing in $\rho^g(\Theta_{\frI}) = v_{\frI}(g) \Theta_{\frI} v_{\frI}(g)^* = (-1)^{\zeta_{\frI}(g)} \Theta_{\frI}$. Equivalently, $\theta(u_{\frI}(g)) = (-1)^{\zeta_{\frI}(g)}$, so $\zeta_{\frI}(g) = \tau(v_{\frI}(g))$ is the parity of the unitary $v_{\frI}$. The triviality of $\ind(\al)$ then means that
	\begin{equation}
		\tau(v_{\caL_n}(g)) = \tau(v_{\caA_{2n}}(g)), \quad \tau(v_{\caR_{n-1}}(g)) = \tau(v_{\caA_{2n-1}}(g))
	\end{equation}
	for all $g \in G$.

	If $\caA$ is a trivial chain then $[v_{\caA_n}] = \boe$ are trivial so the $v_{\frI}$ can be taken to be unitary representations of $G$. Moreover, $v_{\frI}(g) \Theta_{\frI} = \Theta_{\frI} v_{\frI}(g)$ for all $g$, \ie the $G$-representations are even.

	We achieve the same in general by stacking with the conjugate chain $\bar \caA$. Then $\caA \gotimes \bar \caA$ is a trivial chain (Lemma \ref{lem:conjugate stacking}) so the on-site algebras and the overlap algebras corresponding to $\al \gotimes \id$ acting on $\caA \gotimes \bar \caA$ all carry unitary representations of $G$ induced by the $G$-symmetry $\rho^g \gotimes \bar \rho^g$, and the grade automorphism acts on each of them by conjugation with a grading operator that commutes with the unitary $G$-representation, \ie the $G$-representations are all even.

	Let now $\caA' = \caA \gotimes \bar \caA$ and $\al' = \al \gotimes \id$ if $\caA$ is not already a trivial chain, and $\caA' = \caA$ and $\al' = \al$ if $\caA$ is a trivial chain. We denote the on-site algebras of by $\caA'_n$ and the overlap algebras corresponding to $\al'$ by $\caL'_n$ and $\caR'_n$. Let's denote the unitary representations of $G$ carried by $\frI' = \caL'_n, \caA'_{2n}, \caR'_{n-1}, \caA'_{2n-1}$ by $v'_{\frI'}$. We want to find isomorphisms of central simple $G$-systems $\Phi_{\caL'_n} : \caL'_n \rightarrow \caA'_{2n}$ and $\Phi_{\caR'_{n-1}} : \caR'_{n-1} \rightarrow \caA'_{2n-1}$ such that $\Phi_{\caL'_n}(v'_{\caL'_n}) = v'_{\caA'_{2n}}$ and $\Phi_{\caR'_{n-1}}(v'_{\caR'_{n-1}}) = v'_{\caA'_{2n-1}}$. Such isomorphisms can be found if and only if $\caL'_n$ and $\caA_{2n}$ are isomorphic as superalgebras and the unitary representations of $G$ carried by $\caL'_n$ and $\caA'_{2n}$, which are already known to be even, are equivalent as representations (and likewise for $\caR'_{n-1}$ and $\caA'_{2n-1}$). At present we only have two even unitary representations of the same dimension, which in general does not imply isomorphism of representations, and rational superalgebras of the same dimension, which also does not imply isomorphism as superalgebras. We will achieve the required isomorphisms by stacking with theregular chain and the fermion chain.

	If the quantum chain $\caA$ has trivial symmetry then so does $\caA'$ and the unitary $G$-representations $v'_{\frI'}$ are all trivial. In particular $\caL'_n$ and $\caA'_{2n}$ carry equivalent even unitary representations, and so do $\caR'_{n-1}$ and $\caA'_{2n-1}$. In this case we put $\caA'' = \caA'$ and $\al'' = \al'$.

	If the quantum chain $\caA$ does not have trivial symmetry, we get isomorphic representations by stacking with the regular chain $\caA^{\reg}$, following \cite{GSSC2020}. Then $\al'' = \al' \gotimes \id$ acting on $\caA'' = \caA' \gotimes \caA^{\reg}$ has on-site algebras $\caA''_n = \caA'_n \gotimes \caA^{\reg}_n$ and overlap algebras $\caL''_n = \caL'_n \gotimes \caA^{\reg}_{2n}$ and $\caR''_{n-1} = \caR'_{n-1} \gotimes \caA^{\reg}_{2n-1}$. By definition the regular chain is a spin chain, \ie all its elements are even so the preceeding graded tensor products are actually ordinary tensor products and the algebras $\frI'' = \caL''_n, \caA''_{2n}, \caR''_{n-1}, \caA''_{2n-1}$ carry unitary $G$-representations $v''_{\frI''} = v'_{\frI'} \otimes \rho^{\reg}$. Since the dimensions of $v_{\caL'_n}$ and $v_{\caA'_{2n}}$ are equal, it follows from Lemma \ref{lem:make reps isomorphic} that $\caL''_n$ and $\caA''_n$ carry isomorphic even unitary representations of $G$, and likewise for $\caR''_{n-1}$ and $\caA''_{2n-1}$.

	The central simple superalgebras $\caL''_n$ and $\caA''_{2n}$ are both rational and of the same dimension, and likewise for $\caR''_{n-1}$ and $\caA''_{2n-1}$. If $\caA$ has trivial grading, then so does $\caA''$ (because $\bar \caA$ and $\caA^{\reg}$ also have trivial grading) so the equality of dimensions of $\caL''_n$ and $\caA''_{2n}$ implies their isomorphism as (trivially graded) superalgebras, and likewise for $\caR''_{n-1}$ and $\caA''_{2n-1}$. In this case we put $\caA''' = \caA''$ and $\al''' = \al''$. 

	If $\caA$ has non-trivial grading, we get isomorphic algebras by stacking with the fermion chain $\caA^{\fer}$, following \cite{FPPV2019}. Then $\caA''' = \caA'' \gotimes \caA^{\fer}$ has on-site algebras $\caA'''_n = \caA''_n \gotimes \caA_n^{\fer}$ and to the extension $\al''' = \al'' \gotimes \id$ are associated overlap algebras $\caL'''_n = \caL''_n \gotimes \caA^{\fer}_n$ and $\caR'''_{n-1} = \caR''_{n-1} \gotimes \caA^{\fer}_{2n-1}$. Since $\caL''_n$ and $\caA''_{2n}$ are rational central simple algebras that carry isomorpic even representations of $G$, and $\caA_n^{\fer} \simeq \caM^{1|1}$ with trivial $G$-action, it follows from Lemma \ref{lem:fermion stacking I} and Lemma \ref{lem:fermion stacking II} that $\caL_n'''$ and $\caA'''_{2n}$ are isomorphic as central simple $G$-systems, and likewise for $\caR'''_{n-1}$ and $\caA'''_{2n-1}$.

	By Theorem \ref{thm:overlap factorization},
	\begin{equation}
		\caC'''_n = \caA'''_{2n-1} \gotimes \caA'''_{2n} = \caR'''_{n-1} \gotimes \caL'''_n
	\end{equation}
	so we can combine the isomorphisms into a single automorphism of the central simple $G$-system $\caC'''_n = \caA'''_{2n-1} \gotimes \caA'''_{2n}$:
	\begin{equation}
		\Phi_n = \Phi_{\caR'''_{n-1}} \gotimes \Phi_{\caL'''_n} : \caC'''_n \rightarrow \caC'''_n
	\end{equation}
	which satisfies $\Phi_n(\caL'''_n) = \caA'''_{2n}$ and $\Phi_{n}(\caR'''_{n-1}) = \caA'''_{2n-1}$.

	The simultaneous action of all the $\Phi_n$ defines an even equivariant depth 1 circuit $\Phi$ with blocks supported on pairs of neighbouring sites (in the coarse grained chain). Define the even equivariant QCA $\Psi = \Phi \circ \al'''$, it leaves each of the algebras $\caB'''_n = \caA'''_{2n} \gotimes \caA'''_{2n+1}$ invariant. Indeed, by Theorem \ref{thm:overlap factorization},
	\begin{equation}
		\caB_n''' = \al'''^{-1} \big( \caL'''_n \big) \gotimes \al'''^{-1} \big( \caR'''_n \big)
	\end{equation}
	while $\Psi \big( \al'''^{-1}( \caL'''_n) \big) = \Phi( \caL'''_n) = \caA'''_{2n}$ and $\Psi \big( \al'''^{-1}(\caR'''_n) \big) = \Phi(\caR'''_n) = \caA_{2n-1}$. It follows that $\Psi$ is a depth 1 even equivariant circuit with blocks supported on nearest neighbour pairs (in the coarse grained chain).

	We have thus written $\al''' = \Phi^{-1} \circ \Psi$ as a depth 2 circuit consisting of even equivariant blocks of size 2 in the coarse grained chain. Before coarse graining, this corresponds to blocks supported on intervals of length at most twice the range of $\al$. This proves the first claim of the proposition.

	In case we stacked with both the regular and the fermion chain, all the $\caA'''_n$ are balanced central simple $G$-systems carrying $G$-reps $v_n \oplus v_n$ with $v_n$ isomorphic to a direct sum of copies of the regular representation. The second claim of the Proposition then follows from case (i) of Proposition \ref{prop:decoupling circuits}.

	If $\caA$ has trivial symmetry action and we did not stack with the regular chain, but we did stack with the fermion chain, then by Lemma \ref{lem:fermion stacking I}, the $\caA'''_n$ are balanced simple $G$-systems with trivial $G$-action. The second claim of the Proposition then follows from case (ii) of Proposition \ref{prop:decoupling circuits}.

	If all $\caA_n$ are purely even, so we don't stack with the fermion chain, but we do stack with the regular chain, then the $\caA'''_n$ are all purely even central simple $G$-systems carrying $G$-reps isomorphic to direct sums of copies of the regular representation. The second claim of the Proposition then follows from case (iii) of Proposition \ref{prop:decoupling circuits}.

	Finally, if $\caA$ has trivial symmetry and trivial grading, so we don't stack with the regular chain or with the fermion chain, then all the $\caA'''_n$ are purely even central simple $G$-systems carrying trivial $G$-reps. The second claim of the Proposition then follows from case (iv) of Proposition \ref{prop:decoupling circuits}.
\end{proof}

A decoupled automorphism is stably $G$-equivalent to the identity through stacking with its inverse.

\begin{proposition} \label{prop:decoupled is stably G-equivalent to identity}
	Let $\al$ be an even equivariant automorphism acting on a quatum chain $\caA$ that is decoupled according to a partition $\{I_i\}_{i \in \Z}$ of $\Z$. Then $\al \gotimes \al^{-1}$ acting on $\caA \gotimes \caA$ is $G$-equivalent to $\id$ through a path automorphisms that are all decoupled accordingto the partition $\{I_i\}_{i \in \Z}$. 
\end{proposition}

\begin{proof}
	By Lemma \ref{lem:stacking with inverse yields trivial first cohomology}, the $(\al \gotimes \al^{-1})|_{\caA_{I_i}}$ all have trivial first cohomology index. By Lemma \ref{lem:0-dimensional deform to identity} all of the $(\al \gotimes \al^{-1})|_{\caA_{I_i}}$ can be continuously deformed to $\id$ along a path of even equivariant automorphisms of $\caA_{I_i}$. Doing the continuous deformations simultaneously for all $i \in \Z$ yields the required strongly continuous path from $\al \gotimes \al^{-1}$ to $\id$.
\end{proof}

\subsection{Examples} \label{sec:examples}

We construct for each possible value $(d, \zeta, \bnu)$ of the index an even equivariant QCA $\al$ such that $\ind(\al) = (d, \zeta, \bnu)$.

\textbf{Shifts :} Consider a quantum chain $\caA$ with trivial gradign and trivial symmetry and on-site algebras $\caA_n$ all isomorphic to the $D \times D$ matrix algebra $\caM^D$. A right shift $\sigma_D$ on $\caA$ is an automorphism such that $\sigma_D(\caA_n) = \caA_{n+1}$ for all $n$. Such a right shift $\sigma_D$ can be constructed by fixing particular isomorphisms $\phi_n : \caA_n \rightarrow \caM^D$ for each $n$ and setting $\sigma_d(x) = (\phi_{n+1}^{-1} \circ \phi_n)(x)$ for $x \in \caA_n$. Clearly, $\sigma_D$ is nearest neighbour and the right overlap algebras are given by
\begin{equation}
	\caR_n = \sigma_D \big( \caA_{2n} \gotimes \caA_{2n+1} \big) \cap \big( \caA_{2n+1} \gotimes \caA_{2n+2} \big) = \caA_{2n+1} \gotimes \caA_{2n+2}
\end{equation}
so the first component of $\ind(\sigma_D)$ is
\begin{equation}
	\sqrt{ \frac{\dim \caR_n}{ \dim \caA_{2n+1} } } = \sqrt{ \dim \caA_{2n+2}} = D.
\end{equation}
Since the grading and symmetry are tirvial, we have
\begin{equation} \label{eq:ind of shift}
	\ind(\sigma_{D}) = (D, 0, \boe).
\end{equation}

By Proposition \ref{prop:QCA composition and stacking}, we see that $\ind(\sigma_{D_1} \gotimes \sigma_{D_2}^{-1})= (D_1/D_2, 0, \boe)$. We now have examples of QCAs the first component of whose index takes any value in $\Q$. The find values in $\sqrt{2} \Q$, we consider the so-call Kitaev chain, introduced in this context in \cite{FPPV2019}.

\textbf{The Kitaev Chain :} Consider the fermion chain $\caA^{\fer}$ (Definition \ref{def:fermion chain}) which has each local algebra $\caA_n$ isomorphic to $\caM^{1|1}$. Since $\caM^{1|1} \simeq K \gotimes K$ where $K = \C[\ep]$ is the superalgebra generated by a single odd self-adjoint unitary $\ep$, we have for each $n$ that $\caA_n = \caA_{n, L} \gotimes \caA_{n, R}$ where $\caA_{n, L}, \caA_{n, R} \simeq K$. We denote by $\ep_{n, L}$ and $\ep_{n, R}$ the odd self-adjoint unitary generators of $\caA_{n, L}$ and $\caA_{n, R}$ respectively. The algebra $K$ describes a Majorana mode, and we define a right-shift $\sigma_{\Maj}$ of Majorana modes by setting $\sigma_{\Maj}(\ep_{n, L}) = \ep_{n, R}$ and $\sigma_{\Maj}(\ep_{n, R}) = \ep_{n+1, L}$ for all $n \in \Z$. This is easily seen to define an even nearest neighbour QCA on the whole of $\caA^{\fer}$, and its right overlap algebras are
\begin{equation}
	\caR_n = \sigma_{\Maj} \big( \caA_{2n} \gotimes \caA_{2n+1} \big) \cap \big( \caA_{2n+1} \gotimes \caA_{2n+2} \big) = \caA_{2n + 1} \gotimes \caA_{2n+2, L} \simeq \caM^{1|1} \gotimes K,
\end{equation}
which are radical central simple superalgebras, and we find that the first component of $\ind \sigma_{\Maj}$ is
\begin{equation}
	\sqrt{\frac{\dim \caR_n}{\dim \caA_{2n+1}}} = \sqrt{ \frac{2 \times 2 \times 2}{2 \times 2}} = \sqrt{2}.
\end{equation}

Since the symmetry action on $\caA^{\fer}$ is trivial, we have
\begin{equation}
	\ind(\sigma_{\Maj}) = (\sqrt{2}, 0, \boe).
\end{equation}

Using Proposition \ref{prop:QCA composition and stacking} we see that $\ind(\sigma_{\Maj} \gotimes \sigma_{D_1} \gotimes \sigma_{D_2}^{-1}) = (\sqrt{2}D_1 / D_2, 0, \boe)$, so for all values in $\sqrt{2} \Q$ we have an even equivariant QCA the first component of whose index takes that value.

\textbf{Arbitrary $\ind_2 \in \Hom(G, \Z_2)$ : } Fix $\zeta \in \Hom(G, \Z_2)$. Let $g \mapsto v_g$ be a $D$-dimensional unitary representation of $G$ and define a $2D$-dimensional representation $g \mapsto V_g$ on $\caM^{D|D}$ by
\begin{equation}
	V_g = \begin{bmatrix} v_g & \\
				& v_g \end{bmatrix} \quad \text{if} \,\,\, \zeta(g) = 0
\end{equation}
and
\begin{equation}
	V_g = \begin{bmatrix} & v_g\\
				v_g  &  \end{bmatrix} \quad \text{if} \,\,\, \zeta(g) = 1.
\end{equation}
One easily verifies that this is still a unitary representation, and noting that the grading operator on $\caM^{D|D}$ is given by
\begin{equation}
	\Theta = \begin{bmatrix} \I & \\ & -\I \end{bmatrix}
\end{equation}
we see that the central simple $G$-system $(\caM^{D|D}, \Ad_{V_g})$ has index
\begin{equation}
	\ind(\caM^{D|D}, \Ad_{V_g}) = (2D, \zeta, \boe).
\end{equation}

Now consider the quantum chain $\caA$ with symmetry automorphisms $\rho^g$ and local algebras $\caA_n = \caA_{n, L} \gotimes \caA_{n, R}$ such that $\caA_{n, L}, \caA_{n, R} \simeq (\caM^{D|D}, \Ad_{V_g})$ as central simple $G$-systems. The isomorphism as central simple $G$-systems implies the existence of even equivariant isomorphisms between the $\caA_{n, L}, \caA_{n, R}$ which can be used to define an even equivariant automorphism $\al$ such that $\al(\caA_{n, L}) = \caA_{2, R}$ and $\al(\caA_{n, R}) = \caA_{n+1, L}$ for all $n \in \Z$. The automorphism $\al$ is a nearest neighbour QCA and its right-overlap algebras are
\begin{equation}
	\caR_n = \al \big( \caA_{2n} \gotimes \caA_{2n+1} \big) \cap \big( \caA_{2n+1} \gotimes \caA_{2n+2} \big) = \caA_{2n + 1, L} \gotimes \caA_{2n+1, R} \gotimes \caA_{2n+2, L}.
\end{equation}

Using Proposition \ref{prop:stacking for absolute ind} we find
\begin{equation}
	\ind(\caR_n) = (8 D^3, \zeta, \boe)
\end{equation}
and
\begin{equation}
	\ind(\caA_{2n+1}) = (4 D^2, 0, \boe).
\end{equation}

The index of the even equivariant QCA $\al$ is therefore given by
\begin{equation}
	\ind(\al) = \ind(\caR_n) \cdot \ind(\caA_{2n+1})^{-1} = (2D, \zeta, \boe).
\end{equation}

By stacking with copies of (Majaorana) shifts we can now obtain QCAs with index $(D, \zeta, \boe)$ for any $D \in \Q \cup \sqrt{2}$ and any $\zeta \in \Hom(G, \Z_2)$. It remains to construct examples with arbitrary second cohomology index $\bnu \in H^2(G, U(1))$.

\textbf{Arbitrary second cohomology :}  Fix $\bnu \in H^2(G, U(1))$ and let $g \mapsto v_g \in \caM^D$ be a $D$-dimensional projective representation of $G$ whose associated second cohmomlogy is $\bnu$. The conjugate projective representation $g \mapsto \bar v_g \in \caM^D$ then has associated second cohomology $\bnu^{-1}$.

Consider a quantum chain $\caA$ with trivial grading and on-site algebras $\caA_n = \caA_{n, L} \otimes \caA_{n, R}$ and a symmertry action $\rho$ such that each $\caA_{n, L}$, $\caA_{n, R}$ is invariant and $(\caA_{n, L}, \rho) \simeq (\caM^D, \Ad_{v_g})$ and $(\caA_{n, R}, \rho) \simeq (\caM^D, \Ad_{\bar v_g})$ as central simple $G$-systems for each $n \in \Z$. The isomorphism as central simple $G$-systems implies the existence of even equivariant isomorphisms between the $\caA_{n, L}, \caA_{n, R}$ which can be used to define an even equivariant automorphism $\al$ such that $\al(\caA_{n, L}) = \caA_{2, R}$ and $\al(\caA_{n, R}) = \caA_{n+1, L}$ for all $n \in \Z$. The automorphism $\al$ is a nearest neighbour QCA and its right-overlap algebras are
\begin{equation}
	\caR_n = \al \big( \caA_{2n} \gotimes \caA_{2n+1} \big) \cap \big( \caA_{2n+1} \gotimes \caA_{2n+2} \big) = \caA_{2n + 1, L} \gotimes \caA_{2n+1, R} \gotimes \caA_{2n+2, L}.
\end{equation}

Using Proposition \ref{prop:stacking for absolute ind} we find
\begin{equation}
	\ind(\caR_n) = (D^3, 0, \bnu)
\end{equation}
and
\begin{equation}
	\ind(\caA_{2n+1}) = (D^2, 0, \boe)
\end{equation}
where we noted that the second component is always zero because the grading is trivial.

The index of the even equivariant QCA $\al$ is therefore given by
\begin{equation}
	\ind(\al) = \ind(\caR_n) \cdot \ind(\caA_{2n+1})^{-1} = (D, 0, \bnu).
\end{equation}.

Above, we'd already constructed even equivariant QCAs with index $(D, \zeta, \boe)$ for arbitrary $D \in \Q \cup \sqrt{2} \Q$ and arbitrary $\zeta \in \Hom(G, U(1))$. By stacking these with the automorphism just constructed we obtain QCAs with index taking any value in $\Q \cup \sqrt{2} \times \Hom(G, \Z_2) \times H^2(G, U(1))$.

\section{Index theory for quasi-local even equivariant automorphisms} \label{sec:index theory for quasi-local}

\subsection{Approximation of quasi-local automorphisms by QCAs}

The results of this subsection are graded equivariant extensions of those presented in Section 5.1 of \cite{RWW2020}. This generalization is rather straightforward and the method of proof is exactly the same, we nevertheless give the proofs in full in order to keep this work self-contained.

\subsubsection{Extending even equivariant automorphisms to the von Neumann algebra}

We want to analyze quasi-local even equivariant automorphisms using Theorem \ref{thm:near inclusions}, this requires extending them to the von Neumann algebra $\caA^{\vN}$, see section \ref{sec:operator algebraic preliminaries}. We want the extension of $\al$ to retain its locality properties. This is verified in the following lemma, which is a graded equivariant version of Lemma 3.6 of \cite{RWW2020}:
\begin{lemma} \label{lem:vN extension}
	\begin{enumerate}[label=(\roman*)]
		\item Suppose $\al$ is an even automorphism of $\caA$ such that $\al(\caA_X) \nsub{\ep} \caA_{Y}$. Then the extension of $\al$ to the von Neumann algebra $\caA^{\vN}$ satisfies $\al(\caA_{X}^{\vN}) \nsub{12\ep} \caA_Y^{\vN}$. In particular, if $\al$ is quasi-local with $f(r)$-tails, then $\al$ extends such that $\al(\caA_X^{\vN}) \nsub{12 f(r)} \caA_{B(X, r)}^{\vN}$ for any interval $X$.

		\item Suppose $\al$ is an even automorphism of $\caA^{\vN}$ such that $\al(\caA_X^{\vN}) \nsub{f(r)} \caA_{X^{(r)}}^{\vN}$ for any interval $X$ and some function $f(r)$. Then
			\begin{equation}
				\al^{-1}(\caA_X^{\vN}) \nsub{16 f(r)} \caA_{X^{(r)}}^{\vN}
			\end{equation}
			for all intervals $X$.

		\item If $\al$ is an even automorphism of $\caA^{\vN}$ such that $\al(\caA_X^{\vN}) \nsub{f(r)} \caA_{X^{(r)}}^{\vN}$ for any interval $X$ and some function $f$ with $\lim_{r \uparrow \infty} = 0$, then $\al$ restricts to an even quasi-local automorphism on $\caA_{\Z}$ with $f(r)$-tails.

		\item If $\al$ is an even quasi-local automorphism on $\caA$ with $f(r)$-tails, then $\al^{-1}$ is an even quasi-local automorphism with $192 f(r)$-tails.
	\end{enumerate}
\end{lemma}

In the proof of this lemma we will take recourse to von Neumann algebras
\begin{equation}
	\caM^{\Theta} := \big( \caM \cup \{\Theta\} \big)''
\end{equation}
for arbitrary von Neumann subalgebras $\caM \subset \caA^{\vN}$ and $\Theta$ the grading unitary implementing $\theta$ on $\caB(\caH)$. They appear because the conditional expectations constructed in Appendix \ref{app:super vN} take values in such algebras. We first prove a near-inclusion property of such algebras.
\note{I would like to prove that $\Theta \not\in \caA^{\vN}$.}

\begin{lemma} \label{lem:property of Theta-adjoined}
Let $\caM, \caN \subset \caA^{\vN}$ be super von Neumann algebras such that $\caM \nsub{\ep} \caN^{\Theta}$, then $\caM \nsub{\ep} \caN$.
\end{lemma}

\begin{proof}
	We fist define an involutive *-automorphism $\mu$ on $(\caA*{\vN})^{\Theta}$ by setting $\mu(a) = a$ for any $a \in \caA^{\vN}$ and $\mu(\Theta) = -\Theta$. Since $\caA^{\vN}$ is invariant under $\Ad_{\Theta}$, any element of $(\caA^{\vN})^{\Theta}$ is of the form $a_1 + a_2 \Theta$ with $a_1, a_2 \in \caA^{\vN}$. The involution $\mu$ is therefore defined on the whole of $(\caA^{\vN})^{\Theta}$, and using the fact that $\Theta$ supercommutes with all elements of $\caA^{\vN}$ one easily checks that $\mu$ is indeed a $*$-automorphism. An element $a \in (\caA^{\vN})^{\Theta}$ is $\mu$-even if $\mu(a) = a$ and $\mu$-odd if $\mu(a) = -a$. The $\mu$-even elements of $(\caA^{\vN})^{\Theta}$ are precisely the elements of $\caA^{\vN}$.

	If $\caM \subset \caA^{\vN}$ is a super von Neumann algebra, then it is closed under $\mu$. Any element $x \in \caM$ then decomposes into unique $\mu$-even and $\mu$-odd parts,
	\begin{equation}
		x_{\mu \pm} = \frac{1}{2} ( x \pm \mu(x) ).
	\end{equation}
	Note that if $\norm{x} \leq \ep$, then $\norm{x_{\mu +}}, \norm{x_{}\mu-} \leq \ep$ also. Indeed,
	\begin{equation}
		\norm{x_{\mu_{\pm}}} \leq \frac{1}{2} \left( \norm{x} \pm \norm{\mu(x)} \right) \leq \ep.
	\end{equation}

	Now, since $\caM \nsub{\ep} \caN^{\Theta}$ we have for each $x \in \caM$ an element $y \in \caN^{\Theta}$ such that $\norm{x - y} \leq \ep \norm{x}$. By te above, also the $\mu$-even part is bounded, \ie
	\begin{equation}
		\norm{x_{\mu+} - y_{\mu+}} \leq \ep \norm{x}.
	\end{equation}
	But $\caM \subset \caA^{\vN}$ is purely $\mu$-even, so $x_{\mu+} = x$, and we get $\norm{x - y_{\mu+}} \leq \ep \norm{x}$ with $y_{\mu+} \in \caN$ (because it is even). We conclude that $\caM \nsub{\ep} \caN$.
\end{proof}

\begin{proofof}[Lemma \ref{lem:vN extension}]
	(i) By Kaplansky's density theorem there is a net $x_i$ in $\caA_X$ converging strongly to $x$ such that $\norm{x_i} \leq \norm{x}$. Since
	\begin{equation}
		\al(\caA_X) \nsub{\ep} \caA_Y \subset \caA_Y^{\vN} \subset (\caA_{Y^c}^{\vN})^{\sharp}
	\end{equation}
	we get from Lemma \ref{lem:near inclusions and supercommutators} that
	\begin{equation}
		\norm{[x_i, c]_s} \leq 4 \ep \norm{x_i} \norm{c}
	\end{equation}
	for all $c \in \caA_{Y^c}^{\vN}$. Then using the lower semicontinuity of the operator norm and $\norm{x_i} \leq \norm{x}$ we get
	\begin{equation}
		\norm{[x, c]_s} \leq \sup_i \norm{[x_i, c]_s} \leq 4 \ep \norm{x} \norm{c}
	\end{equation}
	for all $x \in \al(\caA_X)^{\vN}$ and all $c \in \caA_{Y^c}^{\vN}$. Since all algebras concerned are $\theta$-hyperfinite and $\widetilde \caA_X^{\vN} \subset (\caA^{\vN})^{\Theta}$ for any $X \subset \Z$, the second part of Lemma \ref{lem:near inclusions and supercommutators} gives
	\begin{equation}
		\al(\caA_X)^{\vN} \nsub{12 \ep} ( \caA_{Y^c}^{\vN})^{\sharp} \cap (\caA^{\vN})^{\Theta} = (\caA_Y^{\vN})^{\Theta}.
	\end{equation}
	Finally, using Lemma \ref{lem:property of Theta-adjoined} and the fact that $\al(\caA_X)^{\vN} = \al(\caA_X^{\vN})$ we conclude that $\al(\caA_X^{\vN}) \nsub{12 \ep} \caA_Y^{\vN}$. This proves (i).

	(ii)  $(X^{(r)})^c$ is a disjoint union of at most two intervals $Y_1$, $Y_2$ and $Y_i^{(r)} \subset X^c$ so
	\begin{equation}
		\al \big( \caA^{\vN}_{Y_i} \big) \nsub{f(r)} \caA_{X^c}^{\vN}.
	\end{equation}
	Since all algebras concerned are $\theta$-hyperfinite and $[\al(\caA_{Y_1}^{\vN}), \al(\caA_{Y_2}^{\vN})]_s = 0$, applying Proposition \ref{prop:simultaneous near inclusions of theta-hyperfinite vN algebras} with $\caC = \caA^{\vN}$ gives
	\begin{equation}
		(\caA_{X^c}^{\vN})^{\sharp} \cap \caA^{\vN} \nsub{16 f(r)} \al(\caA_{(X^{(r)})^c}^{\vN})^{\sharp} \cap (\caA^{\vN})^{\Theta} 
	\end{equation}
	where we used $(\caA^{\vN} \cup \widetilde{\al(\caA_{Y_1}^{\vN})} \cup \widetilde{\al(\caA_{Y_2}^{\vN})}  )'' = (\caA^{\vN})^{\Theta}$. This then yields immediately
	\begin{equation}
		\caA_X^{\vN} \nsub{16 f(r)} \al( \caA_{X^{(r)}}^{\vN} )^{\Theta}.
	\end{equation}
	Using Lemma \ref{lem:property of Theta-adjoined} now yields the claim.

	(iii) Take $x \in \caA$. We must show that $\al(x), \al^{-1}(x) \in \caA$ also. If $x$ is strictly local, then $x \in \caA_X$ for some finite interval $X$ and since $\al(\caA_{X}^{\vN}) \nsub{f(r)} \caA_{X^{(r)}}^{\vN} = \caA_{X^{(r)}}$ for all $r$, there is a sequence of strictly local operators $y_r \in \caA_{X^{(r)}}$ converging to $\al(x)$ in norm. Hence $\al(x) \in \caA$. If $x$ is not strictly local, then let $\{x_i\}$ be a sequence of strictly local operators convering to $x$ in norm. Then $\al(x_i)$ is a sequence in $\caA$ converging in norm to $\al(x)$, so again $\al(x) \in \caA$. Using (ii), the same argument shows that $\al^{-1}(x) \in \caA$. Thus, $\al$ restricts to an even *-automorphism of $\caA$. That $\al|_{\caA}$ is quasi-local with $f$-tails now follows immediately from the assumption that $\al(\caA_X)^{\vN} \nsub{f(r)} \caA_{X^{(r)}}^{\vN}$ for all intervals $X$.

	(iv) By (i), $\al$ extends to an even automorphism of $\caA^{\vN}$ with $12f$-tails. By (ii), the inverse $\al^{-1}$ of this extension is an even automorphism with $192f$-tails. Finally, by (iii) $\al^{-1}$ restricts to $\caA$ as an even quasi-local automorphism with $192f$-tails.
\end{proofof}

\subsubsection{Approximation by QCAs}

We begin with graded equivariant versions of \cite{RWW2020}'s Lemmas 5.1 and 5.2.

\begin{lemma} \label{lem:localize to the right}
	There are universal constants $C, \ep_0 > 0$ such that if $\al$	is an even equivariant $\ep$-nearest neighbour automorphism of $\caA^{\vN}$ with $\ep \leq \ep_0$ and
	\begin{equation}
		\al \big( \caA_{\geq n}^{\vN} \big) \subset \caA_{\geq n-1}^{\vN}
	\end{equation}
	for some site $n \in \Z$, then there exists an even $G$-invariant unitary $u \in \caA_{\geq n-1}^{\vN}$ with $\norm{u - \I} \leq C \ep$ such that $\tilde \al = \Ad_{u} \circ \al$ satisfies
	\begin{align}
		\tilde \al \big( \caA_{\leq n-1}^{\vN} \big) &\subset \caA_{\leq n}^{\vN}, \label{eq:localize right 1} \\
		\tilde \al \big( \caA_{\geq n}^{\vN} \big) &\subset \caA_{\geq n-1}^{\vN}. \label{eq:localize right 2}
	\end{align}

	If $\al$ is quasi-local with $f$-tails \note{as an automorphism of $\caA^{\vN}$ or of $\caA$?} then $\tilde \al$ restricts to a quasi-local even equivariant automorphism of $\caA$ with $\caO(f)$-tails and
	\begin{equation} \label{eq:localize right 3}
		\norm{  (\al - \tilde \al)|_{\caA_{\leq n-r}^{\vN} \gotimes \caA_{\geq n + r + 1}^{\vN}   }  } = \caO \big( f(r) \big)
	\end{equation}
\end{lemma}

\begin{proof}
	The automorphism $\al^{-1}$ is $8 \ep$-nearest neighbour by Lemma \ref{lem:vN extension} (ii) so $\al^{-1}\big( \caA^{\vN}_{\geq n+1} \big) \nsub{8 \ep} \caA_{\geq n}^{\vN}$ whence $\caA_{\geq n+1}^{\vN} \nsub{8 \ep} \al \big(  \caA_{\geq n}^{\vN} \big)$. If $\ep < 1/64$ Theorem \ref{thm:near inclusions} provides an even equivariant unitary $u \in \left(\caA_{\geq n+1}^{\vN} \cup \al \big( \caA_{\geq n}^{\vN}  \big) \right)''$ such that
	\begin{equation}
		u^* \caA_{\geq n+1}^{\vN} u \subset \caA_{\geq n}^{\vN}
	\end{equation}
	and $\norm{u - \I} \leq 96 \ep$. Since $\al \big( \caA_{\geq n}^{\vN} \big) \subset \caA_{\geq n-1}^{\vN}$ we have in fact $u \in \caA_{\geq n-1}^{\vN}$. Define $\tilde \al := \Ad_{u} \circ \al$, then $\caA_{\geq n+1}^{\vN} \subset \tilde \al \big( \caA_{\geq n}^{\vN} \big)$ and we obtain \eqref{eq:localize right 1} by taking supercommutants. Further, by assumption we have $\al \big( \caA_{\geq n}^{\vN} \big) \subset \caA_{\geq n-1}^{\vN}$ and we've seen that $u \in \caA_{\geq n-1}^{\vN}$, so
	\begin{equation}
		\tilde \al \big( \caA_{\geq n}^{\vN} \big) = u \al \big( \caA_{\geq n}^{\vN} \big) u^* \subset u \caA_{\geq n-1}^{\vN} u^* = \caA_{\geq n-1}^{\vN},
	\end{equation}
	hence \eqref{eq:localize right 2} is also satisfied.

	Suppose $\al$ is quasi-local with $f$-tails. Lemma \ref{lem:local errors control global errors} says that in order to show \eqref{eq:localize right 3} it is sufficient to show $\norm{(\al - \tilde \al)|_{\caA_{\leq n-r}^{\vN}}} = \caO(f(r))$ and $\norm{(\al - \tilde \al)|_{\caA_{\geq n+r+1}^{\vN}}} = \caO(f(r))$.
	To see the former, note that since $\al \big( \caA_{\leq n-r}^{\vN} \big) \nsub{f(r)} \caA_{\geq n}^{\vN}$ we have for each $x \in \al \big( \caA_{\leq n-r}^{\vN} \big)$ that $\norm{[x, y]_s} \leq 4 f(r)$ for all $y \in \caA_{\geq n+1}^{\vN} \cup \al \big( \caA_{\geq n}^{\vN} \big)$. Theorem \ref{thm:near inclusions} then implies that $\norm{u^* x u - x} = \caO(f(r)) \norm{x}$, from which $\norm{(\al - \tilde \al)|_{\caA_{\leq n-r}^{\vN}}} = \caO(f(r))$ immediately follows.

	To see the latter, namely that $\norm{(\al - \tilde \al)|_{\caA_{\geq n+r+1}^{\vN}}} = \caO(f(r))$, note that since $\al \big( \caA_{\geq n+r+1}^{\vN} \big) \nsub{f(r)} \caA_{\geq n+1}$ and $\al \big( \caA_{\geq n+r+1}^{\vN} \big) \subset \al \big( \caA_{\geq n}^{\vN} \big)$, so for each $x \in \al \big(  \caA_{\geq n+r+1}^{\vN} \big)$ we have that $x \nin{f(r)} \caA_{\geq n+1}^{\vN}$ and $x \nin \al \big(  \caA_{\geq n} \big)$. It then follows from Theorem \ref{thm:near inclusions} that $\norm{u^* x u - x} = \caO(f(r)) \norm{x}$, and $\norm{(\al - \tilde \al)|_{\caA_{\geq n+r+1}^{\vN}}} = \caO(f(r))$ immediately follows from this.

	Finally, by Lemma \ref{lem:vN extension} (iii), the even equivariant automorphism $\tilde \al$ of $\caA^{\vN}$ restricts to a quasi local even equivariant automorphism of $\caA$ with $\caO(f(r))$-tails.
\end{proof}

\begin{lemma} \label{lem:localize to the left}
	There exist universal constants $C, \ep_0 > 0$ such that if $\al$ is an even equivariant $\ep$-nearest neighbour automorphism of $\caA^{\vN}$ with $\ep \leq \ep_0$ and
	\begin{equation}
		\al \big( \caA_{\leq n}^{\vN} \big) \subset \caA^{\vN}_{\leq n+1}
	\end{equation}
	for some site $n \in \Z$, then there is an even $G$-invariant unitary $u \in \caA_{\geq n+1}^{\vN}$ with $\norm{u - \I} \leq C \ep$ and such that $\tilde \al = \al \circ \Ad_u$ satisfies
	\begin{align}
		\tilde \al \big( \caA^{\vN}_{\geq n+3}  \big) &\subset \caA_{\geq n+2}^{\vN}, \label{eq:localize left 1} \\
		\tilde \al \big( \caA_{\leq n}^{\vN} \big) &\subset \caA_{\leq n+1}^{\vN}. \label{eq:localize left 2}
	\end{align}

	If $\al$ is quasi-local with $f$-tails \note{as an automorphism of $\caA$ or $\caA^{\vN}$?} then $\tilde \al$ restricts to a quasi-local even equivariant automorphism of $\caA$ with $\caO(f)$-tails and
	\begin{equation} \label{eq:localize left 3}
		\norm{ (\al - \tilde \al)|_{ \caA_{\leq n}^{\vN} \gotimes \caA_{\geq n+r+2}^{\vN} } } = \caO(f(r)).
	\end{equation}
\end{lemma}

\begin{proof}
	From $\al \big( \caA_{\geq n+3}^{\vN} \nsub{\ep} \caA_{\geq n+2}^{\vN}$ it follows that $\caA_{\geq n+3}^{\vN} \nsub{\ep} \al^{-1} \big( \caA_{\geq n+2}^{\vN} \big)$ so if $\ep < 1/8$, Theorem \ref{thm:near inclusions} provides an even $G$-invariant unitary $u \in \big( \caA_{\geq n+3}^{\vN} \cup \al^{-1}(\caA_{\geq n+2}^{\vN} ) \big)''$ with $\norm{u - \I} \leq 12 \ep$ and such that
	\begin{equation}
		u \caA_{\geq n+3}^{\vN} u^* \subset \al^{-1} \big( \caA_{\geq n+2}^{\vN} \big).
	\end{equation}

	The automorphism $\tilde \al = \al \circ \Ad_u$ then satisfies \eqref{eq:localize left 1}.

	Note further that since by assumption $\al \big( \caA_{\leq n} \big) \subset \caA_{\leq n+1}^{\vN}$ we have $\caA_{\leq n}^{\vN} \subset \al^{-1} \big( \caA_{\leq n+1}^{\vN} \big)$. Taking supercommutants yields $\al^{-1}\big( \caA_{\geq n+2}^{\vN} \big) \subset \caA_{\geq n+1}^{\vN}$ so $u \in \big( \caA_{\geq n+3}^{\vN} \cup \al^{-1}(\caA_{\geq n+2}^{\vN} ) \big)'' \subset \caA_{\geq n+1}^{\vN}$ as required. In particular, the even element $u$ commutes with $\caA_{\leq n}^{\vN}$ so
	\begin{equation}
		\tilde \al \big( \caA_{\leq n}^{\vN} \big) = \al \big( u \caA_{\leq n}^{\vN} u^* \big) = \al \big( \caA_{\leq n}^{\vN} \big) \subset \caA_{\leq n+1}^{\vN},
	\end{equation}
	which shows \eqref{eq:localize left 2}.

	If $\al$ is quasi-local with $f$-tails then $\al \big( \caA_{\geq n+r+2}^{\vN} \big) \nsub{f(r)} \caA_{\geq n+2}^{\vN}$ so if $x \in \caA_{\geq n+r+2}^{\vN}$ then $x \nin{f(r)} \al^{-1} \big( \al^{-1}(\caA_{\geq n+2}^{\vN}) \big)$ and trivially (if $r \geq 1$), $x \in \caA_{\geq n+3}^{\vN}$. It then follows from Theorem \ref{thm:near inclusions} that $\norm{u x u^* - x} = \caO(f(r)) \norm{x}$ for all such $x$, hence $\norm{(\al - \tilde \al)|_{\caA_{\geq n+r+2}^{\vN}} } = \caO(f(r))$. Moreover, since $u \in \caA_{\geq n+1}^{\vN}$ we have $\norm{(\al - \tilde \al)|_{\caA_{\leq n}^{\vN}}}$. The required \eqref{eq:localize left 3} then follows from Lemma \ref{lem:local errors control global errors}.

	Finally, by Lemma \ref{lem:vN extension} (iii), the even equivariant automorphism $\tilde \al$ of $\caA^{\vN}$ restricts to a quasi local even equivariant automorphism of $\caA$ with $\caO(f(r))$-tails.
\end{proof}

The following Proposition is a graded equivariant version of \cite{RWW2020}'s Proposition 5.3. It provides for any $\ep$-nearest neighbour automorphism an approximating automorphism that acts as a QCA near a given site.

\begin{proposition} \label{prop:localize locally}
	There are universal constants $C_1, \ep_1 > 0$ such that for any even equivariant $\ep$-nearest neighbour automorphism $\al$ of $\caA^{\vN}$ with $\ep \leq \ep_1$ and for any $n \in \Z$, there exists even equivariant automorphisms $\al_n$ of $\caA^{\vN}$ such that for $k = 0, 1, 2, 3$,
	\begin{align*}
		\al_n \big( \caA^{\vN}_{\leq 2n+2k-1} \big) &\subset \caA_{\leq 2n+2k}, \\
		\al_n \big( \caA_{\geq 2n+2k}^{\vN} \big) &\subset \caA_{\geq 2n+2k-1}^{\vN}, \\
		\norm{\al - \al_n} \leq C_1 \ep.
	\end{align*}

	In particular, denoting $\caB_m = \caA_{2m} \gotimes \caA_{2m+1}$ and $\caC_m = \caA_{2m-1} \gotimes \caA_{2m}$ we have
	\begin{align*}
		\al_n \big( \caB_m \big) \subset \caC_m \gotimes \caC_{m+1} \quad \text{for} \,\,\, m = n, n+1, n+2,
		\al_n^{-1} \big( \caC_{m} \big) \subset \caB_{m-1} \gotimes \caB_m \quad \text{for} \,\,\, m = n+1, n+2.
	\end{align*}

	If $\al$ is a quasi-local automorphism with $f$-tails then $\al_n$ is also quasi-local and
	\begin{equation}
		\norm{(\al - \al_n)|_{\caA_{\leq 2n-r-2}^{\vN} \gotimes \caA_{\geq 2n+r+6}^{\vN} }} = \caO(f(r)).
	\end{equation}
\end{proposition}

\begin{proof}
	We construct a sequence of eight automorphisms $\al_n^{i}$, $i = 1, \cdots, 8$ with $\al_n = \al_n^{8}$.

	Since
	\begin{equation}
		\al \big(  \caA_{\geq 2n}^{\vN} \big) \nsub{\ep} \caA_{\geq 2n-1}^{\vN} 
	\end{equation}
	Theorem \ref{thm:near inclusions} provides an even $G$-invariant unitary $u_1 \in \big( \al(\caA_{\geq 2n}^{\vN}) \cup \caA_{\geq 2n-1}^{\vN} \big)''$ with $\norm{u_1 - \I} \leq 12 \ep$ and such that
	\begin{equation}
		u_1 \al \big( \caA_{\geq 2n}^{\vN} \big) u_1^* \subset \caA_{\geq 2n-1}^{\vN}.
	\end{equation}
	We set $\al_n^{(1)} := \Ad_{u_1} \circ \al$, so $\al_n^{(1)}$ is an even equivariant automorphism that satisfies
	\begin{align*}
		\al_n^{(1)} \big( \caA_{\geq 2n}^{\vN} \big) &\subset \caA_{\geq 2n-1}^{\vN}, \\
		\norm{\al - \al_n^{(1)}} &= \caO(\ep).
	\end{align*}
	In particular, the automorphism $\al_n^{(1)}$ is $\caO(\ep)$-nearest neighbour and we can apply Lemma \ref{lem:localize to the right} to obtain an even equivariant automorphism $\al_n^{(2)} = \Ad_{u_2} \circ \al_n^{(1)}$, where $u_2$ is an even $G$-invaraint unitary in $\caA_{\geq 2n-1}^{\vN}$, such that
	\begin{align*}
		\al_n^{(2)} \big( \caA_{\geq 2n}^{\vN} \big) &\subset \caA_{\geq 2n-1}^{\vN}, \\
		\al_n^{(2)} \big( \caA_{\leq 2n-1}^{\vN} \big) &\subset \caA_{\leq 2n}^{\vN}, \\
		\norm{\al_n^{(1)} - \al_n^{(2)}} &= \caO(\ep).
	\end{align*}
	In particular, $\al_n^{(2)}$ is again a $\caO(\ep)$-nearest neighbour automorphism and we can apply Lemma \ref{lem:localize to the left} to obtain an even equivariant automorphism $\al_n^{(3)} = \al \circ \Ad_{u_3}$, where $u_3$ is an even $G$-invariant unitary element of $\caA_{\geq 2n}^{\vN}$, such that
	\begin{align*}
		\al_n^{(3)} \big( \caA_{\geq 2n+2}^{\vN} \big) &\subset \caA_{\geq 2n+1}^{\vN}, \\
		\norm{\al_n^{(2)} - \al_n^{(3)}} &= \caO(\ep)
	\end{align*}
	and such that $\al_n^{(3)}$ still statisfies the above locality properties satisfied by $\al_n^{(2)}$.

	Continuing in this manner, we next apply Lemma \ref{lem:localize to the right} on site $2n+2$ to obtain $\al_n^{(4)}$, then Lemma \ref{lem:localize to the left} on site $2n+1$ to obtain $\al_n^{(5)}$, then Lemma \ref{lem:localize to the right} on site $2n+4$ to obtain $\al_n^{(6)}$. again Lemma \ref{lem:localize to the left} on site $2n+3$ to obtain $\al_n^{(7)}$ and finally Lemma \ref{lem:localize to the right} on site $2n+5$ to obtain $\al_n^{(8)} =: \al_n$. The even equivariant automorphism $
	al_n$ then satisfies the required locality properties. In order for the Lemmas \ref{lem:localize to the left} and \ref{lem:localize to the right} to apply at each step, $\ep$ has to be sufficiently small, thus determining the universal constant $\ep_1$.

	If $\al$ is moreover quasi-local with $f$-tails then then, for $x \in \al \big(  \caA_{\geq 2n+r-1}^{\vN} \big)$ we have
	\begin{equation}
		x \nin{f(r)} \al \big( \caA_{\geq 2n}^{\vN} \big) \quad \text{and} \quad x \nin{f(r)} \caA_{\geq 2n-1}^{\vN}
	\end{equation}
	so Theorem \ref{thm:near inclusions} implies that the unitary $u_1$ satisfies $\norm{u_1 x u_1^* - x} \leq 46f(r) \norm{x}$ for all such $x$. This implies
	\begin{equation}
		\norm{\al - \al_n^{(1)}|_{\caA_{\geq 2n+r-1}^{\vN}}} = \caO(f(r)).
	\end{equation}

	Since $\al \big( \caA_{\leq 2n-r-2}^{\vN} \big) \nsub{f(r)} \caA_{\leq 2n - 1}^{\vN}$ we have for all $x \in \al \big( \caA_{\leq 2n -r -2}^{\vN} \big)$ and all $y \in \al \big( \caA_{\geq 2n}^{\vN} \big) \cup \caA_{\geq 2n-1}^{\vN}$ that $\norm{[x_{\pm}, y]_s} = \caO(f(r))$ so Theorem \ref{thm:near inclusions} implies that $\norm{u_1 x u_1^* - x} = \caO(f(r))$ for all such $x$. This in turn implies that
	\begin{equation}
		\norm{(\al - \al_n^{(1)})|_{\caA_{\leq 2n-r-2}^{\vN}}} = \caO(f(r)).
	\end{equation}

	These two norm-bounds together with Lemma \ref{lem:local errors control global errors} imply that
	\begin{equation}
		\norm{(\al - \al_n^{(1)})|_{ \caA_{\leq 2n-r-2}^{\vN} \gotimes \caA_{\geq 2n+r-1}^{\vN} }} = \caO(f{r)}.
	\end{equation}

	The unitaries $u_2, \cdots, u_8$ obtained from Lemmas \ref{lem:localize to the left} and \ref{lem:localize to the right} are such that
	\begin{equation}
		\norm{(\al_n^{(i)} - \al_n^{(i-1)})|_{ \caA_{\leq 2n - r + l_i}^{\vN} \gotimes \caA_{\geq 2n + r + r_i}^{\vN} }} = \caO(f{r)}
	\end{equation}
	for $i = 2, \cdots, 8$ and $l_2 = 0$, $l_3 = -1$, $l_4 = 2$, $l_5 = 1$, $l_6 = 4$, $l_7 = 3$ and $l_8 = 5$ while $r_2 = 1$, $r_3 = 1$, $r_4 = 3$, $r_5 = 3$, $r_6 = 5$, $r_7 = 5$ and $r_8 = 6$. Combining these bounds gives
	\begin{equation}
		\norm{(\al - \al_n)|_{\caA_{\leq 2n - r - 1}^{\vN} \gotimes \caA_{\geq 2n + r + 6}}} = \caO(f(r))
	\end{equation}
	thus concluding the proof.
\end{proof}

Next we state and prove a graded equivariant version of \cite{RWW2020}'s Proposition 5.4.
\begin{proposition} \label{prop:QCA approximation of ep-nn}
	There is a universal constant $\ep_2 > 0$ such that if $\al$ is an even equivariant $\ep$-nearest neighbour automorphism of $\caA$ with $\ep \leq \ep_2$, then then there is an even equivariant QCA $\beta$ of range 2 such that
	\begin{equation}
		\norm{(\al - \beta)|_{\caA_X}} = \caO(\ep \abs{X})
	\end{equation}
	for any $X \ssubset \Z$.
\end{proposition}

\begin{proof}
	By Lemma \ref{lem:vN extension} (i), the automorphism $\al$ extends to an even equivariant $12 \ep$-nearest neighbour automorphism of the von Neumann algebra $\caA^{\vN}$ which we continue to denote by $\al$. Let $C_1, \ep_1$ be the universal constants from Proposition \ref{prop:localize locally} and let $\ep \leq \ep_2 := \min \left\lbrace \frac{\ep_1}{12}, \frac{1}{48 C_1} \right\rbrace$. Recall the notation $\caB_n = \caA_{2n} \gotimes \caA_{2n+1}$ and $\caC_n = \caA_{2n-1} \gotimes \caA_{2n}$. We now apply Proposition \ref{prop:localize locally} to obtain for each $m \in \Z$ an automorphism $\al_m$ that satisfies
	\begin{align*}
		\al_m \big( \caB_n \big) \subset \caC_n \gotimes \caC_{n+1} \quad &\text{for} \,\,\, n = m, m+1, m+2, \\
		\al_m^{-1} \big( \caC_n \big) \subset \caB_{n-1} \gotimes \caB_n \quad &\text{for} \,\,\, n = m+1, m+2.
	\end{align*}
	Theorem \ref{thm:overlap factorization} then implies that the overlap algebras
	\begin{align*}
		\caL_n^{(m)} = \al_m \big( \caB_n \big) \cap \caC_n, \\
		\caR_{n-1}^{(m)} = \al_m \big( \caB_{n-1} \big) \cap \caC_n
	\end{align*}
	satisfy
	\begin{equation} \label{eq:factorization of C_n}
		\caC_n = \caR_{n-1}^{(m)} \gotimes \caL_n^{{m}}
	\end{equation}
	for $m = n-1, n-2$ and
	\begin{equation} \label{eq:factorization of B_n}
		\caB_n = \al_{n-1}^{-1} \big( \caL_n^{(n-1)} \big) \, \gotimes \, \al_{n-1}^{-1} \big( \caR_n^{(n-1)} \big)
	\end{equation}
	for all $n$.
	
	From $\norm{\al_{n-1} - \al_{n-2}} \leq \norm{\al_{n-1} - \al} + \norm{\al_{n-2} - \al} \leq 2 C_1 \ep$ (using Proposition \ref{prop:localize locally}), the locality properties of $\al_{n-1}$ and $\al_{n-2}$, and Lemma \ref{lem:QCA stability} we get an even $G$-invariant unitary $u_n \in \caC_n$ such that $u_n \caL_n^{(n-1)} u_n^* = \caL_n^{(n-2)}$ with $\norm{u - \I} \leq 72 C_1 \ep$. We now put
	\begin{equation}
		\beta_n : \caB_n \rightarrow \caC_n \gotimes \caC_{n+1} : x \mapsto \beta_n(x) = u_n \al_{n-1}(x) u_n^*.
	\end{equation}
	these are injective, equivariant super *-homomorphisms, and using Eq. \eqref{eq:factorization of B_n} we get
	\begin{equation} \label{eq:factorization of beta_n B_n}
		\beta_n \big( \caB_n \big) = u_n \left( \caL_n^{(n-1)} \gotimes \caR_n^{(n-1)} \right) u_n^* = \caL_n^{(n-2)} \gotimes \caR_n^{(n-1)},
	\end{equation}
	where we used that $u_n$ is an even element of $\caC_n$ and therefore commutes with $\caR_n^{(n-1)}$. We conclude that $\beta_b \big( \caB_n \big)$ and $\beta_n \big( \caB_m \big)$ supercommute whenever $n \neq m$, so the simultaneous action of the $\beta_n$ defines an injective, equivariant super *-homomorphism on $\caA$. This homomorphism is moreover surjective, indeed $\beta_{n-1}\big( \caB_{n-1} \big) \gotimes \beta_{n}\big( \caB_n \big) \supset \caR_{n-1}^{(n-2)} \gotimes \caL_{n}^{(n-2)} = \caC_n$, so the image of $\beta$ cotains each $\caC_n$, and therefore the entire quanrum chain $\caA$. We conclude that $\beta$ is an even equivariant automorphism of $\caA$. By construction, $\beta$ is a QCA of range 2.

	Finally, if $x$ is supported on a single site, then $x \in \caB_n$ for some $n$, so
	\begin{align*}
		\norm{\beta(x) - \al(x)} &= \norm{\beta_n(x) - \al(x)} = \norm{\Ad_{u_n}\big( \al_{n-1}(x) \big) - \al(x)  } \\
					 &\leq \norm{\Ad_{u_n}\big( \al_{n-1}(x) - \al(x) \big)} + \norm{\Ad_{u_n}\big( \al(x) \big) - \al(x)} \\
					 & \leq \norm{\Ad_{u_n}} + \norm{(\al_{n-1} - \al)} \leq 2\norm{u_n - \I} + \norm{\al_{n-1} - \al} \\
					 &= \caO(\ep).
	\end{align*}
	The final claim then follows from Lemma \ref{lem:local errors control global errors}
\end{proof}

We finally prove the following graded equivariant extension of \cite{RWW2020}'s Theorem 5.5:

\begin{proposition} \label{prop:QCA approximation}
	For any quasi-local even equivariant automorphism $\al$ with $f(n)$-tails on a quantum chain $\caA$ there is a sequence of even equivariant QCAs $\{\beta_j\}_{j = 1}^{\infty}$ of range $2j$ converging strongly to $\al$ such that for any $X \ssubset \Z$,
	\begin{equation} \label{eq:QCA approximation}
		\norm{(\al - \beta_j)|_{\caA_X}} = \caO \left( f(j)  \min \left\lbrace \abs{X}, \ceil{\frac{\diam(X)}{j}}  \right\rbrace \right).
	\end{equation}
\end{proposition}

\begin{proof}
	If we define new sites by grouping intervals of length $j$ together, we have that $\al$ is $\ep_j$-nearest neighbour with $\ep_j = \caO(f(j))$ with respect to the coarse grained local structure. For $j$ sufficiently large, Proposition \ref{prop:QCA approximation of ep-nn} applies to yield an even equivariant QCA $\beta_j$ wich is of range 2 in the coarse grained lattice, and therefore of range 2j in the original lattice. This QCA moreover satisfies $\norm{(\al - \beta)|_{\caA_X}} = \caO \big(f(j)m \big)$ for any region $X$ consisting of $m$ coarse-grained sites. Since any finite set $X$ can be covered by at most $\abs{X}$ coarse grained sites, we get
	\begin{equation}
		\norm{(\al - \beta_j)|_{\caA_X}} = \caO \big(  f(j) \abs{X} \big)
	\end{equation}
	for any $X \ssubset \Z$. Moreover, a finite set $X$ of diameter $\diam(X)$ can be covered by at most $\ceil{\frac{\diam(X)}{j}}$ coarse grained sites, yielding
	\begin{equation}
		\norm{(\al - \beta_j)|_{\caA_X}} = \caO \left(  f(j) \ceil{\frac{\diam(X)}{j}} \right).
	\end{equation}
	Together, these bounds yield Eq. \eqref{eq:QCA approximation}.

	For $j$ too small for Proposition \ref{prop:QCA approximation of ep-nn} to apply, let $\beta_j$ be an arbitrary even equivariant QCA (\eg $\id$), then the bound holds with $\norm{\al - \beta_j} \leq 2$.

	To see that the sequence $\{\beta_j\}$ converges strongy to $\al$, take $x \in \caA$ and let $\{ x_n \}$ be a sequence of operators with $x_n$ supported on an interval $I_n$ of length $n$, converging in norm to $x$. Then
	\begin{align*}
		\limsup_{j \uparrow \infty} \, \norm{\al(x) - \beta_j(x)} \, &\leq \limsup_{j \uparrow \infty} \, \left( \norm{\al(x) - \al(x_n)} + \norm{\al(x_n) - \beta_j(x_n)} + \norm{\beta_j(x_n) - \beta_j(x)} \right) \\
									     & \leq 2 \norm{x - x_n} + \limsup_{j \uparrow \infty} \, \norm{\al(x_n) - \beta_j(x_n)} = 2 \norm{x - x_n}.
	\end{align*}
	This inequality holds for all $n$, so we coclude that $\limsup_{j \uparrow \infty} \, \norm{\al(x) - \beta_j(x)} = 0$ for any $x \in \caA$, \ie $\beta_j \rightarrow \al$ strongly.
\end{proof}

\subsection{Index for quasi-local automorphisms}

\begin{definition} \label{def:index for quasi-local automorophisms}
	The index of a quasi-local even equivariant automorphism $\al$ is defined by
	\begin{equation}
		\ind(\al) := \lim_{j \uparrow \infty} \ind(\beta_j) \in (\Q \cup \sqrt{2}\Q) \times \Hom(G, \Z_2) \times H^2(G, U(1))
	\end{equation}
	where $\{\beta_j\}$ is any sequence of even equivariant QCAs approximating $\al$ as in Eq. \eqref{eq:QCA approximation}.
\end{definition}

The existence of such an approximating sequence is guaranteed by Proposition \ref{prop:QCA approximation}. The existence of the limit and its independence of the approximating sequence of QCAs is shown in the following proposition, which extends part of Theorem 5.7 of \cite{RWW2020} to the graded equivariant case.

\begin{proposition} \label{prop:index for quasi-local automorphisms}
	Let $\al$ be a quasi-local even equivariant automorphism with $f$-tails acting on a quantum chain $\caA$. Let $\{\beta_j\}_{j = 1}^{\infty}$ be a sequence of even equivariant QCAs approximating $\al$ as in Eq. \eqref{eq:QCA approximation}. Then
	\begin{enumerate} [label=(\roman*)]
		\item $\ind(\beta_j)$ is constant for $j \geq j_0$ for some $j_0$ depending only on $f$. In particular, the limit $\,\, \lim_{j \uparrow \infty} \ind(\beta_j)$ exists. Moreover, this limit does not depend on the approximating sequence $\{\beta_j\}$.

		\item There is a constant $n_1$ depending only on $f$ such that the following holds. Let $\al'$ be a quasi-local even equivariant automorphism with $f(n)$-tails. If $\norm{(\al - \al')|_{\caA_X}} \leq 1/48$ for any interval $X \subset \Z$ of length greater than or equal to $n_1$, then $\ind(\al) = \ind(\al')$.
	\end{enumerate}
\end{proposition}

\begin{proof}
	The QCAs $\beta_{j}$ and $\beta_{j+1}$ each satisfy the bound \eqref{eq:QCA approximation} so for any finite interval $X \ssubset \Z$ we have
	\begin{equation}
		\norm{(\beta_j - \beta_{j+1})|_{\caA_X}} \leq \norm{ (\beta_j - \al)|_{\caA_X} } + \norm{(\beta_{j+1} - \al)|_{\caA_X}} = \caO \left( f(j) \ceil{\frac{\diam{X}}{j}} \right).
	\end{equation}
	Both QCAs are in nearest neighbour form w.r.t. a coarse graining in blocks of size $2(j-1)$. Take $X$ to be the union of two neighbouring blocks, then $\norm{(\beta_j - \beta_{j+1})|_{\caA_X}} = \caO(f(j))$ and it follows from Lemma \ref{lem:QCA stability} that $\ind(\beta_j) = \beta_{j+1}$ for all $j \geq j_0$ where $j_0$ is the smallest value of $j$ for which $f(j_0) < 1/24$. This shows that any such sequence $\{\beta_j\}$ is eventually constant and therefore has a limit. The fact that this limit does not depend on the particular approximating sequence $\{\beta_j\}$ actually follows from (ii), so let us now prove (ii).

	Let $\{\beta'_j\}$ be a sequence of even equivariant QCAs approximating $\al'$ as in Eq. \eqref{eq:QCA approximation}. We've just shown that both sequences $\ind(\beta_j)$ and $\ind(\beta'_j)$ are constant for $j \geq j_0$, to show that they are in fact equal it is sufficient to show $\ind(\beta_j) = \ind(\beta'_j)$ for any $j \geq j_0$. For any finite interval $X \ssubset \Z$ we have from \eqref{eq:QCA approximation} that
	\begin{align*}
		\norm{(\beta_j - \beta'_j)|_{\caA_X}} &\leq \norm{(\beta_j - \al)|_{\caA_X}} + \norm{(\al - \al')|_{\caA_X}} + \norm{(\beta'_{j} - \al')|_{\caA_X}} \\
						      & \leq \caO \left(f(j) \frac{\diam(X)}{j} \right) + \norm{(\al - \al')|_{\caA_X}}.
	\end{align*}
	The QCAs $\beta_j$ and $\beta'_j$ are nearest neighbour w.r.t. a coarse graining in blocks of size $2j$, so Lemma \ref{lem:QCA stability} implies $\ind(\beta_j) = \ind(\beta'_{j})$ provided that $\norm{(\beta_j - \beta'_j)|_{\caA_X}} \leq 1/24$ for an interval $X$ consisting of two such coarse grained blocks, \ie for any interval of length $4j$. Fors such $X$, there is a $j_1 \geq j_0$ such that the quantity $\caO \left(f(j) \frac{\diam(X)}{j} \right) = \caO((f(j)))$ is less than $1/48$. If we take $n_1 = 4j_1$ then the assumption moreover gives $\norm{(\al - \al')|_{\caA_X} } \leq 1/48$ for any interval $X$ such that $\abs{X} \geq n_1$. We then get the required $\norm{(\beta_j - \beta'_j)|_{\caA_X}} \leq 1/24$ for all $j \geq j_1$, so the sequences $\ind(\beta_j)$ and $\ind(\beta'_j)$ have the same limit.
\end{proof}

\subsection{Basic properties of the index for quasi-local automorphisms}

\begin{lemma} \label{lem:inverse, composition, and stacking of quasi-local automorphisms}
	If $\al$ and $\beta$ are quasi-local even equivariant automorphisms then
	\begin{equation}
		\ind(\al \gotimes \beta) = \ind(\al) \cdot \ind(\beta).
	\end{equation}
	If $\al$ and $\beta$ are defined on the same quantum chain then
	\begin{equation}
		\ind(\al \circ \beta) = \ind(\al) \cdot \ind(\beta).
	\end{equation}
	In particular, using that $\ind(\id) = (1, 0, \boe)$ is trivial, we have
	\begin{equation}
		\ind(\al^{-1}) = \ind(\al)^{-1}.
	\end{equation}
\end{lemma}

\begin{proof}
	Suppose $\al$ acts on a quantum chain $\caA_1$ and is quasi-local with $f_1$-tails, while $\beta$ acts on a quantum chain $\caA_2$ and is quasi-local with $f_2$-tails. Then $\al \gotimes \beta$ acts on $\caA_1 \gotimes \caA_2$ and is quasi-local with $f$-tails where $f(j) = \max \lbrace f_1(j), f_2(j) \rbrace$.

	Let $\{\al_j\}$ and $\{\beta_j\}$ be sequences of even equivariant QCAs approximating $\al$ resp. $\beta$ as in Eq. \eqref{eq:QCA approximation}. By Proposition \ref{prop:index for quasi-local automorphisms} (i), the sequences $\{\ind(\al_j)\}$ and $\{\ind(\beta_j)\}$ are eventually constant and equal by definition to $\ind(\al)$ and $\ind(\beta)$ respectively.

	By Proposition \ref{prop:QCA composition and stacking}, the QCAs $\al_j \gotimes \beta_j$ have index $\ind(\al_j \gotimes \beta_j) = \ind(\al_j) \cdot \ind(\beta_j)$. Moreover, for any finite set $X \ssubset \Z$ we have
	\begin{equation}
		\norm{\big((\al_j \gotimes \beta_j) - \al \gotimes \beta \big)|_{(\caA_1 \gotimes \caA_2)_X}} \leq \norm{(\al_j - \al)|_{(\caA_1)_X}} + \norm{(\beta_j - \beta)|_{(\caA_2)_X}} = \caO \left( f(j) \min \left\lbrace  \abs{X}, \frac{\diam(X)}{j} \right\rbrace \right),
	\end{equation}
	\ie the sequence of even equivariant QCAs $\{ \al_j \gotimes \beta_j \}$ approximates $\al \gotimes \beta$ in the sense of Eq. \eqref{eq:QCA approximation}. It then follows from Proposition \ref{prop:index for quasi-local automorphisms} (i) that the sequence $\{ \ind(\al_j \gotimes \beta_j) \}$ is eventually equal to $\ind(\al \gotimes \beta)$. Combining the above findings, we get for large enough $j$ that
	\begin{equation}
		\ind(\al \gotimes \beta) = \ind(\al_j \gotimes \beta_j) = \ind(\al_j) \cdot \ind(\beta_j) = \ind(\al) \cdot \ind(\beta),
	\end{equation}
	as required.

	If $\al$ and $\beta$ are defined on the same quantum chain $\caA$, then $\al \circ \beta$ is a quasi-local automorphism on $\caA$ with $f$-tails where $f(j) = \min \{ f_1(j_1) + f_2(j_2) \, : \, j_1 + j_2 = j \}$. One then shows that the sequence of QCAs $\{\al_j \circ \beta_j\}$ approximates $\al \circ \beta$ in the sense of Eq. \ref{eq:QCA approximation}, and the desired $\ind(\al \circ \beta) = \ind(\al) \cdot \ind(\beta)$ follows in the same way as above, using that Proposition \ref{prop:QCA composition and stacking} implies $\ind(\al_j \circ \beta_j) = \ind(\al_j) \cdot \ind(\beta_j)$.
\end{proof}

\subsection{Completeness} \label{sec:quasi-local completeness}

We begin with a lemma that says that an even equivariant QCA that is `locally small' is $G$-equivalent to the identity.
\begin{lemma} \label{lem:small QCA is G-equivalent to id}
	Let $\beta$ be an even equivariant QCA of range $R$ on a quantum chain $\caA$. Suppose there is $\ep < 1/145$ such that $\norm{\beta(x) - x} \leq \ep\norm{x}$ for all $x \in \caA$ supported on an interval of length at most $2R$. Then $\beta$ is $G$-equivalent to the identity through a path of even equivariant QCAs of range $3R$.
	
	Moreover, this path may be taken to be generated by a time-dependent Hamiltonian interaction $[0, 1] \ni t \mapsto H_t$ such that there are parititions $\{I^{(1)}_i\}$ and $\{ I^{(2)}_i\}$ of $\Z$ with blocks of size at most $2R$ such that for $t \in [0, 1/2]$ we have $H_t(I^{(1)}_i) = h^{(1)}_i$ and for $t \in (1/2, 1]$ we have $H_t(I^{(2)}_i) = h^{(2)}_i$ for even $G$-invariant self-adjoint elements $h^{(1)}_i \in \caA_{I^{(1)}_i}$ and $h^{(2)}_i \in \caA_{I^{(2)}_i}$ satisfying $\norm{h^{(1)}_i}, \norm{h^{(2)}_i} \leq 18 \ep$.
\end{lemma}

\begin{proof}
	Coarse grain the chain in blocks of size $R$ so $\beta$ is in nearest-neighbour form. Denote by $\caL_n$ and $\caR_n$ the left and right overlap algebras. By Lemma \ref{lem:QCA stability} the central simple $G$-systems $(\caL_n, \rho|_{\caL_n})$ and $(\caA_{2n}, \rho|_{\caA_{2n}})$ are isomorphic through conjugation with an even $G$-invariant unitary $u_n \in \caC_n$ with $\norm{u_n - \I} \leq 36 \ep$. Since $\caC_n = \caR_{n-1} \gotimes \caL_n$ it follows that also $(\caR_{n-1}, \rho|_{\caR_{n-1}})$ and $(\caA_{2n-1}, \rho|_{\caA_{2n-1}})$ are isomorphic central simple $G$-systems with isomorphism given by conjugation with $u_n$. Denote by $\Phi_n = \Ad_{u_n}$ and let $\Phi$ be the automorphism of $\caA$ defined by the simultaneous action of all the $\Phi_n$. Set $\Psi = \beta \circ \Phi$, then $\Psi(\caB_n) = \Phi_n( \caL_n \gotimes \caR_n ) = \caA_{2n} \gotimes \caA_{2n+1} = \caB_n$. Denote by $\Psi_n = \Psi|_{\caB_n}$. We have for all $x \in \caB_n$
	\begin{align*}
		\norm{\Psi_n(x) - x} &\leq \norm{ (u_n \gotimes u_{n+1}) \beta(x) (u_n \gotimes u_{n+1})^* - \beta(x) } + \norm{\beta(x) - x} \\
				     &\leq 2 \norm{ u_n \gotimes u_{n+1} - \I } \norm{x} + \ep \norm{x} \\
				     &\leq (4 \cdot 36 + 1) \ep \norm{x} \leq 145 \ep \norm{x} < 1.
	\end{align*}
	It then follows from Proposition \ref{prop:making homomorphisms inner} that $\Psi_n$ is given by conjugation with an even $G$-invariant unitary $v_n \in \caB_n$ such that $\norm{v_n - \I} \leq \sqrt{2} \ep$.

	Any even $G$-invariant unitary $u$ can be continuously deformed to $\I$ through a path of even $G$-invariant unitaries. Indeed, we can write $u = \ed^{\iu h}$ with $h$ an even $G$-invariant self-adjoint operator such that $\norm{h} \leq \norm{u - \I}/2$, then $[0, 1] \ni t \mapsto \ed^{\iu t h}$ is the desired path. It follows that all the $\Phi_n^{-1}$ and all the $\Psi_n$ are connected to $\id$ along a continuous path of even equivariant automorphisms generated by even $G$-invariant $h^{(1)}_n \in \caC_n$ and $h^{(2)}_n \in \caB_n$ respectively, satisfying $\norm{h^{(1)}_n}, \norm{h^{(2)}_n} \leq 18 \ep$ for all $n$. Doing these deformations simultaneously yields strongly continuous paths $[0,1] \ni t \mapsto \Phi(t)^{-1}$ from $\id$ to $\Phi^{-1}$ and $[0, 1] \ni t \mapsto \Psi(t)$ from $\id$ to $\Psi$ generated by Hamiltonian interactions $H^{(1)}(\caC_n) = h^{(1)}_n$ and $H^{(2)}(\caB_n) = h^{(2)}_n$ respectively. The strongly continuous path $[0, 1] \ni t \mapsto \beta(t)$ generated by the time-dependent Hamiltonian interaction $H_t$ such that $H_t = 2 H^{(1)}$ for $t \in [1, 1/2]$ and $H_t = 2 H^{(2)}$ for $t \in (1/2, 1]$ then interpolates between $\id$ and $\beta$ along even equivariant QCAs of range 3 in the coarse grained chain, hence of range $3R$ in the original chain.
\end{proof}

We now show that any quasi-local even equivariant automorphism is $G$-equivalent to an even equivariant QCA.

\begin{proposition} \label{prop:quasi-local is G-equivalent to QCA}
	Any quasi-local even equivariant autmorphism $\al$ with $f$-tails on a quantum chain $\caA$ is $G$-equivalent to an even equivariant QCA of range $R_0$ depending only on $f$, through a path of quasi-local automorphisms with $g(n) = f(Cn)$-tails, where $C$ is a universal constant.
\end{proposition}

\begin{proof}
	Let $\{\beta_j\}$ be a sequence of even equivariant QCAs approximating $\al$ as in Eq. \eqref{eq:QCA approximation}. The existence of such a sequence is guaranteed by Proposition \ref{prop:QCA approximation}. Let $\gamma_k = \beta_{2^{k}} \circ \beta_{2^{k-1}}^{-1}$. These are even equivariant QCAs of radius less than $r_k = 2^{k+2}$ and
	\begin{align*}
		\norm{ (\gamma_k - \id)|_{\caA_X} } &= \norm{ ( \beta_{2^{k}} - \beta_{2^{k-1}} )|_{\caA_X} } \\
						    &= \norm{ (\beta_{2^{k}} - \al)|_{\caA_X} } + \norm{ (\beta_{2^{k-1}} - \al)|_{\caA_X} } \\
					 &= \caO \left( f(2^k) \right)
	\end{align*}
	for any $X \ssubset \Z$ with $\diam(X) \leq 2 \times r_k = 2^{k+3}$, where the constant implicit in the $\caO$-notation depends only on $f$. The reason for choosing $\diam(X)$ less then or equal to twice the radius of $\gamma_k$ is to apply Lemma \ref{lem:small QCA is G-equivalent to id}, which gives us a $k_0$ such that for all $k \geq k_0$ we have that $\gamma_k$ is $G$-equivalent to the identity automorphism through a path $[0, 1] \ni t \mapsto \gamma_{k, t}$ of even equivariant QCAs of range at most $3 \times 2^{k + 2}$, generated by a time-dependent Hamiltonian interaction $[0, 1] \ni t \mapsto H_{k, t}$ such that $H_{k, t}(X) \neq 0$ only if $\diam(X) \leq 2 r_k = 8 \times 2^k$ and $\norm{H_{k, t}(X)} \leq \caO \big( f(2^k) \big)$.

	Concatenate the Hamiltonian interactions $[0, 1] \ni t \mapsto H_{k, t}$ to form $[0, \infty) \ni t \mapsto \widetilde H_t$ defined by
	\begin{equation}
		\widetilde H_t = H_{\floor{t} + k_0, t - \floor{t} }.
	\end{equation}
	The Hamiltonian interaction $\widetilde H_t$ generates a strongly continuous path of even equivariant QCAs $\tilde \al_t$ with $\tilde \al_0 = \beta_{2^{k_0}}$ such that for integer $t \geq 1$ we have
	\begin{equation}
		\tilde \al_t = \gamma_{k_0 + t} \circ \cdots \gamma_{k_0} \circ \beta_{2^{k_0}} = \beta_{2^{k_0 + t}},
	\end{equation}
	and for all $t$ we have $\widetilde H_t(X) \neq 0$ only if $\diam{X} \leq 8 2^{\floor{t} + k_0}$ and $\norm{H_t(X)} \leq \caO \big( f(2^{\floor{t} + k_0}) \big)$.

	Since $\beta_j \rightarrow \al$ strongly, we have that $\tilde \al_t \rightarrow \al$ stronly as $t \uparrow \infty$. We speed up to evolution so that we reach $\al$ in finite time. Let
	\begin{equation}
		t ( \tilde t ) = \frac{\tilde t}{\tilde t + 1}, \quad \tilde t (t) = \frac{t}{1 - t}
	\end{equation}
and put $\al_t := \tilde \al_{\tilde t (t)}$. Since $t \mapsto \tilde t$ is continuous for all $t \in [0, 1)$ we have that $t \mapsto \al_t$ is strongly continuous for all $t \in [0, 1)$. Moreover, $\al_{t(k)} = \beta_{2^{k_0 + k}}$ for all $k \in \N$ so $\al_t \rightarrow \al$ strongly as $t \uparrow 1$, so $[0, 1] \ni t \mapsto \al_t$ is strongly continuous for all $t \in [0, 1]$ and interpolates from the range $R_0 = 2 k_0$ even equivariant QCA $\al_0 = \beta_{2^{k_0}}$ to $\al_1 = \al$. Note that the range $R_0 = 2 k_0$ depends only on $f$.

	Finally, $t \mapsto \al_t$ is generated by Hamiltonian interaction
	\begin{equation}
		H_t = \frac{\dd \tilde t}{\dd t} \widetilde H_{\tilde t} = (\tilde t + 1 )^2 \widetilde H_{\tilde t}.
	\end{equation}
	It then follows from the properties of 	$\widetilde H_{\tilde t}$ that for all $t \in [0, 1]$ there is an integer $l = 2^{\floor{\tilde t} + k_0}$ such that $H_t(X) \neq 0$ only if $\diam(X) \leq 16 l$ and $\norm{H_t(X)} \leq \caO \big( f(l) \log(l) \big)$.

	We finally show that $\al_t$ has $g(n) = \caO\big( f(Cn) \big)$	tails for all $t \in [0, 1]$. This is the case if and only if $\tilde \al_{\tilde t}$ has such tail bounds for all $\tilde t \in [0, \infty)$, so we'll prove that instead. (Note that for $\al_1 = \al$ we already have $f$-tails).

	Let $\tilde t = l + s$ where $l$ is an integer and $s \in [0, 1)$. Then $\tilde \al_{\tilde t}$ is a QCA of range $3 r_l$. Pick $n > 3 r_l$, then there is a $k$ such that $3 r_k \leq n \leq 3 r_{k+1}$. Fix an interval $X \subset \Z$.

	If $k \geq l$ then $\tilde \al_{\tilde t}(\caA_X) \subset \caA_{X^{(3 r_l)}} \subset \caA_{X^{(n)}}$.

	If $k > l$, then write $X = X_{i} + X_{e}$ where $X_i = \{ n \in X \, : \, \dist(x, X^c) > 3 r_l \}$ contains the sites in $X$ that are futher than  $3 r_l$ away from the complement of $X$ and $X_e$ is the complement of $X_i$ in $X$, \ie the sites in $X$ a distance at most $3 r_l$ from $X^c$. Then $\tilde \al_{\tilde t}(\caA_{X_i}) \subset \caA_{X}$ and since $X_e$ consists of at most two intervals of size $3 r_l$ it follows from Lemma \ref{lem:local errors control global errors} and
	\begin{equation}
		\norm{(\tilde \al_{\tilde t} - \al)|_{\caA_X}} \leq \norm{ (\gamma_{l, s} - \id)|_{\caA_X} } + \norm{(\al - \beta_{2^{k_0 + l}}} = \caO \big( f(2^l) \big)
	\end{equation}
	for $\diam(X) \leq r_l$ that $\norm{(\al - \tilde \al_{\tilde t})|_{\caA}} = \caO \big( f(2^l) \big)$ where the constant implicit in the $\caO$-notation is universal. Since $\al$ has $f$-tails by assumption, $\al(\caA_{X_e}) \nsub{\caO(f(n))} \caA_{X^{(n)}}$, and since $n < 3 r_l < 2^{l+4}$ we get $\tilde \al_{\tilde t}(\caA_{X_e}) \nsub{\caO\big( f(n/16) \big)} \caA_{X^{(n)}}$. Since $(\caA_{X_e}^{\vN} \cup \caA_{X_i}^{\vN})'' = \caA_{X}^{\vN}$ and $\caA_{X^{(n)}}^{\vN}$ are $\theta$-hyperfinite von Neumann algebras, Proposition \ref{prop:simultaneous near inclusions of theta-hyperfinite vN algebras} implies
	\begin{equation}
		\tilde \al_{\tilde t}(\caA_X) \nsub{\caO(f(n/16) )} \caA_{X^{(n)}}
	\end{equation}
	where the constant implicit in the $\caO$-notation is universal. This concludes the proof.
\end{proof}

\begin{proposition} \label{prop:decoupling trivial quasi-local automorphisms}
	A quasi-local even equivariant automorophism $\al$ with $f$-tails on a quantum chain $\caA$ is stably $G$-equivalent to a decoupled even equivariant QCA if and only if $\ind(\al) = (1, 0, \boe)$ is trivial. The equivalence may be realized through stacking with the opposite quantum chain $\overline{\caA}$, the regular chain $\caA^{\reg}$ and the fermion chain $\caA^{\fer}$, and going through a path of quasi-local even equivariant automorphisms with $g(n) = f(C n)$-tails, where $C$ is a universal constant.
	Moreover, if the quantum chain $\caA$ is trivial in the sense of Definition \ref{def:trivial chain}, then the stacking with the conjugate chain $\bar \caA$ can be omitted. If $\caA$ has trivial symmetry then the stacking with $\caA^{\reg}$ can be ommited. If $\caA$ has trivial grading then the stacking with $\caA^{\fer}$ can be omitted.
\end{proposition}

\begin{proof}
	This is an immediate consequence of Propositions \ref{prop:quasi-local is G-equivalent to QCA} and \ref{prop:decoupling trivial QCA} and the stability of the index.
\end{proof}

We finally prove the main Theorem \ref{thm:Main Theorem}.

\begin{proofof}[Theorem \ref{thm:Main Theorem}]
	The behaviour of $\ind$ under stacking and composition is shown in Lemma \ref{lem:inverse, composition, and stacking of quasi-local automorphisms}. The properties of the binary operation are shown in Lemma \ref{lem:abelian group}. The fact that for each $(d, \zeta, \bnu) \in \Q \cup \sqrt{2} \Q \times \Hom(G, \Z_2) \times H^2(G, U(1))$ there is an even equivariant QCA $\al$ with $\ind(\al)$ is shown in Section \ref{sec:examples}.

	Let us now show the equivalence of statements (i) and (ii) in the Theorem. First, suppose even equivariant quasi-local automorphisms $\al$ and $\beta$ are stably $G$-equivalent in a robust class $\scrA_Q$. This means that after stacking with auxiliary decoupled automorphisms, they can be deformed into eachother along a strongly continuous path in $\scrA_{Q}$. By the triviality of decoupled automorphisms (Lemma \ref{lem:triviality of circuits}), the behaviour of the index under stacking (Lemma \ref{lem:inverse, composition, and stacking of quasi-local automorphisms}) and the stability of the index under strongly continuous deformation, which follows immediately from Proposition \ref{prop:index for quasi-local automorphisms}, we find that $\al$ and $\beta$ must have the same index.

	Conversely, suppose $\al$ and $\beta$ belong to robust class $\scrA_Q$ and have the same index. It follows from Lemma \ref{lem:inverse, composition, and stacking of quasi-local automorphisms} and Proposition \ref{prop:decoupling trivial quasi-local automorphisms} that $\beta^{-1} \gotimes \beta$ is stably $G$-equivalent along a path in $\scrA_Q$ to a decoupled automorphism $\gamma_1$. We have therefore that $\al \gotimes \gamma_1$ is stably $G$-equivalent along a path in $\scrA_Q$ to $\al \gotimes \beta^{-1} \gotimes \beta$. Again by Lemma \ref{lem:inverse, composition, and stacking of quasi-local automorphisms} and Proposition \ref{prop:decoupling trivial quasi-local automorphisms}, the automorphism $\al \gotimes \beta^{-1}$ is stably $G$-equivalent along a path in $\scrA_Q$ to a decoupled automorphism $\gamma_2$, so $\al \gotimes \beta^{-1} \gotimes \beta$ is stably $G$-equivalent along a path in $\scrA_Q$ to $\gamma_2 \gotimes \beta$. Concatenating these paths, we find that $\al \gotimes \gamma_1$ is stably $G$-equivalent to $\gamma_2 \gotimes \beta$. Finally, since $\gamma_1$ and $\gamma_2$ are decoupled, we conclude that $\al$ is stably $G$-equivalent to $\beta$ along a path in $\scrA_Q$.
\end{proofof}

\subsection{Relation to the cohomology index of SPT states on spin chains}

Let $\caA$ be a spin chain, \ie a quantum chain with trivial grading, equipped with a $G$-symmetry $\rho$. A state $\psi$ on $\caA$ is $G$-invariant if $\psi \circ \rho^g = \psi$ for all $g \in G$. In \cite{KSY2020} a second cohomology index is associated to any pure $G$-invariant state $\psi$ that satisfies an additional condition called invertibility, which essentially says that there is sufficient decay of correlations (see aslo \cite{CGW2011, Ogata2021}).

We will start from a $G$-invariant pure product state $\psi_0$ (\ie such that $\psi_0|_{\caA_n}$ is pure for every $n \in \Z$) and consider a $G$-invariant pure state $\psi := \psi_0 \circ \al$ for a quasi-local equivariant automorphism $\al$. If $\al \in \scrA_{n^{-\infty}}$ then Theorem 5.11 of \cite{RWW2020} implies that it is a composition of a shift and a time-one evolution generated by an interaction $\Phi_t$ such that for all $t \in [0, 1]$ there is a $k \in \N$ such that $\Phi_t(X)$ is non vanishing only for $\abs{X} \leq 16k$, for each $x \in \Z$ there are at most two sets $X \ni x$ such that $\Phi_t(X) \neq 0$ and $\norm{\Phi_t(X)} = \caO(k^{-\infty} \log(k)) = \caO(k^{-\infty})$. It follows that $\psi$ is a $G$-invariant invertible state according to \cite{KSY2020}.

We denote by $\rho^g_{\geq n}$ the automorphism of $\caA_{\Z}$ that acts as $\rho^{g}$ on $\caA_{\geq n}$ and as $\id$ on $\caA_{< n}$. The following Proposition is proven in Section 3 of \cite{KSY2020}.
\begin{proposition} \label{prop:KSY2020 cohomology index}
	Let $(\caH, \pi, \Omega)$ be the GNS data of $\psi$. There is a projective unitary representation $g \mapsto V_{\geq n}(g) \in \caB(\caH)$ of $G$ such that
	\begin{equation} \label{eq:KSY2020 projective rep}
		V_{\geq n}(g) \pi(x) V_{\geq n}(g)^* = \pi \left( \rho^g_{\geq n}(x) \right)
	\end{equation}
	for all $x \in \caA_{\Z}$. The unitaries $V_{\geq n}(g)$ are uniquely determined by \eqref{eq:KSY2020 projective rep} up to a phase and therefore determine a unique cohomology class $[V_{\geq n}] \in H^2(G, U(1))$.

	Moreover, the cohomology class depends neither on $n$ nor on the choice of $G$-invariant product state $\psi_0$.
\end{proposition}

\begin{proof}
	All statements follow immediately from the analysis in Section 3 of \cite{KSY2020} except for the independence from the choice of $\psi_0$. Let $\tilde \psi_0$ be another $G$-invariant product state. It is shown in appendix A of \cite{KSY2020} that there exists auxilliary spin chains $\caA'$ and $\caA''$ whose local dimensions grow at most polynomially with $\abs{n}$ (the distance from the origin) and $G$-invariant product states $\psi'$,$\psi''$ on $\caA', \caA''$ respectively such that $\psi_0 \otimes \psi' = (\tilde \psi_0 \otimes \psi'') \circ \beta$ where $\beta$ is an equivariant automorphism generated for time 1 by a $G$-invariant interaction with $\caO(n^{-\infty})$ tails. It follows that
	\begin{equation}
		(\psi_0 \circ \al \otimes \psi') = (\tilde \psi_0 \circ \al \otimes \psi'') \circ (\al^{-1} \otimes \id) \circ \beta \circ (\al \otimes \id).
	\end{equation}
	The automorphism $(\al^{-1} \otimes \id) \circ \beta \circ (\al \otimes \id)$ is still equivariant and generated by a $G$-invariant interaction with $\caO(n^{-\infty})$ tails. It is shown in Section 3 of \cite{KSY2020} that the cohomology class is invariant under such automorphisms. Moreover, the cohomology class is multiplicative under stacking (Corollary 2.1 of \cite{KSY2020}) and the cohomology class of a product state is the trivial class $\boe$. We conclude that $\psi_0 \circ \al$ and $\tilde \psi_0 \circ \al$ yield the same cohomology class.
\end{proof}

This allows us to associate a cohomology class to the automorphism $\al$:
\begin{definition} \label{def:cohomology index for SPT}
	Let $\al \in \scrA_{n^{-\infty}}$ be an equivariant quasi-local automorphism with superpolynomial tails on a spin chain. For any $G$-invariant product state, let $\psi = \psi_0 \circ \al$ and define 
	\begin{equation}
		\ind_G(\al) \in H^2(G, U(1))
	\end{equation}
	to be the unique cohomology class associated to $\al$ by Proposition \ref{prop:KSY2020 cohomology index}.
\end{definition}

We now show that $\ind_G(\al) = \ind_3(\al)$ coincides with the cohomology component of $\ind(\al)$.

\begin{proposition} \label{prop:extends cohomology index for QCA}
	Let $\al \in \scrA_{n^{-\infty}}$ be a quasi-local equivariant automorphism with superpolynomial tails on a spin chain $\caA$  and let $\psi_0$ be any $G$-invariant product state on $\caA$, then
	\begin{equation}
		\ind_G(\al) = \ind_3(\al).
	\end{equation}
\end{proposition}

\begin{proof}
	Since $\ind_G$ and $ind_3$ are both constant under strongly continuous deformation of $\al$ through $\scrA_{n^{-\infty}}$, and $\al$ can be so deformed to a QCA according to Proposition \ref{prop:quasi-local is G-equivalent to QCA}, it is sufficient to prove the statement if $\al$ is a QCA.

	Coarse grain the lattice so $\al$ is in nearest neighbour form.	Let $\psi_0$ be any $G$-invariant product state and let $\psi = \psi_0 \circ \al$. Let $\caB_n = \caA_{2n} \otimes \caA_{2n+1}$ and $\caC_{n} = \caA_{2n-1} \otimes \caA_{2n}$ and let $\caL_n$, $\caR_n$ be the overlap algebras defined by Eqs. (\ref{eq:left overlap}, \ref{eq:right overlap}) which, by Theorem \ref{thm:overlap factorization}, satisfy
	\begin{align}\begin{split}
		\caC_n &= \caL_n \otimes \caR_{n-1}, \\
		\caB_n &= \al^{-1}(\caL_n) \otimes \al^{-1}(\caR_n).
	\end{split}\end{align}
	For any $x \in \al^{-1}(\caC_n)$ we have $\psi(x) = \psi_0(x)$ so $\psi$ is a $G$-invariant product state with respect to the tensor decomposition
	\begin{equation}
		\caA = \cdots \otimes  \al^{-1}(\caC_{n-1}) \otimes \al^{-1}(\caC_n) \otimes \al^{-1}(\caC_{n+1}) \otimes \cdots
	\end{equation}
	So in the GNS representation $(\caH, \pi, \Omega)$ associated to $\psi$ we can write
	\begin{equation}
		|\Omega\rangle = \cdots \otimes | \phi_{n-1} \rangle \otimes | \phi_n \rangle \otimes | \phi_{n+1} \rangle \otimes \cdots
	\end{equation}

	Since
	\begin{equation}
		\caA_{\geq 2n} = \al^{-1}(\caL_n) \otimes \al^{-1}(\caC_{n+1}) \otimes \al^{-1}(\caC_{n+2}) \otimes \cdots
	\end{equation}
	we see that the action of $\rho_{\geq 2n}^g$ is given by conjugation with the (formal) unitary
	\begin{equation}
		\al^{-1}(v_n(g)) \otimes \al^{-1}\big(u_{2n+1}(g) \otimes u_{2n+2}(g)\big) \otimes \al^{-1}\big(u_{2n+3}(g) \otimes u_{2n+4}(g)\big) \otimes \cdots
	\end{equation}
	where $v_n$ is the projective representation of $G$ carried by $\caL_n$ and the $u_m$ are the on-site $G$-representations defining the action of the $\rho^g$.
	In the GNS representation of $\psi$ this becomes
	\begin{equation}
		V_{\geq 2n}(g) = \pi \left(   \al^{-1}(v_n(g)) \otimes \al^{-1}\big(u_{2n+1}(g) \otimes u_{2n+2}(g)\big) \otimes \cdots \right)
	\end{equation}
	This formal expression indeed defines a unitary $V(g)$ on $\caH$. The reason is that any state obtained from $| \Omega \rangle$ by acting with an operator supported on $[2n - m, 2n + m]$ is left invariant by it outside the region $[2n - m - 1, 2n + m + 1]$ (because all the states $\phi_m$ are $G$-invariant) \ie the action of this unitary does not take us out of the GNS Hilbert space.

	It is clear that the unitaries $V_{\geq 2n}(g)$ form a projective representation of $G$ as in Proposition \ref{prop:KSY2020 cohomology index} and that its cohomology class is the same as that of $v_n$. Since the former defines $\widetilde \ind_G(\al)$ and the latter defines $\ind_3(\al)$ we have $\ind_G(\al) = \ind_3(\al)$ as required.
\end{proof}

\appendix

\section{super von Neumann algebras} \label{app:super vN}

\subsection{Basic definitions}

We consider the algebra of bounded operators $\caB(\caH)$ on a Hilbert space $\caH$, equipped with a grade automorphism $\theta$. By Wigner's theorem there is a unitary $\Theta \in \caB(\caH)$ such that $\theta = \Ad_{\Theta}$. Moreover, since $\theta^2 = \id$ we can choose $\Theta$ such that $\Theta^2 = \I$. $\Theta$ and $-\Theta$ are the only two unitaries with this propery.

A subset $\caS \subset \caB(\caH)$ is called self-adjoint if $\caS^* := \{ s^* \, : \, s \in \caS \} = \caS$.
A subset $\caS \subset \caB(\caH)$ is a super subset if $\theta(\caS) = \caS$. A super subset that is also a sub vector space is called a super subspace.

An element $x \in \caB(\caH)$ is called homogeneous if $\theta(x) = (-1)^{\tau(x)} x$ with $\tau(x) \in \{0, 1\}$. If $\tau(x) = 0$ then $x$ is said to be even, and $x$ is odd if $\tau(x) = 1$. If $x$ and $y$ are both homogeneous then
\begin{equation}
	\theta(x y) = \theta(x) \theta(y) = (-1)^{\tau(x) + \tau(y)} xy
\end{equation}
so $xy$ is homogeneous of degree $\tau(xy) = \tau(x) + \tau(y)$ where the addition is modulo 2.

Similarly, a vector $\xi \in \caH$ is said to be homogeneous if $\Theta \xi = (-1)^{\tau(\xi)} \xi$. The `parity' $\tau(\xi)$ depends on the choice of $\pm\Theta$, but the prperty of homogeneity is independent of this choice.

Every element $x \in \caB(\caH)$ is the sum of it's even and odd parts, $x = x_+ + x_-$ with
\begin{equation}
	x_{\pm} := \frac{x \pm \theta(x)}{2}.
\end{equation}

The supercommutator is defined for homogeneous elements $x, y \in \caB(\caH)$ by
\begin{equation} \label{def:supercommutator}
	[x, y]_s = x y - (-)^{\tau(x) \tau(y)} y x
\end{equation}
and extended to all elements of $\caB(\caH)$ by linearity.

\begin{lemma} \label{lem:basic properties of the supercommutator}
	For any $x, y \in \caB(\caH)$ we have
	\begin{enumerate} [label=(\roman*)]
		\item \begin{equation}
			\big(  [x, y]_s \big)^* = [y^*, x^*]_s,
		\end{equation}
	\item	\begin{equation}
			\theta( [x, y]_s ) = [\theta(x), \theta(y)]_s.
		\end{equation}
	\item 	\begin{equation}
			\norm{[x, y]_s} \leq 2 (\norm{x_+} + \norm{x_-}) \norm{y} \leq 4 \norm{x} \norm{y}.
		\end{equation}
	\end{enumerate}
\end{lemma}

\begin{proof}
	(i) It is sufficient to show this for homogeneous $x$ and $y$, the general case then follows by linearity. For homogeneous $x$ and $y$ we have
	\begin{equation}
		([x, y]_s)^* = y^* x^* - (-1)^{\tau(x^*) \tau(y^*)} x^* y^* = [y^*, x^*]_s
	\end{equation}
	where we used that $x^*$ and $y^*$ are homogeneous of the same degree as $x$ and $y$ respectively.

	(ii) It is sufficient to show this for homogeneous $x$ and $y$, the general case then follows by linearity. For homogeneous $x$ and $y$ we have
	\begin{equation}
		\theta([x, y]_s) = \theta(x) \theta(y) - (-1)^{\tau(\theta(x)) \tau(\theta(y))} \theta(y) \theta(x) = [\theta(x), \theta (y)]_s
	\end{equation}
	where we used that $\theta(x)$ and $\theta(y)$ are homogeneous of the same degree as $x$ and $y$ respectively.

	(iii) Simply note that $[x, y]_s = [x_+, y] + \Theta [\Theta x_-, y]$, the bounds follow immediately.

\end{proof}

\begin{definition} \label{def:supercommutant}
	The supercommutant of $\caS \subset \caB(\caH)$, denoted $\caS^{\sharp}$, is the set of all elements of $\caB(\caH)$ that supercommute with all elements of $\caS$:
	\begin{equation} \label{eq:supercommutant}
		\caS^{\sharp} := \{ x \in \caB(\caH) \, : \, [x, s]_s = 0 \,\,\, \text{for all} \,\, s \in \caS \}.	
	\end{equation}
\end{definition}

\begin{remark}
	We could have defined $\caS^{\sharp} = \{ x \in \caB(\caH) \, : \, [s, x]_s = 0 \,\,\, \text{for all} \,\,\, s \in \caS \}$ instead. Since the supercommutant is not (anti-)symmetric, this could be a different set. However, point (ii) of the following Lemma shows that if $\caS$ is a super subspace, the two definitions actually coincide. We will use this fact without comment in the sequel.
\end{remark}

\begin{lemma} \label{lem:basic properties of the supercommutant}
	For any super subspace $\caS \subset \caB(\caH)$ such that $\theta(\caS) = \caS$ we have
	\begin{enumerate}[label=(\roman*)]
		\item $\caS^{\sharp}$ contains the identity.
		\item $\theta \big( \caS^{\sharp} \big) = \caS^{\sharp}$ and $\caS^{\sharp}$ is the linear span of its homogeneous elements. Moreover, $\caS^{\sharp} = \{ x \in \caB(\caH) \, : \, [s, x]_s = 0 \,\,\, \text{for all} \,\,\, s \in \caS \}$.
		\item $\caS^{\sharp}$ is closed in the WOT.
		\item $\caS^{\sharp}$ is closed under linear operations and multiplication.
		\item If $\caS$ is self-adjoint, then so is $\caS^{\sharp}$.	
	\end{enumerate}
\end{lemma}

\begin{proof}
	(i) is obvious.

	(ii) If $x \in \caS^{\sharp}$ then $[x, s]_s = 0$ for all $s \in \caS$. For all homogeneous $s \in \caS$, using that $\theta(s) \in \caS$ is homogeneous of the same degree as $s$, and using Lemma \ref{lem:basic properties of the supercommutator} (ii) we then get	
	\begin{equation}
		0 = \theta \big( [x, \theta(s)]_s \big) = [\theta(x), s]_s,
	\end{equation}
	so $\theta(x) \in \caS^{\sharp}$. Since $x \in \caS^{\sharp}$ was arbitrary, we find $\theta(\caS^{\sharp}) = \caS^{\sharp}$.

	It follows that for any $x \in \caS^{\sharp}$ also the even and odd parts $x_{\pm} \in \caS^{\sharp}$, so $\caS^{\sharp}$ is the linear span of its even and odd elements.

	Then $x \in \caS^{\sharp}$ if and only if $x_{\pm} \in \caS^{\sharp}$, \ie $[x_{\pm}, s]_s = 0$ for all $s \in \caS$. Since $\caS$ is a super subspace it is the linear span of its homogeneous elements, so $x \in \caS^{\sharp}$ if and only if $0 = [x_{\pm}, s]_s = -(-1)^{\tau(x_{\pm} \tau(s))} [s, x_{\pm}]$ for all homogeneous elements $s$ of $\caS$. It follows that $x \in \caS^{\sharp}$ if and only if $[s, x]_s = 0$ for all $s \in \caS$. 

	(iii) Suppose $\{x_i\}_{i \in I}$ is a net in $\caS^{\sharp}$ converging in the WOT to $x \in \caB(\caH)$. \ie for all $\zeta, \xi \in \caH$ we have $\langle \zeta, x_i \xi \rangle \rightarrow \langle \zeta, x \xi \rangle$. Then also, for all $\zeta, \xi \in \caH$, $\langle \zeta, (x_i \pm \theta(x_i) ) \zeta \rangle = \langle \zeta, x \xi \rangle \pm \langle \Theta^* \zeta, x_i \Theta \xi \rangle \rightarrow \langle \zeta, x \xi \rangle \pm \langle \Theta^* \zeta, x \Theta \xi \rangle = \langle \zeta, (x \pm \theta(x)) \xi \rangle$ where $\Theta \in \caU(\caH)$ is such that $\theta = \Ad_{\Theta}$, hence the nets $\{(x_i)_\pm\}_{i \in I}$ converge in the WOT to $x_{\pm}$. By (ii), the nets $\{ (x_i)_{\pm} \}$ are in $\caS^{\sharp}$. We then find for any $\zeta, \xi \in \caH$ and any homogeneous $s \in \caS$ that
	\begin{align*}
		\langle \zeta, [s, x_{\pm}]_s \xi \rangle &= \langle \zeta, (s x_{\pm} - (-1)^{\tau(s) \tau(x_{\pm})} x_{\pm} s   ) \xi \rangle \\
							  &= \langle s^* \zeta, x_{\pm} \xi \rangle - (-1)^{\tau(s) \tau(x_{\pm})} \langle \zeta, x_{\pm} \, s \xi \rangle \\
							  &= \lim_{i \rightarrow \infty} \langle s^* \zeta, (x_i)_{\pm} \xi \rangle - (-1)^{\tau(s) \tau((x_i)_{\pm})} \langle \zeta, (x_i)_{\pm} \, s \xi \rangle \\
							  &= \lim_{i \rightarrow \infty} \langle \zeta, [s, (x_i)_{\pm}]_s \xi \rangle = 0 
	\end{align*}
	where we used $\tau((x_i)_{\pm}) = \tau(x_{\pm})$ for all $i \in I$. Since $\theta(\caS) = \caS$ we have for any $s \in \caS$ that also $s_{\pm} \in \caS$ ,and by linearity of the supercommutator we find
	\begin{equation}
		\langle \zeta, [s, x]_s \xi \rangle = \langle \zeta, [s_+, x_+]_s \xi \rangle + \langle \zeta, [s_+, x_-]_s \xi \rangle + \langle \zeta, [s_-, x_+]_s \xi \rangle + \langle \zeta, [s_-, x_-]_s \xi \rangle = 0
	\end{equation}
	for all $\zeta, \xi \in \caH$. We conclude that $x \in \caS^{\sharp}$, \ie $\caS^{\sharp}$ is closed in the weak operator topology.

	(iv) If $x, y \in \caS^{\sharp}$ and $\lambda \in \C$ then by linearity of the supercommutator we have for any $s \in \caS$ that
	\begin{equation}
		[s, x + \lambda y]_s = [s, x]_s + \lambda [s, y]_s = 0,
	\end{equation}
	so $x + \lambda y \in \caS^{\sharp}$, \ie $\caS^{\sharp}$ is closed under linear operations.

	If $s \in \caS$ and $x, y \in \caS^{\sharp}$ are homogeneous, then
	\begin{align*}
		[s, xy]_s &= s x y - (-1)^{\tau(s) \tau(xy)} xy s \\
			  &=  \big( s x y - (-1)^{\tau(s) \tau(x)} x s y \big) - \big( (-1)^{\tau(s) \tau(xy)} xy s - (-1)^{\tau(s)\tau(xy)} (-1)^{\tau(s)\tau(y)} x s y \big) \\
			  &= [s, x]_s y - (-1)^{\tau(s) \tau(xy)} x [y, s]_s = 0
	\end{align*}
	where we used $\tau(s)\tau(xy) + \tau(s) \tau(y) = \tau(s) \tau(x)$ to get the second line. For arbitrary $s \in \caS$ and $x, y \in \caS^{\sharp}$ we then get $[s, xy]_s = 0$ by decomposing $s, x, y$ in their homogeneous parts. We conclude that $xy \in \caS^{\sharp}$, \ie $\caS^{\sharp}$ is closed under multiplication.

	(v) Take $x \in \caS^{\sharp}$. If $\caS$ is self-adjoint then $[s^*, x]_s = 0$ for all $s \in \caS$ so by point (i) of Lemma \ref{lem:basic properties of the supercommutator}, we find $0 = ([s^*, x]_s)^* = [x^*, s]_s$ for all $s \in \caS$, \ie $x^* \in \caS^{\sharp}$. Since $x$ was arbitrary, we conclude that $\big( \caS^{\sharp} \big)^* = \caS^{\sharp}$.
\end{proof}

\subsection{Double supercommutant theorem}

We prove a super-version of von Neumann's double commutant theorem:

\begin{theorem} \label{thm:double supercommutant theorem}
	Let $\caA \subset \caB(\caH)$ be super *-algebra containing the identity. Then $\caA$ is dense in $\caA^{\sharp \sharp}$ in the strong operator topology. 
\end{theorem}

An immediate consequence is
\begin{corollary} \label{cor:double super commutant is vN}
	Let $\caA \subset \caB(\caH)$ be super *-algebra containing the identity. Then $\caA^{\sharp \sharp} = \caA''$ is the von Neumann algebra generated by $\caA$.
\end{corollary}

\begin{proof}
	From Theorem \ref{thm:double supercommutant theorem} $\caA$ is strongly dense, hence weakly dense in $\caA^{\sharp \sharp}$. By Lemma \ref{lem:basic properties of the supercommutant} (iii) $\caA^{\sharp \sharp}$ is weakly closed. It follows that $\caA^{\sharp \sharp}$ is the weak closure of $\caA$. By von Neumann's bicommutant theorem, this is $\caA''$, the von Neumann algebra generated by $\caA$.
\end{proof}

To prove Theorem \ref{thm:double supercommutant theorem} we follow closely the proof of von Neumann's double commutant theorem as presented in \cite{Arveson2012}.

\begin{proofof}[Theorem \ref{thm:double supercommutant theorem}]
	Let $\caA_s$ and $\caA_w$ be the strong and weak closure of $\caA$ respectively. Then $\caA_s \subset \caA_w \subset \caA^{\sharp \sharp}$ because clearly $\caA \subset \caA^{\sharp \sharp}$ and the latter is weakly closed by Lemma \ref{lem:basic properties of the supercommutant} (iii).

	It therefore suffices to show that every $x \in \caA^{\sharp \sharp}$ can be strongly approximated by operators in $\caA$. By Lemma \ref{lem:homogeneous base for SOT} it is sufficient to show that for every $\ep > 0$, and every finite set of homogeneous vectors $\xi_1, \cdots, \xi_n \in \caH$ there is an operator $a \in \caA$ also in
	\begin{equation}
		\caU(x ; \xi_1, \cdots \xi_n ; \ep) = \{ y \in \caB(\caH) \, : \, \sum_{i=1}^n \norm{ (x-y) \xi_i}^2 < \ep^2 \}.
	\end{equation}

	We first deal with the case $n = 1$. Fix a grading operator $\Theta$ and let $\Theta \xi = (-1)^{\tau_i} \xi_i$ for $\tau_i \in \{0, 1\}$. Let $p$ be the orthogonal projection on $[\caA \xi_1]$, the (Hilbert space) norm closure of $\caA \xi_1 = \{ a \xi_1 \, : \, a \in \caA \}$. Since $[\caA \xi_1]$ and its orthogonal complement are both invariant under $\caA$ we have $p \in \caA'$. Moreover, $p$ is even. Indeed, since $\xi_1$ is homogeneous and $\caA = \theta(\caA)$, we find that $[\caA \xi_1]$ is invariant under $\Theta$, and so is $[\caA \xi_1]^{\perp}$. Decomposing an aribtary $\xi = \xi_{\parallel} + \xi_{\perp}$ with $\xi_{\parallel} \in [\caA \xi_1]$ and $\xi_{\perp} \in [\caA \xi_1]^{\perp}$ we have
	\begin{equation}
		\theta(p) \xi = \Theta p (\Theta \xi_{\parallel} + \Theta \xi_{\perp} ) = \Theta^2 \xi_{\parallel} = p \xi.
	\end{equation}
	Since this holds for any $\xi \in \caH$, we have $\theta(p) = p$, \ie $p$ is even.

	Since $p$ is an even element of $\caA'$ also $p \in \caA^{\sharp}$.

	Since $\I \in \caA$ we have $\xi_1 \in [\caA \xi_1]$ so $p^{\perp} \xi_1 = 0$. The operator $x \in \caA^{\sharp \sharp}$ commutes with the even $p \in \caA^{\sharp}$ so $x$ leaves the range of $p$ invariant. In particular $x \xi_1 \in [\caA \xi_1]$ so for any $\ep > 0$ there is an $a \in \caA$ such that $\norm{(x - a) \xi_1} < \ep$, \ie $a \in \caU(x ; \xi_1 ; \ep)$, as required.

	If $n > 1$, we have homogeneous vectors $\xi_1, \cdots, \xi_n$. Order them such that $\Theta \xi_i = \xi_i$ for $i = 1, \cdots, p$ and $\Theta \xi_i = - \xi_i$ for $i = p+1, \cdots p+q$ for inergers $p+q=n$. Let $\widetilde \caH = \caH^{\oplus n}$ and equip $\caB(\widetilde \caH)$ with a $\Z_2$ grading $\tilde \theta \ \Ad_{\widetilde \Theta}$ where $\widetilde \Theta = \Theta^{\oplus p} \oplus (- \Theta)^{\oplus q}$. Then the vector $\eta = \xi_1 \oplus \cdots \oplus \xi_n \in \widetilde \caH$ is homogeneous w.r.t $\widetilde \Theta$, in fact $\widetilde \Theta \eta = \eta$.

	We will often decompose elements of $\caB(\widetilde \caH)$ as $2 \times 2$ block matrices
	\begin{equation}
		\tilde x = \begin{bmatrix} x_{++} & x_{+-} \\ x_{-+} & x_{--} \end{bmatrix}
	\end{equation}
	with $x_{++} ,x_{--}, x_{+-}$ and $x_{-+}$ matrices of dimension $p \times p$, $q \times q$, $p \times q$ and $q \times p$ respectivley, with entries from $\caB(\caH)$.

	An element $\tilde x \in \caB(\widetilde \caH)$ is even ($\tilde \theta(\tilde x) = \tilde x$) if and only if the entries of $x_{++}$ and $x_{--}$ are even elements of $\caB(\caH)$ and the entries of $x_{+-}$ and $x_{-+}$ are odd elements of $\caB(\caH)$. Similarly, an element $\tilde x \in \caB(\widetilde \caH)$ is odd ($\tilde \theta(\tilde x) = -\tilde x$) if and only if the entries of $x_{++}$ and $x_{--}$ are odd elements of $\caB(\caH)$ and the entries of $x_{+-}$ and $x_{-+}$ are even elements of $\caB(\caH)$.

	Consider the $\tilde \theta$-invariant *-algebra $\widetilde \caA := \caA^{\oplus n} \subset \caB(\widetilde \caH)$. \ie $\widetilde \caA$ consists of elements
	\begin{equation}
		\tilde a = \begin{bmatrix} a & &  & 0 \\ & a & & \\ & & \ddots & \\ 0 & & & a \end{bmatrix}
	\end{equation}
	for $a \in \caA$.
	
	For a matrix $y$ with entries in $\caB(\caH)$, let $\Theta y$ denote the matrix obtained by multiplying each of the entries of $y$ with the grading operator $\Theta$.	A straightforward calculation shows that $\widetilde \caA^{\sharp}$ is spanned by homogeneous elements of the form
	\begin{equation}
		\tilde y = \begin{bmatrix} y_{++} & \Theta y_{+-} \\ \Theta y_{-+} & y_{--} \end{bmatrix}
	\end{equation}
	where all the entries of $y_{++}, y_{--}, y_{+-}$ and $y_{+-}$ are elements of $\caA^{\sharp}$.

	A similar calculation \note{this might need elaboration...} shows that $\widetilde \caA^{\sharp \sharp}$ consists of elements of the form
	\begin{equation}
		\tilde x = \begin{bmatrix} x^{\oplus p} & 0 \\ 0 & x^{\oplus q} \end{bmatrix} = x^{\oplus n} 
	\end{equation}
	where $x \in \caA^{\sharp \sharp}$. \ie $\widetilde \caA^{\sharp \sharp} = (\caA^{\sharp \sharp})^{\oplus n}$.

	For $x \in \caA^{\sharp \sharp}$ we therefore have $\tilde x = x^{\oplus n} \in (\widetilde \caA)^{\sharp \sharp}$. We can then apply the case $n=1$ to find for every $\ep > 0$ an element $\tilde a \in \widetilde \caA$ such that $\tilde a \in \caU(\tilde x ; \eta ; \ep)$. Since $\tilde a = a^{\oplus n}$ for some $a \in \caA$, this yields $a \in \caU(x ; \xi_1, \cdots, \xi_n ; \ep)$ as required.
\end{proofof}

\begin{lemma} \label{lem:homogeneous base for SOT}
	The sets
	\begin{equation}
		\caU(x ; \xi_1, \cdots, \xi_n ; \ep) := \{ y \in \caB(\caH) \, : \, \sum_{i=1}^n \norm{ (x-y) \xi_i}^2 < \ep^2 \} \subset \caB(\caH)
	\end{equation}
	for $x \in \caB(\caH)$, $\ep > 0$ and $\xi_1, \cdots, \xi_n \in \caH$ homogenous form a base for the strong operator topology on $\caB(\caH)$.
\end{lemma}

\begin{proof}
	Let $\caU_s \subset \caB(\caH)$ be a strongly open set and take $x \in \caU_s$. We want to find a finite set of homogeneous vectors $\{\eta_i\}_{i \in I}$ and $\ep  0$ such that $x \in \caU(x ; \{\eta_i\}_{i \in I} ; \ep) \subset \caU_s$.

	The sets $\caU(x ; \xi_1, \cdots, \xi_n ; \ep)$ where the $\xi_i$ are not neccessarily homogeneous form a base of the strong operator topology. So there are vector $\xi_1, \cdots \xi_n$ and $\ep > 0$ such that $x \in \caU(x : \xi_1, \cdots, \xi_n ; \ep) \subset \caU_s$.

	Let $\eta_{i, \pm} = \frac{1}{2}(\xi_i \pm \Theta \xi_i)$ so $\xi_i = \eta_{i, +} + \eta_{i, -}$. The vectors $\eta_{i, \pm}$ are homogeneous, and have parity $\pm$ w.r.t. the grading unitary $\Theta$.

	We have $y \in \caU( x ; \{ \eta_{i, \pm} \}, ; \delta )$ if and only if
	\begin{equation}
		\sum_i \left(  \norm{(x-y) \eta_{i, +}}^2 + \norm{ (x-y) \eta_{i, -} }^2  \right) \leq \delta^2
	\end{equation}
	in which case $\norm{(x-y) \eta_{i, +}} \leq \delta$ for all $i = 1, \cdots, n$. It follows that
	\begin{equation}
		\norm{(x-y)\xi_i} = \norm{ (x-y) (\eta_{i, +} + \eta_{i, -}) } \leq 2 \delta
	\end{equation}
	for all $i$, so taking $\delta = \ep / (2 \sqrt{n})$ we get
	\begin{equation}
		\sum_i \norm{(x-y) \xi_i}^2 < n (2\delta)^2 = \ep^2
	\end{equation}
	so $y \in \caU(x ; \{\xi_i \} ; \ep)$. We conclude that
	\begin{equation}
		x \in \caU( x ; \{ \eta_{i, \pm} \}, ; \delta ) \subset \caU(x ; \{\xi_i \} ; \ep) \subset \caU_s
	\end{equation}
	where $\{\eta_{i \pm}\}$ is a finite collection of homogeneous vectors. This concludes the proof.
\end{proof}

\subsection{Conditional Expectations}

We recall first the definition:
\begin{definition} \label{def:conditional expectation}
	A conditional expectation on a von Neumann algebra $\caM \subset \caB(\caH)$ is a linear map $\E : \caB(\caH) \rightarrow \caM$ such that
	\begin{enumerate}[label=(\roman*)]
		\item $\E$ is completely positive.
		\item for all $x \in \caB(\caH)$ we have $\norm{\E(x)} \leq \norm{x}$, \ie $\E$ is a contraction.
		\item for all $a, b \in \caM$ and all $x \in \caB(\caH)$ we have $\E(axb) = a \E(x) b$.
	\end{enumerate}
\end{definition}

We will construct conditional expectations on the supecommutants of certain hyperfinite super von Neumann algebras. We begin by stating a useful result from \cite{NSW2013} (Proposition 4.1):
\begin{theorem} \label{thm:good CE on commutant}
	If $\caA \subset \caB(\caH)$ is a hyperfinite von Neuman algebra, then there is a conditional expectation $\E_{\caA'} : \caB(\caH) \rightarrow \caA'$ such that if $\norm{[x, a]} \leq \ep \norm{x} \norm{a}$ for all $a \in \caA$, then $\norm{\E_{\caA'}(x) - x} \leq \ep \norm{x}$. 
\end{theorem}

We want to construct conditional expectations on $\caA^{\sharp}$. We can use the above theorem to obtain such a conditional expectation if we know that $\widetilde \caA := (\caA^{\sharp})'$ is hyperfinite. For the following class of super von Neumann algebras, we can describe $\widetilde \caA$ explicitly:

\begin{definition} \label{def:inner-super}
	A von Neumann algebra $\caA \subset \caB(\caH)$ is said to be inner-super if there is a unitary $\Theta_{\caA} \in \caA$ with $\Theta_{\caA}^2 = \I$ implementing the grading, $\theta(a) = \Ad_{\Theta_{\caA}}(a)$ for all $a \in \caA$. \ie $\theta|_{\caA}$ is inner.
\end{definition}

\begin{lemma} \label{lem:tilde A of inner-super A}
	If $\caA \subset \caB(\caH)$ is inner-super with grading implemented by $\Theta_{\caA} \in \caA$ then
	\begin{equation}
		\widetilde \caA := (\caA^{\sharp})' = \{ a_+ + a_- \Theta_{\caA^{\sharp}} \, : \, a \in \caA  \}.
	\end{equation}
	where $\Theta_{\caA^{\sharp}} = \Theta_{\caA} \Theta$.

	Moreover, $\widetilde \caA$ is isomorphic to $\caA$ as a super algebra.
\end{lemma}

\begin{proof}
	$\Theta_{\caA^{\sharp}} = \Theta_{\caA} \Theta$ is an even element of $\caA^{\sharp}$ such that $\theta(b) = \Ad_{\Theta_{\caA^{\sharp}}}(b)$ for all $b \in \caA^{\sharp}$. Indeed, to see that $\Theta_{\caA_{\sharp}} \in \caA^{\sharp}$, note that for $x \in \caA$ we have $\Theta_{\caA^{\sharp}}^* x \Theta_{\caA^{\sharp}} = \Theta^* \Theta_{\caA}^* x \Theta_{\caA} \Theta = \theta( \theta(x)) = x$, so $\Theta_{\caA^{\sharp}} \in \caA'$. since $\Theta_{\caA^{\sharp}}$ is even, we also have $\Theta_{\caA^{\sharp}} \in \caA^{\sharp}$. To see that $\Theta_{\caA^{\sharp}}$ implements $\theta$ on $\caA^{\sharp}$, take $b \in \caA^{\sharp}$ and compute $\theta(b) = \Theta^* b \Theta = \Theta_{\caA^{\sharp}}^* \Theta_{\caA^*} b \Theta_{\caA} \Theta_{\caA^{\sharp}} = \Ad_{\Theta_{\caA^{\sharp}}}(b)$, where we used that $\Theta_{\caA} \in \caA$ is even and therefore commutes with all elements of $\caA^{\sharp}$.

	Denote $\caC = \{ a_+ + a_- \Theta_{\caA^{\sharp}} \, : \, a \in \caA  \}$. We show that $\widetilde \caA = \caC$ by showing that $\caC' = \caA^{\sharp}$. 

A homogeneous element $x$ is in the commutant of $\caC$ if and only if $[x, a_+] = [x, a_+]_s = 0$ and $[x, a_- \Theta_{\caA^{\sharp}}] = 0$ for all $a \in \caA$. The first condition in particular entails that $[x, \Theta_{\caA}] = 0$, because $\Theta_{\caA} \in \caA^+$. The second condition can then be written
	\begin{align*}
		0 &= x a_- \Theta_{\caA} \Theta - a_- \Theta_{\caA} \Theta x \\
		  &= \big( x a_- \Theta_{\caA} - (-1)^{\tau(x)}  a_- \Theta_{\caA} x \big) \Theta \\
		  &= \big(   x a_- \Theta_{\caA} - (-1)^{\tau(x)} a_- x \Theta_{\caA} + (-1)^{\tau(x)} a_- x \Theta_{\caA} - (-1)^{\tau(x)}  a_- \Theta_{\caA} x  \big) \Theta \\
		  &= \big( [x, a_-]_s \Theta_{\caA} + (-1)^{\tau(x)} a_- [x, \Theta_{\caA}] \big) \Theta \\
		  &= [x, a_-]_s \Theta_{\caA^{\sharp}}
	\end{align*}
	where in obtaining the last line we used $[x, \Theta_{\caA}] = 0$. We find that $x \in \caC'$ if and only if $[x, a_+]_s = [x, a_-]_s$ for all $a \in \caA$. By Lemma \ref{lem:basic properties of the supercommutant} (ii) we conclude that $x \in \caC'$ if and only if $x \in \caA^{\sharp}$. Since both algebras are spanned by their homogeneous elements, we conclude that they are equal.

Finally, $\widetilde \caA$ is the image of $\Phi_{\caA} : \caA \rightarrow \caB(\caH) : a \mapsto a_+ + a_- \Theta_{\caA^{\sharp}}$. Using that $\Theta_{\caA_{\sharp}}$ commutes with all elements of $\caA$, is self-adjoint and squares to $\I$, one easily checks that $\Phi_{\caA}$ is a unital and injective super *-homomorphism, so $\widetilde \caA$ is a superalgebra isomorphic to $\caA$.
\end{proof}

We have further,
\begin{lemma} \label{lem:supercommute with A iff commute with tilde A}
Let $\caA \subset \caB(\caH)$ be an inner-super factor and suppose $x \in \caB(\caH)$ satisfies $\norm{[x, a]_s} \leq \ep \norm{x} \norm{a}$ for all $a \in \caA$. Then
	\begin{equation}
		\norm{[x, \tilde a]} \leq 3 \ep \norm{x} \norm{\tilde a}
	\end{equation}
	for any $\tilde a \in \widetilde \caA$.

	Conversely, if $x \in \caB(\caH)$ satisfies $\norm{[x, \tilde a]} \leq \ep \norm{x} \norm{\tilde a}$ for all $\tilde a \in \widetilde \caA$ then
	\begin{equation}
		\norm{[x, a]_s} \leq 3 \ep \norm{x} \norm{a}
	\end{equation}
	for all $a \in \caA$.
\end{lemma}

\begin{proof}
	For the first claim, take $\tilde a \in \widetilde \caA$. By Lemma \ref{lem:tilde A of inner-super A} it is of the form $\tilde a = a_+ + a_-\Theta_{\caA^{\sharp}}$ with $a_+, a_- \in \caA$ even and odd respectivley. Then
	\begin{align*}
	[x, \tilde a] &= [x, \Theta_{\caA^{\sharp}} a_-] + [x, a_+]_s = x \Theta \Theta_{\caA} a_- - \Theta \Theta_{\caA} a_- x + [x, a_+] \\
			      &= \Theta_{\caA} \left( a_- \theta(x) - x a_- \right) \Theta + [\Theta_{\caA}, x]_s a_- + [x, a_+] \\
			      &= -\Theta_{\caA} [x, a_-]_s \Theta + [\Theta_{\caA}, x]_s a_- + [x, a_+].
	\end{align*}
	Each of the three terms in the last line is bounded by $\ep \norm{x} \norm{\tilde a}$ (using $\norm{a_{\pm}} \leq \norm{a} = \norm{\tilde a}$) so
	\begin{equation}
		\norm{ [x, \tilde a] } \leq 3 \ep \norm{x} \norm{\tilde a}
	\end{equation}
	as required.

	For the second claim, take $a \in \caA$ and compute
	\begin{align*}
		[x, a]_s &= [x, a_+] + [x_+, a_- \Theta_{\caA^{\sharp}}^2] + [x_-, a_- \Theta_{\caA^{\sharp}}^2]_s \\
			 &= [x, a_+] + [x, a_- \Theta_{\caA^{\sharp}}] \Theta_{\caA^{\sharp}} - a_- \Theta_{\caA^{\sharp}} [\Theta_{\caA}, x] \Theta
	\end{align*}
	where we used
	\begin{equation}
		\Theta_{\caA^{\sharp}} x_{\pm} = \Theta_{\caA} \Theta x_{\pm} = \pm \Theta_{\caA} x_{\pm} \Theta = \pm x_{\pm} \Theta_{\caA^{\sharp}} \pm [\Theta_{\caA}, x_{\pm}] \Theta.
	\end{equation}
	Since $a_+$, $a_- \Theta_{\caA^{\sharp}}$ and $\Theta_{\caA}$ are elements of $\widetilde \caA$ we have
	\begin{equation}
		\norm{[x, a]_s} \leq 3 \ep \norm{x} \norm{a}
	\end{equation}
	as required.
\end{proof}

Note that super algebras isomorphic to full matrix algebras and containing $\I$ are inner-super factors.

\begin{proposition} \label{prop:conditional expectation on supercommutant of a matrix algebra}
	Let $\caA$ be super subalgebra of $\caB(\caH)$ isomorphic to a full matrix algebra and containing $\I$. Then
	\begin{equation} \label{eq:condtional expectation on supercommutant of a matrix algebra}
		\E_{\caA^{\sharp}}(x) := \int_{\caU(\caA)} \dd u \, \Phi_{\caA}(u)^* x \Phi_{\caA}(u) = \int_{\caU(\widetilde \caA)} \dd u \, u^* x u,
	\end{equation}
	where the integrations are w.r.t. the Haar measure on the unitary groups, is a conditional expectation from $\caB(\caH)$ to $\caA^{\sharp}$.

	Moreover, if $x \in \caB(\caH)$ is homogeneous, then $\E_{\caA^{\sharp}}(x)$ is homogeneous of the same degree as $x$.
\end{proposition}


\begin{proof}
	The map
	\begin{equation}
		x \mapsto \int_{\caU(\widetilde \caA)} \dd u \, u^* x u
	\end{equation}
	is a conditional expectation from $\caB(\caH)$ to $\widetilde \caA'$ (see for example II.6.10.4 (iv) of \cite{Blackadar2006}). From Lemma \ref{lem:tilde A of inner-super A} we have $\widetilde \caA' = \caA^{\sharp}$, so we find that $\E_{\caA^{\sharp}}$ is indeed a conditional expectation on $\caA^{\sharp}$.

	If $x$ is homogeneous then, using the $\theta$-invariance of $\caU(\widetilde \caA)$, we find
	\begin{equation}
		\theta \left( \E_{\caA^{\sharp}}(x) \right) = \int_{\caU(\widetilde \caA)} \dd u \, \, \theta(u)^* \, \theta(x) \, \theta(u) = (-1)^{\tau(x)} \int_{\caU(\widetilde \caA)} \dd u \, \, u^* x u = (-1)^{\tau(x)} \E_{\caA^{\sharp}}(x),
	\end{equation}
	as required.
\end{proof}

\begin{lemma} \label{lem:conditional expectation preserves near supercommutation}
	Let $\caA$ be super subalgebra of $\caB(\caH)$, containing $\I$, and isomorphic to a full matrix algebra. Suppose $x \in \caB(\caH)$ is such that $\norm{[x, a]_s} \leq \ep \norm{x} \norm{a}$ for all $a \in \caA$. Then
	\begin{equation}
		\norm{ \E_{\caA^{\sharp}}(x) - x} \leq 3 \ep \norm{x}.
	\end{equation}
\end{lemma}

\begin{proof}
	We have
	\begin{equation}
		\E_{\caA^{\sharp}}(x) - x = \int_{\caU(\widetilde \caA)} \dd u \, u^* [x, u]
	\end{equation}
	so the bound follows immediately from Lemma \ref{lem:supercommute with A iff commute with tilde A}.
\end{proof}

\begin{lemma} \label{lem:double conditional expectation supercommutes with both}
	Let $\caA$ and $\caB$ be super subalgebras of $\caB(\caH)$, containing $\I$ and isomorphic to full matrix algebras, and such that $[\caA, \caB]_s = 0$. Then
	\begin{equation}
		\E_{\caA^{\sharp}} \left(  \E_{\caB^{\sharp}}(x) \right) \in \Span \left( \caA \cup \caB \right)^{\sharp}
	\end{equation}
	for any $x \in \caB(\caH)$.
\end{lemma}

\begin{proof}
	Fix $x \in \caB(\caH)$ and let $y = \E_{\caA^{\sharp}} \left(  \E_{\caB^{\sharp}}(x) \right)$. It follows immediately from Proposition \ref{prop:conditional expectation on supercommutant of a matrix algebra} that $y \in \caA^{\sharp}$, we now show that also $y \in \caB^{\sharp}$. Let $x = x_+ + x_-$ be the decomposition of $x$ in its even and odd parts, then again by Proposition \ref{prop:conditional expectation on supercommutant of a matrix algebra},
	\begin{equation}
		y_{\pm} = \E_{\caA^{\sharp}} \left(  \E_{\caB^{\sharp}}(x_{\pm}) \right) 
	\end{equation}
	are the even and odd parts of $y$.

	Take $b \in \caB$ homogeneous. To show that $y$ supercommutes with $b$ we must show that each of $y_{\pm}$ supercommutes with $b$. Because $[\caA, \caB]_s = 0$ we have $b \in \caA^{\sharp} = \widetilde \caA'$ by Lemma \ref{lem:tilde A of inner-super A} so
	\begin{align*}
		b y_{\pm} &= \int_{\caU(\widetilde \caA)} \dd u \,\, y u^* \E_{\caB^{\sharp}}(x_{\pm}) u = (-1)^{\tau(y) \tau(x_{\pm})} \int_{\caU(\widetilde \caA)} \dd u \,\, u^* \E_{\caB^{\sharp}}(x_{\pm}) u \, b \\
			  &= (-1)^{\tau(y) \tau(x_{\pm})} y_{\pm} b = (-1)^{\tau(b) \tau(y_{\pm})} y_{\pm} b,
	\end{align*}
	where we also used that $\E_{\caB^{\sharp}}(x_{\pm}) \in \caB^{\sharp}$ are homogeneous of the same degree as $x_{\pm}$. We conclude that $y$ supercommutes with all homogeneous elements of $\caB$, \ie $y \in \caA^{\sharp} \cap \caB^{\sharp} = \Span(\caA \cup \caB)^{\sharp}$.
\end{proof}

\begin{definition}
	We say a super von Neumann algebra $\caA \subset \caB(\caH)$ is $\theta$-hyperfinite if it is the weak closure of an increasing family of super algebras that are isomorphic to full matrix algebras with common identity $\I$.
\end{definition}

We follow \cite{NSW2013} to obtain:

\begin{proposition} \label{prop:conditional expectation on supercommutant of theta-hyperfinite vN algebra}
	For any $\theta$-hyperfinite von Neumann algebra $\caA \subset \caB(\caH)$ there is a conditional expectation $\E_{\caA^{\sharp}} : \caB(\caH) \rightarrow \caA^{\sharp}$ such that if $x \in \caB(\caH)$ is such that $\norm{[x, a]_s} \leq \ep \norm{x} \norm{a}$ for all $a \in \caA$ then
	\begin{equation}
		\norm{\E_{\caA^{\sharp}}(x) - x} \leq 3 \ep \norm{x}.
	\end{equation}
	Moreover, $\E_{\caA^{\sharp}}(x) \in (\{x\} \cup \widetilde \caA)''$.
\end{proposition}

\begin{proof}
	Let $\{\caA_{\al}\}$ be an increasing family of super algebras, each isomorphic to a full matrix algebra, whose weak closure is $\caA$. For $x \in \caB(\caH)$ let $y_{\al} = \E_{\caA_{\al}^{\sharp}}$ with $\E_{\caA_{\al}}^{\sharp}$ the conditional expectation of Proposition \ref{prop:conditional expectation on supercommutant of a matrix algebra}. Since $\norm{y_{\al}} = \norm{ \E_{\caA_{\al}^{\sharp}}(x) } \leq \norm{x}$ is bounded independent of $\al$, the family $\{y_{\al}\}$ has a weak accumulation point $y$ with $\norm{y} \leq \norm{x}$. Set $\E_{\caA^{\sharp}}(x) = y$.

	To see that $y \in \caA^{\sharp}$, note that by Proposition \ref{prop:conditional expectation on supercommutant of a matrix algebra} we have $y_{\al} \in \caA_{\al}^{\sharp}$ for each $\al$. Since $\{ \caA_{\al} \}$ is an increasing family, the family of supercommutants $\{ \caA_{\al}^{\sharp} \}$ is decreasing. Moreover, by Lemma \ref{lem:basic properties of the supercommutant} (iii), each $\caA_{\al}^{\sharp}$ is weakly closed. It follows that $y \in \caA_{\al}^{\sharp}$ for each $\al$, hence
	\begin{equation}
		y \in \bigcap_{\al} \caA_{\al}^{\sharp} = \Span \left( \cup_{\al} \caA_{\al} \right)^{\sharp} = \caA^{\sharp}
	\end{equation}
	where we used that the supercommutant of a superspace is the supercommutant of the weak closure of that space.

	If $\norm{[x, a]_s} \leq \ep \norm{x} \norm{a}$ for all $a \in \caA$ then it follows from Lemma \ref{lem:conditional expectation preserves near supercommutation} that $\norm{\E_{\caA_{\al}^{\sharp}}(x) - x} \leq 3 \ep \norm{x}$ for all $\al$. Since $\E_{\caA^{\sharp}}(x) - x$ is a weak accumulation point of $\{ \E_{\caA_{\al}^{\sharp}}(x) - x \}$ it follows that also $\norm{\E_{\caA^{\sharp}}(x) - x} \leq 3 \ep \norm{x}$.

	In particular, if $x \in \caA^{\sharp}$ then $\E_{\caA^{\sharp}}(x) = x$ so $\E_{\caA^{\sharp}} : \caB(\caH) \rightarrow \caA^{\sharp}$ is a projection on the von Neumann algebra $\caA^{\sharp}$ which contains the identity. It follows that $\E_{\caA^{\sharp}}$ is completely positive, is of norm 1, and satisfies $\E_{\caA^{\sharp}}(a x b) = a \E_{\caA^{\sharp}}(x) b$ for $a, b \in \caA^{\sharp}$ and $x \in \caB(\caH)$ \cite{Tomiyama1959, EvansLewis1977} \margin{read the references}, \ie $\E_{\caA^{\sharp}}$ is a conditional expectation.

	The final claim follows by noting that $\E_{\caA^{\sharp}}(x)$ is a weak accumulation point of the
	\begin{equation}
		\E_{\caA_{\al}^{\sharp}}(x) = \int_{\caU(\widetilde \caA_{\al})} \, \dd u \, u^* x u  \in (\{x\} \cup \widetilde \caA)''.
	\end{equation}
\end{proof}

\begin{lemma} \label{lem:tilde A of theta-hyperfinite A}
	If $\caA \subset \caB(\caH)$ is $\theta$-hyperfinite, \ie $\caA = \overline{\cap_{\al} \caA_{\al}}^w$ for some increasing family $\{\caA_{\al}\}$ of super algebras containing $\I$ and isomorphic to full matrix algebras, then
	\begin{equation}
		\widetilde \caA = \overline{ \bigcap_{\al}  \widetilde \caA_{\al} }^w.
	\end{equation}
	In particular, $\widetilde \caA$ is $\theta$-hyperfinite.
	
	Moreover, if $\caB$ is any super subalgebra of $\caA$ containing $\I$, then $\widetilde \caB \subset \widetilde \caA$.
\end{lemma}

\begin{proof}
	By definition, $\widetilde \caA$ is the commutant of $\caA^{\sharp}$. If $\caB$ is a super subalgebra of $\caA$ containing $\I$ then $\widetilde \caB$ is the commutant of $\caB^{\sharp} \supset \caA^{\sharp}$, hence $\widetilde \caB = (\caB^{\sharp})' \subset (\caA^{\sharp})' = \widetilde \caA$. This proves the last claim of the lemma.

	In particular, we have $\widetilde \caA_{\al} \in \widetilde \caA$ for all $\al$, so $\widetilde \caA \supset \overline{\cap_{\al} \widetilde \caA_{\al}}^w$. To see the opposite inclusion, note that $x \in \left( \overline{\cap_{\al} \widetilde \caA_{\al}}^w \right)'$ if and only if $x \in \widetilde \caA_{\al}' = \caA_{\al}^{\sharp}$ for all $\al$. But then $x \in \caA^{\sharp} = \widetilde \caA'$, so $\left( \overline{\cap_{\al} \widetilde \caA_{\al}}^w \right)' \subset \widetilde \caA'$. Taking commutants then yields $\widetilde \caA \subset \overline{\cap_{\al} \widetilde \caA_{\al}}^w$.
\end{proof}

\subsubsection{near inclusions}

We start with a super-version of \cite{RWW2020}'s Lemma 2.4, which in turn finds its origin in the work \cite{Christensen1977b}:
\begin{lemma} \label{lem:near inclusions and supercommutators}
	Let $\caA, \caB \subset  \caB(\caH)$ be super $C^*$-subalgebras. If $\caA \nsub{\ep} \caB^{\sharp}$, then
	\begin{equation}
		\norm{[a, b]_s} \leq 4 \ep \norm{a} \norm{b}
	\end{equation}
	for all $a \in \caA$, $b \in \caB$.

	Conversely, if $\caB$ is a $\theta$-hyperfinite super von Neumann algebra such that
	\begin{equation}
		\norm{[a, b]_s} \leq \ep \norm{a} \norm{b}
	\end{equation}
	for all $a \in \caA$ and all $b \in \caB$, then $\caA \nsub{3\ep} \caB^{\sharp}$. If $\caA, \widetilde \caB \subset \caM$ for some von Neumann algebra $\caM$, then $\caA \nsub{3\ep} \caB^{\sharp} \cap \caM$. \note{I could just use Bruno's conditional expectation to prove this...}
\end{lemma}

\begin{proof}
	For the first claim, take $a \in \caA$ and $b \in \caB$. By the near inclusion, there is a $c \in \caB^{\sharp}$ such that $\norm{a - c} \leq \ep \norm{a}$, hence
	\begin{equation}
		\norm{[a, b]_s} = \norm{[(a-c), b]_s} \leq 4 \ep \norm{a} \norm{b}
	\end{equation}
	where we used Lemma \ref{lem:basic properties of the supercommutator} (iii).

	For the second claim, take $a \in \caA$. It follows from Proposition \ref{prop:conditional expectation on supercommutant of theta-hyperfinite vN algebra} that $\norm{ \E_{\caB^{\sharp}}(a) - a} \leq 3 \ep \norm{a}$ and $\E_{\caB^{\sharp}}(a) \in \caB^{\sharp}$, \ie $\caA \nsub{3 \ep} \caB^{\sharp}$, as required. In fact, $\E_{\caB^{\sharp}}(a) \in (\{a\} \cup \widetilde \caB)'' \subset (\caA \cup \widetilde \caB)'' \subset \caM$, showing the last claim. 
\end{proof}

Using this, we can show
\begin{lemma} \label{lem:near inclusion of supercommutants implies near inclusion of algebras}
	Let $\caA, \caB \subset \caB(\caH)$ be $\theta$-hyperfinite von Neumann algebras such that $\caB^{\sharp} \nsub{\ep} \caA^{\sharp}$ with $\ep \leq 1/16$. Then $\caA \nsub{48 \ep} \caB$.
\end{lemma}

\begin{proof}
	Using \cite{RWW2020}'s Lemma 2.5 we see that $\caB^{\sharp} \nsub{\ep} \caA^{\sharp}$ implies $(\caA^{\sharp})' \nsub{2 \ep} (\caB^{\sharp})'$. But $(\caA^{\sharp})' = \widetilde \caA'' = \widetilde \caA$ and likewise $(\caB^{\sharp})' = \widetilde \caB$, so
	\begin{equation}
		\widetilde \caA \nsub{2\ep} \widetilde \caB.
	\end{equation}

	It follows from Theorem \ref{thm:near inclusions} that there is an even unitary $u \in (\widetilde \caA \cup \widetilde \caB)''$ such that $u^* \widetilde \caA u \subset \widetilde \caA$ and $\norm{u - \I} \leq 24 \ep$. Using the last claim of Lemma \ref{lem:tilde A of theta-hyperfinite A} this implies $u^* \caA u \subset \caB$ and hence $\caA \nsub{48 \ep} \caB$.
\end{proof}

\begin{lemma} \label{lem:near inclusion respects parity}
	If $\norm{a - b} \leq \ep$ then also $\norm{a_+ - b_+} \leq \ep$ and $\norm{a_- - b_-} \leq \ep$. In particular, if $\caA, \caB$ are super subalgebras of $\caB(\caH)$ such that $\caA \nsub{\ep} \caB$, then also $\caA^{\pm} \nsub{\ep} \caB^{\pm}$.
\end{lemma}

\begin{proof}
	We have
	\begin{equation}
		\norm{a_{\pm} - b_{\pm}} = \norm{ \frac{1}{2} \left( (a - b) \pm \theta(a-b) \right)  } \leq \norm{a - b} \leq \ep,
	\end{equation}
	proving the first claim. The second claim follows immediately from the first.
\end{proof}

\begin{lemma} \label{lem:near inclusion implies near commutation}
	Let $\caA$ be a super subalgebra of $\caB(\caH)$ isomorphic to a full matrix algebra, and let $\caB$ be a super subalgebra of $\caB(\caH)$ such that $\caA \nsub{\ep} \caB$. Then for each $\tilde a \in \widetilde \caA$ and each $c \in \caB^{\sharp}$ we have
	\begin{equation}
		\norm{[\tilde a, c]} \leq (12 + 4 \ep) \ep \norm{\tilde a} \norm{c}.
	\end{equation}
\end{lemma}

\begin{proof}
	Let $\Theta_{\caA}$ be a unitary that squares to $\I$, implementing $\theta$ on $\caA$. Since $\caA \nsub{\ep} \caB$ and using Lemma \ref{lem:near inclusion respects parity}, there is an even $\widetilde \Theta_{\caB} \in \caB$ such that $\norm{\Theta_{\caA} - \widetilde \Theta_{\caB}} \leq \ep$.

	If $c \in \caB^{\sharp}$ is homogeneous, then
	\begin{align*}
		\Theta_{\caA^{\sharp}} c &= \Theta \Theta_{\caA} c = \Theta \widetilde \Theta_{\caB} c + \Theta (\Theta_{\caA} - \widetilde \Theta_{\caB}) c = (-1)^{\tau(c)} c \Theta \widetilde \Theta_{\caB} + \Theta (\Theta_{\caA} - \widetilde \Theta_{\caB}) c \\
					 &= (-1)^{\tau(c)}  c \Theta_{\caA^{\sharp}} + (-1)^{\tau(c)} c \Theta (\widetilde \Theta_{\caB} - \Theta_{\caA}) + \Theta (\Theta_{\caA} - \widetilde \Theta_{\caB}) c
	\end{align*}
	hence
	\begin{equation} \label{eq:c almost supercommutes with Theta A sharp}
		\norm{[\Theta_{\caA^{\sharp}}, c]_s} \leq 2 \ep \norm{c}.
	\end{equation}

	Now take $\tilde a = a_+ + a_- \Theta_{\caA^{\sharp}} \in \widetilde \caA$ where $a = a_+ + a_- \in \caA$. since $\caA \nsub{\ep} \caB$ and by Lemma \ref{lem:near inclusion respects parity} there are homogeneous elements $b_+, b_- \in \caB$ of even and odd parity respectively, such that $\norm{a_{\pm} - b_{\pm}} \leq \ep \norm{a}$. Using this and Eq. \eqref{eq:c almost supercommutes with Theta A sharp} we find
	\begin{align*}
		\norm{[a_+ + a_1 \Theta_{\caA^{\sharp}}, c]} &\leq \norm{ [ (a_+ - b_+) + (a_- - b_-) \Theta_{\caA^{\sharp}}, c  ] } + \norm{ [ b_+ + b_- \Theta_{\caA^{\sharp}}, c ]} \\
							     &\leq 4 \ep \norm{a} \norm{c} + \norm{ b_- [\Theta_{\caA^{\sharp}}, c] } \\
							     & \leq (6 + 2 \ep) \ep \norm{a} \norm{c}.
	\end{align*}
	where we also used $\norm{b_-} \leq \norm{b_- - a_-} + \norm{a_-} \leq (1 + \ep)\norm{a}$. Finally noting that $\norm{\tilde a} = \norm{\Phi_{\caA}(a)} = \norm{a}$ and using linearity of the supercommutator to extend the above to arbitrary elements $c \in \caB^{\sharp}$ yields the claim.
\end{proof}

\begin{lemma} \label{lem:simultaneous near inclusions of matrix super algebras}
	Let $\caA_1, \caA_2$ be super subalgebras of $\caB(\caH)$, both isomorphic to full matrix algebras, containing $\I$, and such that $[\caA_1, \caA_2]_s = 0$. Let $\caB$ be a super von Neumann subalgebra of $\caB(\caH)$ such that $\caA_1, \caA_2 \nsub{\ep} \caB$ with $\ep \leq 1$ and let $\caC \subset \caB(\caH)$ be a super von Neumann algebra such that $\caB \subset \caC$  Then $\caB^{\sharp} \cap \caC \nsub{16 \ep} \Span \left( \caA_1 \cup \caA_2 \right)^{\sharp} \cap (\caC \cup \widetilde \caA_1 \cup \widetilde \caA_2)''$.
\end{lemma}

\begin{proof}
	Take $c \in \caB^{\sharp} \cap \caC$ and consider
	\begin{equation} \label{eq:explicit form of double expectation}
		\E_{\caA_1^{\sharp}} \left( \E_{\caA_2^{\sharp}} (c) \right) = \int_{\caU(\widetilde \caA_1)} \dd u_1 \, \int_{\caU(\widetilde \caA_2)} \dd u_2 \,\, u_1^* u_2^* \, c \, u_2 u_1.
	\end{equation}
	It follows from Lemma \ref{lem:near inclusion implies near commutation} that $\norm{[c, u_2 u_1]} \leq 16 \ep \norm{c}$ so
	\begin{equation}
		\norm{ \E_{\caA_1^{\sharp}} \left( \E_{\caA_2^{\sharp}} (c) \right) - c } \leq 16 \ep.
	\end{equation}
	By Lemma \ref{lem:double conditional expectation supercommutes with both} we have $\E_{\caA_1^{\sharp}} \left( \E_{\caA_2^{\sharp}} (c) \right) \in \Span(\caA_1 \cup \caA_2)^{\sharp}$ and by the explicit formula \eqref{eq:explicit form of double expectation} we also have $\E_{\caA_1^{\sharp}} \left( \E_{\caA_2^{\sharp}} (c) \right) \in (\caC \cup \widetilde \caA_1 \cup \widetilde \caA_2)''$. Since $c \in \caB^{\sharp}$ was arbitrary, this concludes the proof.
\end{proof}

\begin{proposition} \label{prop:simultaneous near inclusions of theta-hyperfinite vN algebras}
	Let $\caA_1, \caA_2$ be $\theta$-hyperfinite von Neumann subalgebras of $\caB(\caH)$ such that $[\caA_1, \caA_2]_s = 0$. Let $\caB$ be a super von Neumann subalgebra of $\caB(\caH)$ such that $\caA_1, \caA_2 \nsub{\ep} \caB$ with $\ep \leq 1$. Let $\caC$ be super von Neumann algebra containing $\caB$. Then $\caB^{\sharp} \cap \caC \nsub{16 \ep} \Span \left( \caA_1 \cup \caA_2 \right)^{\sharp} \cap (\caC \cup \widetilde \caA_1 \cap \widetilde \caA_2)''$.

	If moreover $\Span(\caA_1 \cup \caA_2)''$ and $\caB$ are $\theta$-hyperfinite von Neumann algebras and $\ep \leq 1/256$, then $(\caA_1 \cup \caA_2)^{\sharp \sharp} \nsub{768 \ep} \caB$.
\end{proposition}

\begin{proof}
	Let $\{\caA_i^{(\al)}\}$ be an increasing family of super algebras, isomorphic to full matrix algebras, whose weak closure is $\caA_i$ for $i = 1, 2$. Clearly $[\caA_1^{(\al)}, \caA_2^{(\beta)}]_s = 0$ and $\caA_{i}^{(\al)} \nsub{\ep} \caB$ for all $\al, \beta$ and $i = 1, 2$. It follows from Lemma \ref{lem:simultaneous near inclusions of matrix super algebras} that
	\begin{equation}
		\caB^{\sharp} \nsub{16 \ep} \Span( \caA_1^{(\al)} \cup \caA_2^{(\beta)} )^{\sharp} \cap (\caC \cup \widetilde \caA_1^{\al} \cup \widetilde \caA_2^{\beta})''
	\end{equation}
	for all $\al, \beta$. For $c \in \caB^{\sharp} \cap \caC$ let $a^{(\al, \beta)} \in \Span(\caA_1^{(\al)} \cup \caA_2^{(\beta)})^{\sharp} \cap (\caC \cup \widetilde \caA_1^{\al} \cup \widetilde \caA_2^{\beta})''$ be such that $\norm{c - a^{(\al, \beta)}} \leq \ep \norm{c}$. The net $\{a^{(\al, \beta)}\}$ is contained in the ball of radius $\ep \norm{c}$ centered on $c$ and so has a weak accumulation point $a$ in that ball, \ie $\norm{c - a} \leq \ep \norm{c}$. Moreover, since the $\Span(\caA_1^{(\al)} \cup \caA_2^{(\beta)})^{\sharp}$ form a decreasing family of weakly closed spaces, $a$ is an element of each of them, and clearly $a \in (\caC \cup \widetilde \caA_1 \cup \widetilde \caA_2)''$ hence
	\begin{align*}
		\caB^{\sharp} &\nsub{16\ep} \bigcap_{\al, \beta} \Span(\caA_1^{(\al)} \cup \caA_2^{(\beta)})^{\sharp} \cup (\caC \cup \widetilde \caA_1 \cup \widetilde \caA_2)'' \\
			      &= \Span \left( \big( \cup_{\al} \caA_1^{(\al)} \big) \cup \big( \cup_{\beta} \caA_2^{(\beta)} \big)   \right)^{\sharp} \cup (\caC \cup \widetilde \caA_1 \cup \widetilde \caA_2)'' \\
			      &= \Span ( \caA_1 \cup \caA_2 )^{\sharp} \cup (\caC \cup \widetilde \caA_1 \cup \widetilde \caA_2)''
	\end{align*}
	where we used that the supercommutant of a superspace is the supercommutant of the weak closure of that space.

	If moreover $\Span(\caA_1 \cup \caA_2)''$ and $\caB$ are $\theta$-hyperfinite von Neumann algebras and $\ep \leq 1/256$, then the final claim follows from Lemma \ref{lem:near inclusion of supercommutants implies near inclusion of algebras}.
\end{proof}

\section{Equivariant near inclusions}

We consider super von Neumann algebras that are $G$-invariant subalgebras of $B(\caH)$ (itself equipped with a superstructure $\theta$ and $G$ action $g \mapsto \rho^g$). It is assumed that all group actions commute with the parity operation : $\rho^g \circ \theta = \theta \circ \rho^g$ for all $g \in G$. Let us denote by $G' = \langle G, \theta \rangle$, the group generated by $G$ and parity $\theta$, where the parity operation commutes with all elements of $G$. 

\begin{definition}
	A von Neumann algbra $\caA$ equipped with a $G$-action $g \mapsto \rho^g$ is called $G$-hyperfinite if it is the weak closure of an increasing net of $G$-invariant finite matrix subalgebras, all containing the identity.
\end{definition}

We will prove the following graded equivariant version of \cite{RWW2020}'s Theorem 2.6:
\begin{theorem} \label{thm:near inclusions}
	Let $\caA, \caB \subset \caB(\caH)$ be $G'$-hyperfinite $G$-invariant von Neumann superalgebras with $\caA \nsub{\ep} \caB$ for some $\ep < 1/8$. Then there exists an even $G$-invariany unitary $u \in \big( \caA \cup \caB \big)''$ such that $u \caA u^* \subset \caB$ and $\norm{u - \I} \leq 12 \ep$

	Moreover, $u$ can be chosen such that for all $z \in B(\caH)$ with $z \nin{\delta} \caA$ and $z \nin{\delta} \caB$ we have $\norm{uzu^* - z} \leq 46 \delta \norm{z}$. Also, for all $z \in B(\caH)$ such that $\norm{[z_{\pm}, c]_s} \leq \delta \norm{z_{\pm}} \norm{c}$ for all $c \in \caA \cup \caB$ we have $\norm{u^* z u - z} \leq 40 \delta \norm{z}$.
\end{theorem}

\subsection{Proofs}

First a lemma
\begin{lemma} \label{lem:group averaged expectation is expectation}
	If $\caA \subset \caB(\caH)$ is a $G$-invariant von Neumann algebra and $\E : \caB(\caH) \rightarrow \caA$ is a conditional expectation, then its group average
	\begin{equation}
		\E^G := \frac{1}{\abs{G}} \sum_{g \in G} \rho^g \circ \E \circ \rho^{g^{-1}}
	\end{equation}
	is an equivariant conditional expectation.
\end{lemma}

\begin{proof}
	Equivariance of $\E^G$ follows immediately from the definition:
	\begin{equation}
		\rho^h \circ \E^G \circ \rho^{h^{-1}} = \frac{1}{\abs{G}} \sum_{g \in G} \rho^{gh} \circ \E \circ \rho^{(gh)^{-1}} = \E^G.
	\end{equation}
	We now check that $\E^G$ is still a conditional expectation. We have contractivity:
	\begin{align*}
		\norm{\E^G(x)} &\leq \frac{1}{\abs{G}} \sum_{g \in G} \norm{ \rho^g( \E (\rho^{g^{-1}}(x))) } \\
			       &= \frac{1}{\abs{G}} \sum_{g \in G} \norm{ \E(\rho^{g^{-1}}(x))} \\
			       &\leq \frac{1}{\abs{G}} \sum_{g \in G} \norm{\rho^{g^{-1}}(x)} \\
			       &= \norm{x}.
	\end{align*}

	If $y \in \caB(\caH)$ and $x, z \in \caA$ then
	\begin{align*}
		\E^G(xyz) &= \frac{1}{\abs{G}} \sum_{g} \rho^g \left( \E \left(  \rho^{g^{-1}}(x)  \rho^{g^{-1}}(y) \rho^{g^{-1}}(z) \right) \right) \\
			  &= \frac{1}{\abs{G}} \sum_{g \in G} \rho^g \left( \rho^{g^{-1}}(x) \E( \rho^{g^{-1}}(y) ) \rho^{g^{-1}}(z) \right) \\
			  &= x \left(  \frac{1}{\abs{G}} \sum_{g \in G} \rho^g \left( \E( \rho^{g^{-1}}(y) ) \right)  \right) \\
			  &= x \E^G(y) z.
	\end{align*}
	where we used property iii in the definition \ref{def:conditional expectation} of a conditional expectation.

	Finally, we check complete positivity. We want to show that for any $k \in \N$, the map
	\begin{equation}
		(\E^G)^{k \times k} : \caB(\caH)^{k \times k} \rightarrow \caA^{k \times k} : \begin{bmatrix} x_{1 1} & \cdots & x_{1 k} \\
		\vdots & \, & \vdots \\ 
		x_{k 1} & \cdots & x_{k k} \end{bmatrix}
			\mapsto \begin{bmatrix} \E^G(x_{1 1}) & \cdots & \E^G(x_{1 k}) \\
				\vdots & \, & \vdots \\
				\E^G(x_{k 1}) & \cdots & \E^G(x_{k k })
		\end{bmatrix}
	\end{equation}
	is positive. But we have
	\begin{equation}
		(\E^G)^{k \times k} = \frac{1}{\abs{G}} \sum_{g \in G} \rho^g \circ \E^{k \times k} \circ \rho^{g^{-1}}
	\end{equation}
	which is a convex combination of positive maps, and therefore itself positive.
\end{proof}

We generalize Proposition B.1. of \cite{RWW2020}.

\begin{proposition} \label{prop:making homomorphisms inner}
	Consider $G$-invariant $C^*$-superalgebras $\caA_i, \caB_i \subset B(\caH)$ for $i = 1, \cdots, n$ that satisfy $[\caA_i, \caA_j]_s = [\caB_i, \caB_j]_s = 0$ for $i \neq j$, and such that each $\caA_i''$ and $\caB_i''$ is a hyperfinite $G$-invariant von Neumann superalgebra. Denote $\caA = \big( \cup_i \caA_i \big)''$ and $\caB = \big( \cup_i \caB_i \big)''$. Consider unital $G$-equivariant $*$-superhomomorphisms $\Phi_i : \caA_i \rightarrow \caB_i$ with $\norm{\Phi_i(a_i) - a_i} \leq \gamma_i \norm{a_i}$ for all $a_i \in \caA_i$, with $\sum_i \gamma_i < 1$. Then there exists a $G$-invariant even unitary  $u \in (\caA \cup \caB)''$ such that $\Phi_i(a_i) = u^* a_i u$ for all $i$ and $a_i \in \caA_i$, with
	\begin{align}
		\norm{\I - u} &\leq \sqrt{2} \ep \big( 1 + (1 - \ep^2)^{1/2} \big)^{-1/2} \leq \sqrt{2} \ep, \\
		\ep &= \sum_i \gamma_i
	\end{align}
	where the middle expression in the first line is $\sqrt{2} \ep + \caO(\ep^2)$. Moreover, for $\ep \leq 1/8$, the unitary $u$ can be chosen such that for any $z \in B(\caH)$, if $\norm{[z_{\pm}, c]_s} \leq \delta \norm{z_{\pm}} \norm{c}$ for all $c \in \caA \cup \caB$ then $\norm{u^* z u - z} \leq 40 \delta \norm{z}$.
\end{proposition}

\begin{proof}
	By Proposition \ref{prop:G-equiv super homomorphisms extend to vN} the homomorphisms $\Phi_i : \caA_i \rightarrow \caB_i$ can be extended to $G$-equivariant super $*$-isomorphisms $\Phi_i' : \caA''_i \rightarrow \Phi_i(\caA''_i) \subset \caB''_i$ with $\norm{\Phi'_i(a_i) - a_i} \leq \gamma_i \norm{a_i}$ for all $a_i \in \caA''_i$. Because the $\caA''_i$ and $\caB''_i$ are assumed to be hyperfinite, we can now without loss of generality assume that the $\caA_i$, $\Phi_i(\caA_i)$ and $\caB_i$ are hyperfinite $G$-invariant von Neumann superalgebras.

	Consider the following subalgebras of $\caB(\caH \oplus \caH)$:
	\begin{equation}
		\caC_i := \left\lbrace \begin{bmatrix} a_i & \, \\ \, & \Phi_i(a_i) \end{bmatrix} \, : \, a_i \in \caA_i \right\rbrace.
	\end{equation}
	The ambient $\caB(\caH \oplus \caH)$ is $\Z_2$-graded by $\theta \oplus \theta$ and carries $G$-actions $\rho^g \oplus \rho^g$. The subalgebra $\caC_i$ is a hyperfinite $G$-invariant von Neumann superalgebra with respect to this $\Z_2$-grading and these symmetry actions.

	Being hyperfinite, the algebras $\caC_i$ have property P, which we first apply to $\caC_1$ and the even element
	\begin{equation}
		x_0 = \begin{bmatrix} 0 & \I \\ 0 & 0 \end{bmatrix} \in \caB(\caH \oplus \caH)
	\end{equation}
	yielding an element $x'_1 \in \caC'_1$ that is also in 
	\begin{equation}
		\overline{   \hull \{  c_1^* x_0 c_1 \, : \, c_1 \in \caU(\caC_1) \}   }^{w}
	\end{equation}

	Both this set and $\caC_1$ are $\langle \theta, G\rangle$-invariant so the even $G$-invariant element
	\begin{equation}
		\tilde x_1 = \frac{1}{2 \abs{G}} \sum_{g \in G} \left(  \rho^g(x_1) + \theta(\rho^g(x_1))  \right)
	\end{equation}
	is also in $\caC'_1$ and in
	\begin{equation}
		\overline{  \hull \{ c_1^* x_0 c_1 \, : \, c_1 \in \caU(\caC_1) \}  }^{w}.
	\end{equation}

	Elements $c_1^* x_0 c_1$ are of the form
	\begin{equation}
		c_1^* x_0 c_1 = \begin{bmatrix} 0 & u_1^* \Phi_1(u_1) \\ 0 & 0 \end{bmatrix}
	\end{equation}
	so
	\begin{equation}
		\tilde x_1 = \begin{bmatrix} 0 & y_1 \\ 0 & 0 \end{bmatrix}
	\end{equation}
	for some $G$-invariant even element $y_1 \in (\caA_1 \cup \caB_1)''$.

	The fact that $\tilde x_1 \in \caC'_1$ implies that
	\begin{equation}
		\begin{bmatrix}
			a_1 & \, \\ \, & \Phi_1(a_1)
		\end{bmatrix} 
		\begin{bmatrix}
			\, & y_1 \\ \, & \,
		\end{bmatrix}
		=
		\begin{bmatrix}
			\, & y_1 \\ \, & \,
		\end{bmatrix}
		\begin{bmatrix}
			a_1 & \, \\ \, & \Phi_1(a_1)
		\end{bmatrix}
	\end{equation}
	hence
	\begin{equation}
		a_1 y_1 = y_1 \Phi_1(a_1)
	\end{equation}
	for all $a_1 \in \caA_1$.

	Let us now assume that we have a $G$-invariant even element
	\begin{equation}
		\tilde x_i = \begin{bmatrix}
			\, & y_i \\ \, & \,
		\end{bmatrix}
	\end{equation}
	that is in $\caC'_j$ for all $j \leq i$ and such that $a_j y_i = y_i \Phi_j(a_j)$ for all $a_j \in \caA_j$ for $j \leq i$, and
	\begin{equation}
		y_i \in \overline{  \hull \{ u_i^* \cdots u_1^* \Phi_1(u_1) \cdots \Phi_i(u_i) \, : \, u_j \in \caU(\caA_j) \}  }^{w}.
	\end{equation}

	We then apply property P of $\caC_{i+1}$ to the element $\tilde x_i \in \caB(\caH \oplus \caH)$ to obtain 
	\begin{equation}
		x_{i+1} \in \caC'_{i+1} \cap \overline{  \hull \{ c_{i+1}^* \tilde x_i c_{i+1} \, : \, c_{i+1} \in \caU(\caC_{i+1}) \}  }^{w}.
	\end{equation}
	Both sets being intersected are $\langle \theta, G\rangle$-invariant so the $G$-invariant even element
	\begin{equation}
		\tilde x_{i+1} = \frac{1}{2 \abs{G}} \sum_{g \in G} \left( \rho^{g}(x_{i+1}) + \theta(\rho^g(x)) \right)
	\end{equation} 

	Elements $c_{i+1}^* \tilde x_i c_{i+1}$ are of the form
	\begin{equation}
		c_{i+1}^* \tilde x_i c_{i+1} =
		\begin{bmatrix}
			\, & u_{i+1}^* y_i \Phi_{i+1}(u_{i+1}) \\ \, & \, 
		\end{bmatrix}
	\end{equation}
	so
	\begin{equation}
		\tilde x_{i+1} =
		\begin{bmatrix}
			\, & y_{i+1} \\
			\, & \,
		\end{bmatrix}
	\end{equation}
	for some $G$-symmetric even element $y_{i+1}$ of $\left( \cup_{j = 1}^i \caA_j \cup \caB_j \right)''$ in
	\begin{equation}
		\overline{  \hull \{ u_{i+1}^* u_i^* \cdots u_1^* \Phi_1(u_1) \cdots \Phi_i(u_i) \Phi_{i+1}(u_{i+1}) \, : \, u_j \in \caU(\caA_j) \}  }^{w}.
	\end{equation}

	The fact that $\tilde x_{i+1} \in \caC'_{i+1}$ implies that
	\begin{equation}
		a_{i+1} y_{i+1} = y_{i+1} \Phi_{i+1}(a_{i+1})
	\end{equation}
	for all $a_{i+1} \in \caA_{i+1}$.

	Moreover, $\tilde x_i$ is an even element of $\caC'_j$ for all $j \leq i$, so it is in the supercommutant of these $\caC_j$. The algebra $\caC_{i+1}$ graded commutes all these $\caC_j$, so all elements $c_{i+1}^* \tilde x_i c_{i+1}$ for $c_{i+1} \in \caU(\caC_{i+1})$ are in the graded commutant of the $\caC_j$ for $j \leq i$. It follows that
	\begin{equation}
		\tilde x_{i+1} \in \overline{  \hull \{ c_{i+1}^* \tilde x_i c_{i+1} \, : \, c_{i+1} \in \caU(\caC_{i+1}) \}  }^{w}
	\end{equation}
	is in the supercommutant of each of the $\caC_j$ for $j \leq i$. But $\tilde x_{i+1}$ is even, so it is in fact also in the commutants $\caC'_j$ which implies
	\begin{equation}
		a_j y_{i+1} = y_{i+1} \Phi_{j}(a_j)
	\end{equation}
	for all $a_j \in \caA_j$ and all $j \leq i$. We have already shown that the same is true for all $a_{i+1} \in \caA_{i+1}$ also.

	By induction, we have shown that there is a $G$-invariant even element $y = y_n \in (\caA \cup \caB)''$ in
	\begin{equation} \label{eq:convex set}
		\overline{  \hull \{ u_n^* \cdots u_1^* \Phi_1(u_1) \cdots \Phi_n(u_n) \, : \, u_i \in \caU(\caA_i) \}  }^{w}
	\end{equation}
	such that for all $i \in \{1, \cdots, n\}$ and all $a_i \in \caA_i$,
	\begin{equation} \label{eq:intertwining of y}
		a_i y = u \Phi_i(a_i).
	\end{equation}

	Now, $y$ is close to the identity because it is in the weak operator closure of the convex hull of elements
	\begin{align*}
		\norm{\I - u_n^* \cdots u_1^* \Phi_1(u_1) \cdots \Phi_n(u_n)} &= \norm{\I - u_n^* \Phi_n(u_n) + u_n^* \Phi_n(u_n) - u_n^* \cdots u_1^* \Phi_1(u_1) \cdots \Phi_n(u_n)} \\
									      &\leq \norm{u_n^* \Phi_n(u_n) - u_n^* \cdots u_1^* \Phi_1(u_1) \cdots \Phi_n(u_n) } + \norm{u_n - \Phi_n(u_n)} \\
									      &= \norm{\I - u_{n-1}^* \cdots u_1^* \Phi_1(u_1) \cdots \Phi_{n-1}(u_{n-1})} + \norm{u_n - \Phi_n(u_n)} \\
									      & \vdots \\
									      &\leq \sum_{i = 1}^n \norm{u_i - \Phi_i(u_i)},
	\end{align*}
	hence
	\begin{equation}
		\norm{\I - y} \leq \sum_{i = 1}^n \gamma_i = \ep.
	\end{equation}

	We now define the $G$-invariant even unitary $u = y \abs{y}^{-1}$. By Lemma 2.7 in \cite{Christensen1975} we have
	\begin{equation}
		\norm{\I - u} \leq \sqrt{2} \ep \left( 1 + (1 - \ep^2)^{1/2} \right)^{-1/2} \leq \sqrt{2} \ep.
	\end{equation}
	We will now show that $u^* a_i u = \Phi_i(a_i)$ for all $i = 1, \cdots, n$ and all $a_i \in \caA_i$.

	It follows from Eq. \eqref{eq:intertwining of y} that $y^* y = \Phi_i(u_i)^* y^* y \Phi_i(u_i)$ for any $u_i \in \caU(\caA_i)$, hence $[\Phi_i(u_i), y^* y] = 0$. Since any $a_i \in \caA_i$ can be written as a linear combination of unitary elements in $\caA_i$, it follows that $[\Phi_i(a_i), y^* y] = 0$, hence $[\Phi_i(a_i), \abs{y}^{-1}] = 0$ and
	\begin{equation}
		u^* a_i u = \abs{y}^{-1} y^* a_i y \abs{y}^{-1} = \abs{y}^{-1} y^* y \Phi_i(a_i) \abs{y}^{-1} = \abs{y}^{-1} y^* y \abs{y}^{-1} \Phi_i(a_i) = \Phi_i(a_i)
	\end{equation}
	for any $i$ and any $a_i \in \caA_i$.

	Finally, consider a $z \in \caB(\caH)$ such that $\norm{[z_{\pm}, c]_s} \leq \delta \norm{z_{\pm}} \norm{c}$ for all $c \in \caA \cup \caB$. We want to show that $\norm{[z, y]}$ is small. Since for homogeneous elements $z_{\pm}, a, b$ we have the chain rule
	\begin{equation}
		[z_{\pm}, ab]_s = [z_{\pm}, a]_s b + (-)^{\tau(z_{\pm}) \tau(a)} a [z_{\pm, b}]_s	
	\end{equation}
	we find that $\norm{[z, x]_s} \leq 8 \delta \norm{z}$ for any element $x$ in
	\begin{equation}
		\hull \{ u_n^* \cdots u_1^* \Phi_1(u_1) \cdots \Phi_n(u_n) \, : \, u_i \in \caU(\caA_i).
	\end{equation}
	Since $y$ is an in the weak closure of this convex hull there is a net of elements $y_i$ in the convex hull converging wekaly to $y$. For each $y_i$ we have $\norm{[z, y_i]_s} \leq 8 \delta \norm{z}$ and since norm-balls in $\caB(\caH)$ are weakly closed, also $\norm{[z, y]_s} \leq 8 \delta \norm{z}$. Noting that $y$ is even, the supercommutator con be relaced by a commutator so $\norm{[z, y]} \leq 8 \delta \norm{z}$. In the same way we find that $\norm{[z, y^*]} \leq 8 \delta \norm{z}$ so, using Lemma A.2 of \cite{RWW2020} we conclude that
	\begin{equation}
		\norm{u^* z u - z} = \norm{[z, u]} \leq 3 \norm{[z, y]} + 2 \norm{[z, y^*]} \leq 40 \delta \norm{z}.
	\end{equation}
\end{proof}

The following generalizes Proposition B.2 of \cite{RWW2020}.

\begin{proposition} \label{prop:G-equiv super homomorphisms extend to vN}
	Given a unital $G$-symmetric $C^*$-superalgebra $\caA \subset B(\caH)$ and a unital $G$-equivariant $*$-superhomomorphism $\Phi : \caA \rightarrow B(\caH)$ such that $\norm{ \Phi(a) - a } \leq \ep \norm{a}$ for all $a \in \caA$ and some $\ep < 1$, then $\Phi$ can be extended to a $G$-equivariant $*$-superisomorphism $\Phi' : \caA'' \rightarrow \Phi(\caA)''$ with $\norm{\Phi(a) - a} \leq \ep \norm{a}$ for all $a \in \caA''$.
\end{proposition}

\begin{proof}
	Note first of all that the von Neumann algebra $\caA''$ is indeed a $G$-symmetric superalgebra by Lemma \ref{lem:G-invariance extends to vN}.

	Consider the algebra
	\begin{equation}
		\caC := \left\lbrace \begin{bmatrix} a & \, \\ \, & \Phi(a) \end{bmatrix} \, : \, a \in \caA \right\rbrace
	\end{equation}
	which is a subalgebra of $\caB(\caH \oplus \caH)$. It follows from the fact that $\caA$ is a $G$-invariant superalgebra and that $\Phi$ is a $G$-equivariant superhomomorphism that $\caC$ is a $G$-invariant superalgebra with respect to the $G$-symmetries $\rho^g \oplus \rho^g$ and the parity operator $\theta \oplus \theta$.

	Take $a \in \caA''$. We show that there is a unique $b \in \Phi(\caA)''$ such that
	\begin{equation}
		c = \begin{bmatrix} a & \, \\ \, & b \end{bmatrix} \in \caC''.
	\end{equation}

	Kaplansky's density theorem gives us a net $\{a_i\}$ in $\caA$, converging to $a$ in the weak operator topology such that $\norm{a_i} \leq \norm{a}$ for all $i$.But then $\norm{\Phi(a_i) - a} \leq \ep \norm{a}$ and $\norm{\Phi(a_i)} \leq (1 + \ep)\norm{a}$ so the net
	\begin{equation}
		c_i = \begin{bmatrix} a_i & \, \\ \, & \Phi(a_i) \end{bmatrix} \in \caC''
	\end{equation}
	is contained within a closed ball of finite radius in $\caB(\caH \oplus \caH)$. Any such ball is weak-$*$ compact, and therefore also compact in the weak operator topology, so $\{c_i\}$ admits a subnet converging in the weak operator topology to some
	\begin{equation}
		c = \begin{bmatrix} a  & \, \\ \, & b  \end{bmatrix} \in \caC''.
	\end{equation}

	This show the existence of $b$, we now show uniqueness. Suppose there are $b_1, b_2 \in \Phi(\caA)''$ such that
	\begin{equation}
		\begin{bmatrix}
			a & \, \\ \, & b_1
		\end{bmatrix},
		\begin{bmatrix}
			a & \, \\ \, & b_2
		\end{bmatrix} \in \caC''.
	\end{equation}
	Then
	\begin{equation}
		c = \begin{bmatrix} 0 & \, \\ \, & z \end{bmatrix} \in \caC'' \quad \text{with} \quad z = b_1 - b_2 \in \Phi(\caA)''.
	\end{equation}
	We will show that $z = 0$.

	Kaplansky's density theorem provides a net $\{c_i\}$ in $\caC$ that covnerges strongly to $c$ and such that $\norm{c_i} \leq \norm{c} = \norm{z}$ for all $i$. Since $c_i \in \caC$, by definition there is a net $\{a_i\}$ in $\caA$ such that
	\begin{equation}
		c_i = \begin{bmatrix} a_i & \, \\ \, & \Phi(a_i) \end{bmatrix}.
	\end{equation}
	Since the $c_i$ converge strongly to $c$, the $a_i$ converge strongly to zero with $\norm{a_i} \leq \norm{z}$ for all $i$, and the $\Phi(a_i)$ converge strongly to $z$. It follows that $\Phi(a_i) - a_i$ converges strongly to $z$. By lower semicontinuity of the norm in the strong topology, we find that
	\begin{equation}
		\norm{z} \leq \limsup_{i} \norm{\Phi(a_i) - a_i} \leq \ep \norm{z}.
	\end{equation}
	Since $\ep < 1$, we must have $z = 0$, hence $b_1 = b_2$.

	We define $\Phi' : \caA'' \rightarrow \Phi(\caA)''$ by $\Phi(a) = b$.

	Bijectivity of this $\Phi'$ follows if we can show that for any $b \in \Phi(\caA)''$ there is a unique $a \in \caA''$ suc that
	\begin{equation}
		c = \begin{bmatrix} a & \, \\ \, & b \end{bmatrix} \in \caC''.
	\end{equation}

	We first show existence of such an $a$. Kaplansky's density theorem provides a net $\{b_i\}$ in $\Phi(\caA)$ converging to $b$ in the weak operator topology and such that $\norm{b_i} \leq \norm{b}$ for all $i$. Then there is also a net $\{a_i\} \in \caA$ such that $b_i = \Phi(a_i)$, so
	\begin{equation}
		c_i = \begin{bmatrix} a_i & \, \\ \, & b_i \end{bmatrix} \in \caC
	\end{equation}
	is a net in $\caC$. Since $\Phi$ is a $*$-homomorphism, $\norm{a_i} = \norm{b_i} \leq \norm{b}$, so the net $\{c_i\}$ is contained in a bounded closed ball in $\caB(\caH \oplus \caH)$. By Banach-Alaoglu, $\{c_i\}$ admits a subnet converging in the weak operator topology to some
	\begin{equation}
		c = \begin{bmatrix} a & \, \\ \, & b \end{bmatrix} \in \caC''
	\end{equation}
	with $a \in \caA''$. This shows the existence of $a$.

	To see uniqueness, suppose $a_1, a_2 \in \caA''$ are both such that
	\begin{equation}
		\begin{bmatrix}
			a_1 & \, \\ \, & b
		\end{bmatrix},
		\begin{bmatrix}
			a_2 & \, \\ \, & b
		\end{bmatrix} \in \caC''.
	\end{equation}
	Then
	\begin{equation}
		\begin{bmatrix}
			z & \, \\ \, & 0
		\end{bmatrix} \in \caC'' \quad \text{with} \quad z = a_1 - a_2 \in \caA''.
	\end{equation}
	Kaplansky's density theorem gives us a net $\{c_i\} \in \caC$ converging strongly to $c$ with $\norm{c_i} \leq \norm{c} = \norm{z}$. Let
	\begin{equation}
		c_i = \begin{bmatrix} a_i & \, \\ \, \Phi(a_i) \end{bmatrix}
	\end{equation}
	for some net $\{a_i\}$ in $\caA$. Then the net $\{ \Phi(a_i) \}$ converges strongly to zero with $\norm{a_i} = \norm{\Phi(a_i)} \leq \norm{z}$ and $\{a_i\}$ converges strongly to $z$, so $\Phi(a_i) - a_i$ converges strongly to $z$. But
	\begin{equation}
		\norm{z} \leq \limsup_i \norm{\Phi(a_i) - a_i} \leq \ep \norm{z}.
	\end{equation}
	since $\ep < 1$, we conclude that $z = 0$ and $a_1 = a_2$.

	This shows that $\Phi' : \caA'' \rightarrow \Phi(\caC)''$ is a bijection.

	We must still check that it is a $G$-equivariant super $*$-isomorphism.

	For linearity, take any $a_1, a_2 \in \caA''$ and any $\lambda \in \C$. Then $\Phi'(a_1)$ and $\Phi'(a_2)$ are the unique elements of $\Phi(\caA)''$ such that
	\begin{equation}
		\begin{bmatrix}
			a_1 & \, \\ \, & \Phi'(a_1)
		\end{bmatrix},
		\begin{bmatrix}
			a_2 & \, \\ \, & \Phi'(a_2)
		\end{bmatrix} \in \caC''.
	\end{equation}
	Since $\caC''$ is a linear space,
	\begin{equation}
		\begin{bmatrix}
			a_1 + \lambda a_2 & \, \\ \, & \Phi'(a_1) + \lambda \Phi'(a_2)
		\end{bmatrix} \in \caC''.
	\end{equation}
	Hence, by definition of $\Phi'$, we have that $\Phi'(a_1 + \lambda a_2) = \Phi'(a_1) + \lambda \Phi'(a_2)$. So $\Phi'$ is linear.

	Similarly, multiplicativity of $\Phi'$ follows from
	\begin{equation}
		\begin{bmatrix}
			a_1 & \, \\ \, & \Phi'(a_1) 
		\end{bmatrix}
		\begin{bmatrix}
			a_2 & \, \\ \, \Phi'(a_2)
		\end{bmatrix}
		=
		\begin{bmatrix}
			a_1 a_2 & \, \\ \, & \Phi'(a_1) \Phi'(a_2)
		\end{bmatrix},
	\end{equation}
	hence $\Phi'(a_1 a_2) = \Phi'(a_1) \Phi'(a_2)$.

	The $*$-property follows from
	\begin{equation}
		\begin{bmatrix}
			a & \, \\ \, & \Phi'(a)
		\end{bmatrix}^*
		=
		\begin{bmatrix}
			a^* & \, \\ \, & \Phi'(a)^*
		\end{bmatrix}.
	\end{equation}

	This shows that $\Phi' : \caA'' \rightarrow \Phi(\caA)''$ is a $*$-isomorphism.

	To see that it is a $G$-equivariant super isomorphism, we need to show that it commutes with the $\rho^g$ and with the parity operator $\theta$. The argument is the same in both cases, so let $\rho$ be a $*$-automorphism leaving $\caA''$ and $\Phi(\caA)''$ invariant, we then show that $\Phi'(\rho(a)) = \rho(\Phi'(a))$ for any $a \in \caA''$. The algebra
	\begin{equation}
		\caC = \left\lbrace  \begin{bmatrix} a & \, \\ \, & \Phi(a) \end{bmatrix} \, : \, a \in \caA \right\rbrace \subset \caB(\caH \oplus \caH)
	\end{equation}
	is invariant under the action of $\rho \oplus \rho$, and by Lemma \ref{lem:G-invariance extends to vN}, so is $\caC''$. Let $a \in \caA''$, then
	\begin{equation}
		(\rho \oplus \rho) \left( \begin{bmatrix} a & \, \\ \, & \Phi'(a) \end{bmatrix}  \right) = \begin{bmatrix} \rho(a) & \, \\ \, & \rho(\Phi'(a)) \end{bmatrix} \in \caC''
	\end{equation}
	so, by denfinition of $\Phi'$, we have $\Phi'(\rho(a)) = \rho(\Phi'(a))$, as required.
\end{proof}

We now come to the proof of the  main theorem:

\begin{proofof}[Theorem \ref{thm:near inclusions}]
	Since $\caB$ is hyperfinite, it is injective and so there exists a conditional expectation
	\begin{equation}
		\widetilde \E_{\caB} : \caB(\caH) \rightarrow \caB
	\end{equation}
	onto $\caB$. 

	We take a group average over $G'$ to obtain a conditional expectation that is $G$-equivariant and respects the $\Z_2$-graded structure (lemma \ref{lem:group averaged expectation is expectation}):
	\begin{equation}
		\E_{\caB} = \frac{1}{\abs{G'}} \sum_{g \in G'} \, \rho^g \circ \widetilde E_{\caB} \circ (\rho^g)^{-1}
	\end{equation}
	where $\rho^{\theta \, g} = \theta \circ \rho^{g}$ for all $g \in G$. \margin{establish better notation for the goup $G'$ and its representations...}

	From the equivariant Stinespring theorem \ref{prop:equivariant Stinespring} we get a Hilbert space $\caK$ equipped with a unitary representation of $G'$, the group generated by $G$ and parity. This representaion is given by $g \mapsto \tilde \rho^g \in \caU(\caK)$ for $g \in G'$ which satisfies $(\tilde \rho^{\theta})^2 = \I$ and $[\tilde \rho^g, \tilde \rho^{\theta}] = 0$ for all $g \in G$. Further, there is a unital $*$-homomorphism $\pi : \caB(\caH) \rightarrow \caB(\caK)$ and an isometry $v : \caH \rightarrow \caK$ such that
	\begin{equation} \label{eq:Stinespring rep}
		\E_{\caB}(x) = v^* \pi(x) v \quad \text{for all} \,\,\, x \in \caB(\caH)
	\end{equation}
	and
	\begin{equation}
		\pi(\rho^{g}(x)) = \tilde \rho^{g}(\pi(x)) \quad \text{for all } \,\,\, x \in \caB(\caH) \,\,\, \text{and all} \,\,\, g \in G'.
	\end{equation}
	Moreover, the subspace $v(\caH) \subset \caK$ is invariant under all the $\tilde \rho^g$ for $g \in G'$.

	Let $p = vv^* \in \caB(\caK)$, it is a $G'$-invariant projector. Moreover, for $b \in \caB$ we have
	\begin{equation}
		\E_{\caB}(b) = b = v^* \pi(b) v
	\end{equation}
	so the $*$-homomorphism
	\begin{equation}
		b \mapsto v b v^* = p \pi(b) p
	\end{equation}
	is injective.

	We have
	\begin{align*}
		p \pi(b^2) p &= v v^* \pi(b^2) v v^* = v b^2 v^* \\
			     &= v b v^* v b v^* = v v^* \pi(b) v v^* v v^* \pi(b) v v^* \\
			     &= p \pi(b) p \pi(b) p
	\end{align*}
	so $p \pi(b) p^\perp \pi(b) p = 0$ for all $b \in \caB$. If $b$ is self-adjoint, it follows that $p^{\perp} \pi(b) p = 0$ hence $[\pi(b), p] = 0$. But $\caB$ is a $*$-algebra so all of its elements are linear combination of self-adjoint elements and we have $[\pi(b), p] = 0$ for all $b \in \caB$, that is, $p \in \pi(\caB)'$.

	We will now show that $p$ nearly commutes with $\pi(\caA)$. Pick $a \in \caA$ and choose $b \in \caB$ such that $\norm{a - b} \leq \ep \norm{a}$. Then
	\begin{equation}
		\norm{[\pi(a), p]} = \frac{1}{2} \norm{[\pi(a-b), 2p - \I]} \leq \norm{\pi(a - b)} \norm{2p - \I} \leq \ep \norm{a}.
	\end{equation}

	We would now like to use Proposition \ref{prop:equivariant conditional expectation} to get conditional expectation onto $\pi(\caA)'$, but we have no guarantee that $\pi(\caA)''$ is $G'$-hyperfinite. Instead, using that $\caA$ is $G'$-hyperfinite we note that there exists an AF $C^*$-algebra $\caA_0 \subset \caA$ which is the norm closure of an increasing net of $G'$-invariant finite matrix algebras and such that $\caA = \caA_0''$. The $C^*$-algebra $\caA_0$ is automatically $G'$-invariant.

	Then $\pi(\caA_0)$ is also AF, and $\pi(\caA_0)$ is $G'$-hyperfinite and $G'$-invariant. We can therefore apply Proposition \ref{prop:equivariant conditional expectation} to this von Neumann algebra to get a $G'$-equivariant conditional expectation $\E_{\pi(\caA_0)'}$ which, when applied to $p$, gives a $G'$-invariant element
	\begin{equation}
		\E_{\pi(\caA_0)'}(p) \in \pi(\caA_0)'
	\end{equation}
	that satisfies
	\begin{equation}
		\norm{\E_{\pi(\caA_0)'}(p) - p} \leq \ep.
	\end{equation}
	We then further project $\E_{\pi(\caA_0)'}(p)$ onto the $G'$-invariant von Neumann algebra $\big( \{p\} \cup \pi(\caB(\caH))  \big)''$. By Corollary 1.3.2 of \cite{Arveson1969}, $\big( \{p\} \cup \pi(\caB(\caH)) \big)'$ is isomorphic to $\caB'$ \margin{Check out the Arveson corollary} and hence injective, so we can again use a $G'$-invariant conditional expectation, defining a $G'$-invariant element
	\begin{equation}
		x = \E_{(\{p\} \cup \pi(\caB(\caH)))''} \left( \E_{\pi(\caA_0)'}(p) \right) \in \big( \{p\} \cup \pi(\caB(\caH)) \big)''
	\end{equation}
	such that
	\begin{equation}
		\norm{x - p} \leq \ep
	\end{equation}
	where the latter bound follows from $\norm{\E_{\pi(\caA_0)'} - p} \leq \ep$ and the fact that conditional expectations are contractions.

	Moreover, we have $x \in \pi(\caA_0)'$. Indeed, $[x, a] = 0$ for any $a \in \pi(\caA_0)$ because $\E_{\pi(\caA_0)'}(p) \in \pi(\caA_0)'$ and $\pi(\caA_0) \subset \big( \{p\} \cup \pi(\caB(\caH)) \big)''$, and using property (iii) in definition \ref{def:conditional expectation} of the conditional expectation.

	The operator $x$ is self-adjoint and $\ep$-close to the projector $p$, so its spectrum is contained in $[-\ep, \ep] \cup [1 - \ep, 1  + \ep]$. Define $q \in \pi(\caA_0)'$ to be the spectral projection on the spectral subspace of $x$ associated to $[1 - \ep, 1 + \ep]$. Then $\abs{x - q} \leq \ep$ so $\norm{p - q} \leq 2 \ep$. Define the $G'$-invariant element
	\begin{equation}
		y = qp + q^\perp p^\perp.
	\end{equation}
	This element is close to the identity:
	\begin{equation}
		\norm{y - \I} = \norm{(2q - \I)(p -q)} \leq \norm{p-q} \leq 2\ep.
	\end{equation}
	Now define the $G'$-invariant unitary $w = y \abs{y}^{-1}$. By Lemma 2.7 of \cite{Christensen1975}, it is also close to the identity:
	We now define the $G$-invariant even unitary $u = y \abs{y}^{-1}$. By Lemma 2.7 in \cite{Christensen1975} we have
	\begin{equation} \label{eq:w close to I}
		\norm{w - \I} \leq 2 \sqrt{2} \ep.
	\end{equation}
	Since $y^* y = pqp + p^\perp q^\perp p^\perp$, we have $[p, y^* y] = 0$ so $[p, \abs{y}^{-1}] = 0$. Moreover, $yp = qy$, so
	\begin{equation} \label{eq:w intertwines p and q}
		w p w^* = y \abs{y}^{-1} p \abs{y}^{-1} y^* = y p \abs{y}^{-1} \abs{y}^{-1} y^* = q.
	\end{equation}

	Now define the $G'$-equivariant map
	\begin{equation}
		\Phi : \caA_0 \rightarrow \caB \subset \caB(\caH) \, : \, a \mapsto v^* w^* \pi(a) w v.
	\end{equation}
	This is a unital $*$-homomorphism, indeed
	\begin{align*}
		\Phi(a_1) \Phi_{a_2} &= v^* w^* \pi(a_1) w p w^* \pi(a_2) w v = v^* w^* \pi(a_1) q \pi(a_2) w v \\
				     &= v^* w^* \pi(a_1 a_2) q w v = v^* w^* \pi(a_1 a_2) w p v = \Phi(a_1 a_2),
	\end{align*}
	where we used $p = v v^*$, the fact that $q \in \pi(\caA_0)'$ and that $v$ is an isometry. To see that the image of $\Phi$ lies in $\caB$, note that $w^* \pi(a) w \in \big( \{p\} \cup \pi(\caB(\caH)) \big)''$ and recall Eq. \eqref{eq:Stinespring rep}.

	Moreover, for any $a \in \caA_0$ there exists $b \in \caB$ with $\norm{a - b} \leq \ep \norm{a}$ so
	\begin{align*}
		\norm{\Phi(a) - a} &\leq \norm{\Phi(a) - b} + \norm{b - a} \\
				   &= \norm{v^* ( w^* \pi(a) w - \pi(b) ) v} + \norm{b - a} \\
				   &\leq \norm{w^* \pi(a) w - \pi(b)} + \norm{b - a} \\
				   &\leq \norm{w^* \pi(a) w - \pi(a)} + 2 \norm{b - a} \\
				   &= \norm{[\pi(a), w]} + 2 \norm{b - a} \\
				   &= \norm{[\pi(a), w - \I]} + 2 \norm{b-a} \\
				   &\leq 2 \norm{w - \I} \norm{a} + 2 \norm{b - a} \leq 8 \ep \norm{a}.
	\end{align*}
	Proposition \ref{prop:making homomorphisms inner} then supplies a $G'$-invariant unitary $u \in \big( \caA \cup \caB \big)''$ with $\norm{u - \I} \leq 12 \ep$ and such that $u^* a u = \Phi(a) \in \caB$ for all $a \in \caA$, where $\Phi : \caA \rightarrow \caB$ now has its domain extended to $\caA = \caA_0''$. Proposition \ref{prop:making homomorphisms inner} also guarantees that if $z \in \caB(\caH)$ satisfies $\norm{[z_{\pm}, c]_s} \leq \delta \norm{z_{\pm}} \norm{c}$ for all $c \in \caA \cup \caB$, then $\norm{u z u^* - z} \leq 40 \delta \norm{z}$.

	We finally show that if $z \in \caB(\caH)$ is such that $z \nin{\delta} \caA$ and $z \nin{\delta} \caB$, then $\norm{u z u^* - z} \leq 46 \norm{z}$. Take $a \in \caA$ with $\norm{z - a} \leq \delta \norm{x}$. Then
	\begin{equation}
		\norm{u z u^* - z} = \norm{u (x - a) u^* - (x - a) + u a u^* - a} \leq 2 \delta \norm{x} + \norm{\Phi(a) - a}.
	\end{equation}
	Now take $b \in \caB$ with $\norm{z - b} \leq \delta \norm{x}$, hence also $\norm{a-b} \leq 2 \delta \norm{x}$. Then
	\begin{align*}
		\norm{\Phi(a) - a} &\leq \norm{[\pi(a), w]} + 2 \norm{a - b} \\
				   &\leq \norm{[\pi(a), w]} + 4 \delta \norm{x} \\
				   &\leq 3 \norm{[\pi(a), y]} + 2 \norm{[\pi(a), y^*]} + 4 \delta \norm{x}
	\end{align*}
	where we used Lemma A.2 of \cite{RWW2020} in the last line. To continue, note that $p \in \pi(\caB)'$ and $q \in \pi(\caA)'$, so
	\begin{equation}
		\norm{[\pi(a), y]} = \norm{[\pi(a), qp + q^\perp p^\perp]} \leq 2 \norm{[\pi(a-b), p]} \leq 4 \norm{a-b} \leq 8 \delta \norm{x},
	\end{equation}
	and likewise for $[\pi(a), y^*]$. Then
	\begin{equation}
		\norm{u z u^* - z} \leq 46 \delta \norm{x}
	\end{equation}
	as desired.
\end{proofof}

The following lemma generalizes Lemma 2.7 of \cite{RWW2020}.

\begin{lemma} \label{lem:local errors control global errors}
	Consider two unital injective $G$-equivariant super $*$-homomorphisms $\al_1, \al_2 : \caA \rightarrow \caB$ between hyperfinite $G$-invariant super von Neumann algebras. Assume there are mutually supercommuting hyperfinite $G$-invariant super von Neumann subalgebras $\caA_1, \cdots, \caA_n \subset \caA$ that generate $\caA$ in the sense that $\caA = \big( \cup_{i = 1}^n \caA_i \big)''$. Define
	\begin{equation}
		\ep = \sum_{i = 1}^n \norm{(\al_1 - \al_2)|_{\caA_i}}.
	\end{equation}
	Then if $\ep < 1$,
	\begin{equation}
		\norm{\al_1 - \al_2} \leq 2 \sqrt{2} \ep \left( 1 + (1 - \ep^2)^{\frac{1}{2}} \right)^{-\frac{1}{2}} = 2 \sqrt{2} \ep + \caO(\ep^2). 
	\end{equation}
\end{lemma}

\begin{proof}
	Define $G$-equivariant super $*$-isomorphisms $\Phi_i$ between the images $\al_1(\caA_i)$ and $\al_{2}(\caA_i)$ by $\Phi_i(\al_1(a_i)) = \al_2(a_i)$ for all $a_i \in \caA_i$. Then
	\begin{equation}
		\norm{\Phi_i(\al_1(a_i)) - \al_1(a_i)} = \norm{\al_2(a_i) - \al_1{a_i}} \leq \norm{(\al_1 = \al_2)|_{\caA_i}} \norm{a_i}
	\end{equation}
	and $\sum_{i = 1}^n \norm{(\al_1 - \al_2)|_{\caA_i}} = \ep < 1$, so Proposition \ref{prop:making homomorphisms inner} applies to yield a $G$-invariant even unitary $u \in \big( \caA \cup \caB \big)''$ such that $\Phi_i = \Ad_u$ on all $\caA_i$ with
	\begin{equation}
		\norm{u - \I} \leq \sqrt{2} \ep \big( 1 + (1 - \ep^2)^{\frac{1}{2}} \big)^{-\frac{1}{2}}.
	\end{equation}
	But the whole of $\caA$ is generated by the $\caA_i$, so $\al_2(a) = u^* \al_1(a) u$ for all $a \in \caA$, hence
	\begin{align*}
		\norm{\al_1 - \al_2} &= \sup_{a \in \caA \, : \, \norm{a} = 1} \norm{\al_1(a) - \al_2(a)} = \sup_{a \in \caA \, : \, \norm{a} = 1} \norm{\al_1(a) - u^* \al_1(a) u} \\
	&= \sup_{a \in \caA \, : \, \norm{a} = 1} \norm{[u - \I, a]} \leq 2 \sqrt{2} \big( 1 + (1 - \ep^2)^{\frac{1}{2}} \big)^{-\frac{1}{2}},
	\end{align*}
	concluding te proof.
\end{proof}

\begin{lemma} \label{lem:G-invariance extends to vN}
	If $\caA \subset B(\caH)$ is a $G$-invariant set then $\caA''$ is also $G$-invariant.
\end{lemma}

\begin{proof}
	Let the automorphism $\rho^g$ be represented by the unitary $u_g \in B(\caH)$ (such a unitary exists by Wigner's theorem), \ie $\rho^g(a) = u_g^* a u_g$ for all $a \in B(\caH)$. Let $a \in \caA''$, then there exists a net $\{a_i\}$ in $\caA$ converging weakly to $a$. By the $G$-invariance of $\caA$, the net $\{\rho^g(a_i)\}$ is also in $\caA$, and it converges weakly to $\rho^g(a)$. Indeed, the weak convergence $a_i \rightarrow a$ means that $\langle v, a_i w \rangle \rightarrow \langle v, a w \rangle$ for any pair of vectors $v, w \in \caH$, but then
	\begin{equation}
		\langle v, \rho^g(a_i) w \rangle = \langle (u_g v), a_i (u_g w) \rangle \rightarrow \langle (u_g v), a (u_g w) \rangle = \langle v, \rho^g(a) w \rangle.
	\end{equation}
	We conclude that $\rho^g(a) \in \caA''$ and hence $\caA''$ is $G$-invariant.
\end{proof}

We state Kaplansky's density theorem, see for example Theorem 2.4.16 of \cite{BratteliRobinsonVol1}
\begin{theorem}[Kaplansky's density theorem] \label{thm:Kaplansky's density theorem}
	If $\caA \subset B(\caH)$ is a $*$-algebra then the unit ball of $\caA$ is $\sigma$-strongly* dense in the unit ball of $\caA''$.
\end{theorem}

In the statement, $\sigma$-strongly* can be replaced by any weaker topology, such as the strong or weak operator topology.

Of use above is the following equivariant version of part of Theorem 2.3 of \cite{RWW2020}:
\begin{proposition} \label{prop:equivariant conditional expectation}
	Let $\caH$ be a Hilbert space and $g \mapsto \rho^g \in U(\caH)$ a unitary representation of a finite group $G$. Let $\caA \subset B(\caH)$ be a $G$-invariant $G$-hyperfinite von Neumann algebra. Then there exists a $G$-equivariant conditional expectation $\E_{\caA'} : B(\caH) \rightarrow \caA'$ that is such that for any $x \in B(\caH)$, if $\norm{[x, a]} \leq \ep \norm{x} \norm{a}$ for all $a \in \caA$, then $\norm{\E_{\caA'}(x) - x} \leq \ep \norm{x}$.
\end{proposition}

\begin{proof}
	The von Neumann algebra $\caA$ is hyperfinite so it is the weak closure of the union of all its finite dimensional subalgebras $\{\tilde \caM \}$. These form a net, ordered by inclusion. We show that this net can be chosen $G$-invariant. Indeed, the algrebras $\tilde \caM^{(g)} = \rho^{g}(\tilde \caM)$ are a finite number of matrix subalgebras of $\caA$. Let them generate von Neumann algebras $\caM = \big( \cup_{g \in G} \tilde \caM^{(g)} \big)''$. Then all the $\caM$ are also finite dimensional matrix algebras, they are $G$-invariant (Lemma \ref{lem:G-invariance extends to vN}), and $\caA$ is the weak closure of their union.

	For each $\caM$, define a conditional expectation $\E_{\caM} : B(\caH) \rightarrow \caM'$ by
	\begin{equation}
		\E_{\caM'}(X) = \int_{U(\caM)} \; \dd u \; u^* X u.
	\end{equation}

	To see that $\E_{\caM'}(X) \in \caM'$, note that it is obvious from the definition that $\E_{\caM'}(X)$ commutes with any unitary in $\caM$. But the algebra generated by the unitaries in $\caM$ is weak dense in $\caM$, so $\E_{\caM'}(X)$ commutes with all elements of $\caM$.

	If $x$ satisfies the commutator bound $\norm{[x, y]} \leq \ep \norm{x} \norm{y}$ for all $y \in \caM$ then
	\begin{equation}
		\norm{ x - \E_{\caM'}(x) } \leq \int_{U(\caM)} \; \dd u \; \norm{ [u, x] } \leq \ep \norm{x}.
	\end{equation}

	If $y, z \in \caM'$ then obviously $\E_{\caM'}(y x z) = y \E_{\caM'}(x) z$.

	If $\caM_1 \subset \caM_2$ then $U(\caM_1) \subset U(\caM_2)$ so
	\begin{equation} \label{eq:commute with twirled unitaries II}
		[u, \E_{\caM_2}(x)] = 0 \; \text{for} \; u \in U(\caM_1) \; \text{and} \; \caM_1 \subset \caM_2.
	\end{equation}

	Let $(\caM(\al))_{\al \in I}$ be a univeral subnet. For any $x \in B(\caH)$ we have that $\norm{\E_{\caM'}(x)} \leq \norm{x}$ so the universal net $\{\E_{\caM(\al)}(x)\}_i$ is bounded and therefore weak-* convergent. Denote the limit by $\E_{\infty}(x) \in \caA'$. Clearly the map $\E_{\infty}$ is linear, completely positive, leaves every operator $x' \in caA'$ fixed, and if $y, z \in \caA'$ and $x \in B(\caH)$ then $\E_{\infty}(yxz) = y \E_{\infty}(x) z$, \ie $\E_{\infty}$ is a conditional expectation.

	The $G$-equivariance of $\E_{\infty}$ follows from that of the $\E_{\caM'}$ and the fact that the automorphisms $\rho^{g}$ are weak-* continuous.
\end{proof}

We have the following equivariant version of Stinespring's theorem
\begin{proposition}[Equivariant Stinespring theorem \cite{Paulsen1982}] \label{prop:equivariant Stinespring}
	Let $\caH$ carry a unitary representation $g \mapsto \rho^g \in U(\caH)$ of a finite group $G$ and let $\Phi : B(\caH) \rightarrow B(\caH)$ be a unital $G$-equivariant completely positive map. Then there are
	\begin{enumerate}[label=(\roman*)]
		\item a Hilbert space $\caK$,
		\item a representation $\pi$ of $B(\caH)$ on $\caK$,
		\item a unitary representation $\tilde \rho$ of $G$ on $\caK$,
		\item an isometry $V : \caH \rightarrow \caK$ such that
			\begin{enumerate}
				\item $\Phi(a) = V^* \pi(a) V$,
				\item $V(\caH)$ is an invariant subspace of $\tilde \rho$ and $\rho^g = V^* \tilde \rho^g V$,
				\item $\pi$ is equivariant in the sense that $\pi((\rho^g)^* a \rho^g) = \tilde (\rho^g)^* \pi(a) \tilde \rho^g$ for all $a \in B(\caH)$ and all $g \in G$.
			\end{enumerate}
	\end{enumerate}
	Moreover, if $\caH$ is separable, then $\caK$ can also be chosen separable.
\end{proposition}


\printbibliography

@article{GNVW2012,
  title={Index theory of one dimensional quantum walks and cellular automata},
  author={Gross, David and Nesme, Vincent and Vogts, Holger and Werner, Reinhard F},
  journal={Communications in Mathematical Physics},
  volume={310},
  number={2},
  pages={419--454},
  year={2012},
  publisher={Springer}
}

@article{CPSV2017,
  title={Matrix product unitaries: structure, symmetries, and topological invariants},
  author={Cirac, J Ignacio and Perez-Garcia, David and Schuch, Norbert and Verstraete, Frank},
  journal={Journal of Statistical Mechanics: Theory and Experiment},
  volume={2017},
  number={8},
  pages={083105},
  year={2017},
  publisher={IOP Publishing}
}

@article{FPPV2019,
  title={Interacting invariants for Floquet phases of fermions in two dimensions},
  author={Fidkowski, Lukasz and Po, Hoi Chun and Potter, Andrew C and Vishwanath, Ashvin},
  journal={Physical Review B},
  volume={99},
  number={8},
  pages={085115},
  year={2019},
  publisher={APS}
}

@book{Vara2004,
  title={Supersymmetry for Mathematicians: An Introduction: An Introduction},
  author={Varadarajan, Veeravalli S},
  volume={11},
  year={2004},
  publisher={American Mathematical Soc.}
}

@inproceedings{BO2021,
  title={The classification of symmetry protected topological phases of one-dimensional fermion systems},
  author={Bourne, Chris and Ogata, Yoshiko},
  booktitle={Forum of Mathematics, Sigma},
  volume={9},
  year={2021},
  organization={Cambridge University Press}
}

@article{RWW2020,
  title={A converse to Lieb-Robinson bounds in one dimension using index theory},
  author={Ranard, Daniel and Walter, Michael and Witteveen, Freek},
  journal={arXiv preprint arXiv:2012.00741},
  year={2020}
}

@book{BratteliRobinsonVol1,
  title={Operator Algebras and Quantum Statistical Mechanics: Volume 1: C*-and W*-Algebras. Symmetry Groups. Decomposition of States},
  author={Bratteli, Ola and Robinson, Derek William},
  year={2012},
  publisher={Springer Science \& Business Media}
}

@article{Araki2004,
  title={Conditional expectations relative to a product state and the corresponding standard potentials},
  author={Araki, Huzihiro},
  journal={Communications in mathematical physics},
  volume={246},
  number={1},
  pages={113--132},
  year={2004},
  publisher={Springer}
}

@article{GSSC2020,
  title={Classification of matrix-product unitaries with symmetries},
  author={Gong, Zongping and S{\"u}nderhauf, Christoph and Schuch, Norbert and Cirac, J Ignacio},
  journal={Physical review letters},
  volume={124},
  number={10},
  pages={100402},
  year={2020},
  publisher={APS}
}

@book{Serre1977,
  title={Linear representations of finite groups},
  author={Serre, Jean-Pierre},
  volume={42},
  year={1977},
  publisher={Springer}
}

@article{Paulsen1982,
  title={A covariant version of Ext.},
  author={Paulsen, Vern},
  journal={Michigan Mathematical Journal},
  volume={29},
  number={2},
  pages={131--142},
  year={1982},
  publisher={University of Michigan, Department of Mathematics}
}

@book{Blackadar2006,
  title={Operator algebras: theory of C*-algebras and von Neumann algebras},
  author={Blackadar, Bruce},
  volume={122},
  year={2006},
  publisher={Springer Science \& Business Media}
}

@article{Arveson1969,
  title={Subalgebras of {C}*-algebras},
  author={Arveson, William B},
  journal={Acta Mathematica},
  volume={123},
  number={1},
  pages={141--224},
  year={1969},
  publisher={Springer}
}

@book{Arveson2012,
  title={An invitation to C*-algebras},
  author={Arveson, William},
  volume={39},
  year={2012},
  publisher={Springer Science \& Business Media}
}

@article{Christensen1977b,
  title={Perturbations of operator algebras {II}},
  author={Christensen, Erik},
  journal={Indiana University Mathematics Journal},
  volume={26},
  number={5},
  pages={891--904},
  year={1977},
  publisher={JSTOR}
}

@article{Christensen1975,
  title={Perturbations of type {I} von {N}eumann algebras},
  author={Christensen, Erik},
  journal={Journal of the London Mathematical Society},
  volume={2},
  number={3},
  pages={395--405},
  year={1975},
  publisher={Oxford Academic}
}

@incollection{NSW2013,
  title={Local approximation of observables and commutator bounds},
  author={Nachtergaele, Bruno and Scholz, Volkher B and Werner, Reinhard F},
  booktitle={Operator methods in mathematical physics},
  pages={143--149},
  year={2013},
  publisher={Springer}
}

@article{Tomiyama1959,
  title={On the product projection of norm one in the direct product of operator algebras},
  author={Tomiyama, Jun},
  journal={Tohoku Mathematical Journal, Second Series},
  volume={11},
  number={2},
  pages={305--313},
  year={1959},
  publisher={Mathematical Institute, Tohoku University}
}

@book{EvansLewis1977,
  title={Dilations of irreversible evolutions in algebraic quantum theory},
  author={Evans, David Emrys and Lewis, John T},
  number={24},
  year={1977},
  publisher={Dublin Institute for Advanced Studies}
}

@article{KSY2020,
  title={A classification of invertible phases of bosonic quantum lattice systems in one dimension},
  author={Kapustin, Anton and Sopenko, Nikita and Yang, Bowen},
  journal={Journal of Mathematical Physics},
  volume={62},
  number={8},
  pages={081901},
  year={2021},
  publisher={AIP Publishing LLC}
}

@article{CGW2011,
  title={Classification of gapped symmetric phases in one-dimensional spin systems},
  author={Chen, Xie and Gu, Zheng-Cheng and Wen, Xiao-Gang},
  journal={Physical review b},
  volume={83},
  number={3},
  pages={035107},
  year={2011},
  publisher={APS}
}

@article{Ogata2021,
  title={A classification of pure states on quantum spin chains satisfying the split property with on-site finite group symmetries},
  author={Ogata, Yoshiko},
  journal={Transactions of the American Mathematical Society, Series B},
  volume={8},
  number={2},
  pages={39--65},
  year={2021}
}

@article{FMPPV2016,
  title={Chiral floquet phases of many-body localized bosons},
  author={Po, Hoi Chun and Fidkowski, Lukasz and Morimoto, Takahiro and Potter, Andrew C and Vishwanath, Ashvin},
  journal={Physical Review X},
  volume={6},
  number={4},
  pages={041070},
  year={2016},
  publisher={APS}
}

@article{ElseNayak2016,
  title={Classification of topological phases in periodically driven interacting systems},
  author={Else, Dominic V and Nayak, Chetan},
  journal={Physical Review B},
  volume={93},
  number={20},
  pages={201103},
  year={2016},
  publisher={APS}
}

@article{Hastings2013,
  title={Classifying quantum phases with the Kirby torus trick},
  author={Hastings, Matthew B},
  journal={Physical Review B},
  volume={88},
  number={16},
  pages={165114},
  year={2013},
  publisher={APS}
}

\end{document}